\documentclass{article} 
\usepackage{iclr2023_conference,times}

\usepackage{hyperref}
\usepackage{url}
\iclrfinalcopy

\title{Statistical Inference for Fisher Market Equilibrium}

\author{Luofeng Liao, Yuan Gao, Christian Kroer \\
Department of Industrial Engineering and Operations Research\\
Columbia University\\
\texttt{\{ll3530,yg2541,christian.kroer\}@columbia}
}

\usepackage[OT1]{fontenc} 
\iclrfinalcopy 
\usepackage{mathtools}       
\usepackage{amsmath}
\usepackage{amsthm}
\usepackage[linesnumbered,ruled,lined,noend]{algorithm2e} 
\usepackage{accents}
\usepackage{amsmath,amssymb,amsthm}

\usepackage[shortlabels]{enumitem}
\setlist[itemize]{topsep=0pt,}

\hypersetup{
colorlinks   = true, 
urlcolor     = blue, 
linkcolor    = blue, 
citecolor   = blue 
}

\usepackage{dsfont} 
\usepackage{xcolor} 
\usepackage[capitalise]{cleveref}
\usepackage{esvect}

\usepackage{tabularx}
\usepackage{booktabs}
\usepackage[labelfont=bf,format=plain,justification=raggedright,singlelinecheck=false]{caption}

\newcommand{\ubar}[1]{\underaccent{\bar}{#1}}

\def\##1\#{\begin{align}#1\end{align}}
\def\$#1\${\begin{align*}#1\end{align*}}

\def\tr{\mathop{\text{tr}}\kern.2ex}

\def \gam {{ \gamma }}
\def \indi {{ \mathds{1} }}

\def \calME {{ \mathcal{ME}}}
\def \calMEgam {{ \mathcal{ME}^\gamma}}
\def \QME {{ \mathcal{QME}}}
\def \QMEgam {{ \mathcal{QME}^\gam}}

\def\R{{\mathbb R}}
\def\Rp{{\mathbb R _+}}
\def\Rn{{\mathbb R^n}}
\def\Rnp{{\mathbb R^n_+}}

\def\Rnmpp{{\mathbb R^{n-1}_{++}}}

\def\Rnpp{{\mathbb R^n_{++}}}
\def\Q{{\mathbb Q}}

\def\P{{\mathbb P}}
\def\E{{\mathbb E}}

\def\cB{{\mathcal{B}}}

\def\cH{{\mathcal{H}}}

\def\cN{{\mathcal{N}}}

\def\must{\mu^*}
\def\t{\theta}

\def \sfs {{\mathsf{s}}}

\def\rg{{\rangle}}
\def\lg{{\langle}}

\newtheorem{example}{Example}
\newtheorem{remark}{Remark}
\newtheorem{claim}{Claim}
\newtheorem{fact}{Fact}
\newtheorem{lemma}{Lemma}
\newtheorem{defn}{Definition}
\newtheorem{corollary}{Corollary}

\newlist{enumconditions}{enumerate}{1} 
\setlist[enumconditions]{label = \thelemma.\alph*}
\crefname{enumconditionsi}{Condition}{Conditions}

\newlist{enumlmresult}{enumerate}{1} 
\setlist[enumlmresult]{label = \thelemma.\arabic*}
\crefname{enumlmresulti}{Part}{Parts}

\newlist{enumthmresult}{enumerate}{1} 
\setlist[enumthmresult]{label = \thetheorem.\arabic*}
\crefname{enumthmresulti}{Part}{Parts}

\newlist{enumassumption}{enumerate}{1} 
\setlist[enumassumption]{label = \theAssumption.\arabic*}
\crefname{enumassumptioni}{Condition}{Conditions}

\newtheorem{theorem}{Theorem}
\renewcommand*{\thetheorem}{\arabic{theorem}}

\def \I {{\mathrm{{I}}}}
\def \II {{\mathrm{{II}}}}

\def \LNSW {{\small \mathrm{{NSW}}}}
\def \EG {{\scriptscriptstyle \mathrm{{EG}}}}
\def \calME {{\mathcal{{ME}}}}

\def \ME {{\scriptscriptstyle \mathrm{{ME}}}}
\def \QEG {{\scriptscriptstyle \mathrm{{QEG}}}}
\def \REV {{\small \mathrm{{REV}}}}
\def \NSW {{ \scriptscriptstyle \mathrm{{N}}}}

\def \cBst {{ \mathcal{B}^*}}
\def \cBgam {{ \mathcal{B}^\gamma}}
\def \cBgamS {{ \mathcal{B}^\gamma_S}}
\def \cBgamC {{ \mathcal{B}^\gamma_C}}
\def \cBgamX {{ \mathcal{B}^\gamma_X}}

\def \cBstS {{ \mathcal{B}^*_S}}
\def \cBstC {{ \mathcal{B}^*_C}}
\def \cBstX {{ \mathcal{B}^*_X}}

\def \Hst {{H^*}}
\def \hatOmesqi {{ \hat { \Omega}_i^2 }}

\def \ust {{ u^* }}
\def \ugam { u^\gamma }
\def \ugami { u^\gamma_i }
\def \ugamtaui { u^{\gamma,\tau}_i }
\def \usti {{ u^*_i }}

\def \ubarbetai {{\ubar\beta_i}}
\def \ubarbetaQi {{\ubar\beta_{Q,i}}}
\def \vithetau {{ v_i(\theta^\tau) }}
\def \vithe {{ v_i(\theta) }}

\def \vbar {{ \bar v }}
\def \vbarsq {{ \bar v ^2 }}

\def \vbarit {{ \bar v_i^t }}

\def \fbar {{\bar f}}

\def \var {{ \operatorname {var} }}

\def \nui {{ \nu_i}}
\def \nubar {{ \bar\nu}}

\def \sumtau {{\sum_{\tau=1}^{t}}}

\def \sumiton {{ \sum_{i=1}^n }}
\def \sumi {{ \sum_{i} }}

\def \ptau {{ p ^ \tau }}
\def \pgamtau {{ p ^ {\gamma, \tau} }}
\def \pgam {{ p ^ {\gamma} }}

\def \xtaui {{ x^\tau_i }}
\def \xgam {{ x^{\gamma} }}
\def \xgamtaui {{ x^{\gamma,\tau}_i }}

\def \ttinf {{t\to \infty}}
\def \thetau {{ \theta^\tau }}

\def \eps{{ \epsilon }}

\def \sighatsq {{ \hat \sigma ^2}}

\def \pst {{p^*}}
\def \betabar {{\bar \beta}}
\def \betabarQ {{\bar \beta _Q}}

\def \betadia {{\beta^\diamond}}

\def \betast {{\beta^*}}

\def \betasti {{\beta^*_i}}
\def \betastX {{\beta^*_X}}

\def \betai {{\beta_i}}
\def \betagam {\beta^\gamma}

\def \betagami {\beta^\gamma_i}
\def \cov {{\operatorname{Cov}}}
\def \xst {x^*}
\def \xsti {{x^*_i}}

\def \inv {^{-1}}
\def \sq {^{2}}
\def \tp {^{\top}}
\def \st {^{*}}

\def \dom {{ \operatorname*{Dom}\,}}

\def \Diag {{ \operatorname*{Diag}}}
\newcommand{\defeq}{\vcentcolon=}

\newcommand*\diff{\mathop{}\!\mathrm{d}}

\def \toas {{ \,\overset{\mathrm{{a.s.\;}}}{\longrightarrow} \,}}
\def \toprob {{ \,\overset{\mathrm{p}}{\to}\, }}
\def \toepi {{ \,\overset{\mathrm{epi}}{\longrightarrow}\, }}
\def \tod {{ \,\overset{\mathrm{d}}{\to}\, }}

\DeclareMathOperator*{\argmax}{arg\,max}
\DeclareMathOperator*{\argmin}{arg\,min}
\DeclareMathOperator*{\esssup}{ess\,sup}
\begin{document}
\maketitle
\begin{abstract}
    Statistical inference under market equilibrium effects has attracted increasing attention recently. 
    In this paper we focus on the specific case of linear Fisher markets. They have been widely use in fair resource allocation of food/blood donations and budget management in large-scale Internet ad auctions. In resource allocation, it is crucial to quantify the variability of the resource received by the agents (such as blood banks and food banks) in addition to fairness and efficiency properties of the systems. For ad auction markets, it is important to establish statistical properties of the platform's revenues in addition to their expected values. To this end, we propose a statistical framework based on the concept of infinite-dimensional Fisher markets. In our framework, we observe a market formed by a finite number of items sampled from an underlying distribution (the ``observed market'') and aim to infer several important equilibrium quantities of the underlying long-run market. These equilibrium quantities include individual utilities, social welfare, and pacing multipliers. Through the lens of sample average approximation (SSA), we derive a collection of statistical results and show that the observed market provides useful statistical information of the long-run market. In other words, the equilibrium quantities of the observed market converge to the true ones of the long-run market with strong statistical guarantees.
    These include consistency, finite sample bounds, asymptotics, and confidence. As an extension we discuss revenue inference in quasilinear Fisher markets.
 \end{abstract}
\section{Introduction}

In a Fisher market there is a set of $n$ buyers that are interested in buying goods from a distinct seller. 
A market equilibrium (ME) is then a set of prices for the goods, along with a corresponding allocation, such that demand equals supply.

One important application of market equilibrium (ME) is fair allocation using the competitive equilibrium
from equal incomes (CEEI) mechanism~\citep{varian1974equity}. 
In CEEI, each individual is given an endowment
of faux currency and reports her valuations for items; then, a market equilibrium is computed, and
the items are allocated accordingly. The resulting allocation has many desirable properties such as Pareto optimality, envy-freeness and proportionality. For example, Fisher market equilibrium has been used for 
fair work allocation, impressions allocation in certain recommender systems, course seat allocation and scarce computing resources allocation; see \cref{sec:related_works} for an extensive overview.

Despite numerous algorithmic results available for computing Fisher market equilibria, to the best of our knowledge, no statistical results were available for quantifying the randomness of market equilibrium. 
Given that CEEI is a fair and efficient mechanism, such statistical results are useful for quantifying variability in CEEI-based resource allocation.
For example, for systems that assign blood donation to hospitals and blood banks~\citep{mcelfresh2020matching}, 
or donated food to charities in different neighborhoods~\citep{aleksandrov2015online,sinclair2021fairness}, 
it is crucial to quantify the variability of the amount of resources (blood or food donation) received by the participants (hospitals or charities) of these systems as well as the variability of fairness and efficiency metrics of interest in the long run.
Making statistical statements about these metrics is crucial for both evaluating and improving these systems.

In addition to fair resource allocation, statistical results for Fisher markets can also be used in revenue inference in Internet ad auction markets. 
While much of the existing literature uses expected revenue as performance metrics, 
statistical inference on revenue is challenging due to the complex interaction among bidders under coupled supply constraints and common price signals.
As shown by~\citet{conitzer2022pacing}, in budget management through repeated first-price auctions with pacing, the optimal pacing multipliers correspond to the  ``prices-per-utility'' of buyers in a quasilinear Fisher market at equilibrium.
Given the close connection between various solution concepts in Fisher market models and first-price auctions, a statistical framework enables us to quantify the variability in long-run revenue of an advertising platform. 
Furthermore, a statistical framework would also help answer other statistical questions such as the study of counterfactuals and theoretical guarantees for A/B testing in Internet ad auction markets.

For a detailed survey on related work in the areas of statistical inference, applications of Fisher market models, and equilibrium computation algorithms, see~\cref{sec:related_works}. 

Our contributions are as follows.

\textbf{A statistical Fisher market model}.
We formulate a statistical estimation problem for Fisher markets based on the continuous-item model of \citet{gao2022infinite}.
We show that when a finite set of goods are sampled from the continuous model, the observed ME is a good approximation of the long-run market.
In particular, we develop consistency results, finite-sample bounds, central limit theorems, and asymptotically valid confidence interval for various quantities of interests, such as individual utility, Nash social welfare, pacing multipliers, and revenue (for quasilinear Fisher markets).

\textbf{Technical challenges}. In developing central limit theorems for pacing multipliers and utilities in Fisher markets~(\cref{thm:clt_beta_u}), we note that the dual objective is potentially not twice differentiable. This is a required condition, which is common in the sample average approximation or M-estimation literature. We discover three types of market where such differentiability is guaranteed.
Moreover, the sample function is not differentiable, which requires us to verify a set of stochastic differentiability conditions in the proofs for central limit theorems. 
Finally, we achieve a fast statistical rate of the empirical pacing multiplier to the population pacing multiplier measured in the dual objective by exploiting the local strong convexity of the sample function.

\begin{figure}[t]
    \centering
    \includegraphics[scale=.6]{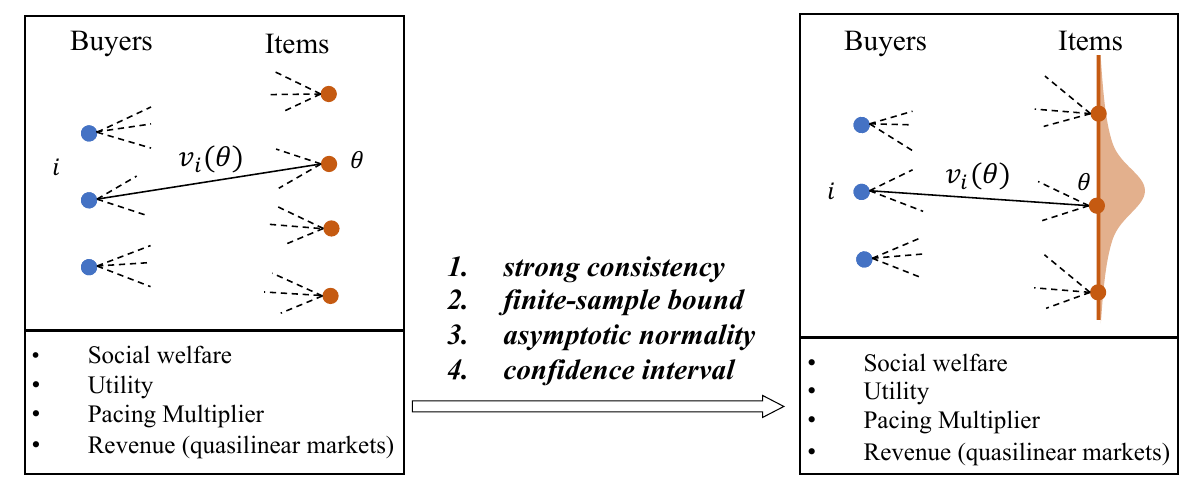}
    \caption{
        \small
        Our contributions. 
    Left panel: a Fisher market with a finite number of divisible
    items. Buyer $i$ has value $v_i(\theta)$ for item $\theta$. The goal is to allocate items 
    so that equilibrium conditions are met (\cref{def:observed_market}).
    Right panel: an infinite-dimensional Fisher market with a continuum of items. 
    Middle arrow: this paper provides various forms of statistical guarantees to characterize the convergence of observed finite Fisher market (left) to 
    the long-run market (right) when the items are drawn from a distribution corresponding to the supply function in the long-run market.}
    \vspace{-4mm}
\end{figure}

\noindent\textbf{Notation.} 
For a sequence of events $A_n$ we define the set limit by $ \liminf _{n \rightarrow \infty} A_{n}=\bigcup_{n \geq 1} \bigcap_{j \geq n} A_{j} =\{ A_t \text{ eventually}\}$ and $\limsup_{n \to \infty } A_{n}=\bigcap_{n \geq 1} \bigcup_{j \geq n} A_{j} = \{ A_t \text{ i.o.}\}$.
For vector $a,b\in \Rn$ we let $a{\cdot} b$ be the elementwise product and let $a_{-i} \in \R^{n-1}$ be the vector $a$ with the $i$-th entry removed. For vector $a$ we let $[a_{(1)},\dots, a_{(n)}]$ denote the sorted entries of $a$ from greatest to least.
Let $[n] = \{1,\dots, n\}$. We use $1_t$ to denote the vector of ones of length $t$ 
and $e_j$ to denote the vector with one in the $j$-th entry and zeros in the others.
For a sequence of random variables $\{X_n\}$, we say $X_n = O_p(1)$ if for any $\epsilon > $ there exists a finite $M_\epsilon$ and a finite $N_\epsilon$ such that $\P(|X_n| > M_\epsilon) < \epsilon$ for all $n\geq N_\epsilon$. We say $X_n = O_p(a_n)$ if $X_n/a_n = O_p(1)$.
We use subscript for indexing buyers and superscript for items. If a function $f$ is twice continuously differentiable at a point $x$, we say $f$ is $C^2$ at $x$.
\section{Problem Setup}

\subsection{The Estimands}\label{sec:estimands}
Following~\citet{gao2022infinite}, we consider a Fisher market with $n$ buyers (individuals), each having a budget $b_i>0$ and a (possibly continuous) set of items $\Theta$. 
We let $L^p$ (and $L^p_+$, resp.) denote the set of (nonnegative, resp.) $L^p$ functions on $\Theta$ for any $p\in [1, \infty]$ (including $p=\infty$). 
The item \emph{supplies} are given by a function  $ s \in L^\infty_+$, i.e., item $\theta\in \Theta$ has supply $s(\theta)$. 
The \emph{valuation} for buyer $i$ is a function $v_i \in L^1_+$, i.e., buyer $i$ has valuation $v_i(\theta)$ for item $\theta\in \Theta$. 
For buyer $i$, an \emph{allocation} of items $x_i \in L^\infty_+$ gives a utility of 
\[ u_i(x_i) := \langle v_i, x_i \rangle := \int_\Theta v_i(\theta) x_i(\theta) \diff \mu(\theta),\] 
where the angle brackets are based on the notation of applying a bounded linear functional $x_i$ to a vector $v_i$ in the Banach space $L^1$ and the integral is the usual Lebesgue integral. 
We will use $x\in (L^\infty_+)^n$ to denote the aggregate allocation of items to all buyers, i.e., the concatenation of all buyers' allocations. 
The \emph{prices} of items are modeled as $p\in L^1_+$. The price of item $\theta\in \Theta$ is $p(\theta)$.
Without loss of generality, we assume  a unit total supply $\int_\Theta s \diff \mu = 1$.
We let $S(A) \defeq \int_A s(\theta)\diff \mu(\theta)$ be the probability measure induced by the supply $s$.

\begin{defn}[The long-run market equilibrium]
The {market equilibrium (ME)} $\calME (b,v,s)$ is an allocation-utility-price tuple $(x^*, u^*, p^*) \in (L^\infty_+)^n \times \Rnp \times L^1_+$ such that the following holds.
(i) Supply feasibility and market clearance: $\sum_i x^*_i \leq s$ and $\langle p^*, s - \sum_i x^*_i \rangle = 0$. 
(ii)
Buyer optimality: $x^*_i \in D_i (p^*)$ and $u^* = \lg v_i, x_i \rg$ for all $i$ where 
        the {demand} $D_i$ of buyer $i$ is its set of utility-maximizing allocations given the prices and budget:
        \[ D_i (p) := \argmax \{ \langle v_i, x_i \rangle : x_i \in L^\infty_+,\, \langle p, x_i\rangle \leq b_i \}.\] 

\label{def:long-run_market}
\end{defn}

Linear Fisher market equilibrium can be characterized by convex programs. 
We state the following result from~\citet{gao2022infinite} which establishes existence and uniqueness of market equilibrium, and more importantly the convex program formulation of the equilibrium. We define the Eisenberg-Gale (EG) convex programs which as we will see are dual to each other.

\vspace{-.7cm}
\begin{align}
\label{eq:pop_eg}
\max_{x \in L^\infty_+(\Theta), u\geq 0}
\bigg\{ 
    \sumiton b_i \log  (u_i)
\;\bigg|\;  
u_i   \leq  \big\langle v_i, {x}_{i} \big\rangle 
  \;\; \forall i \in [n] 
,\;\;
    \sumiton {x}_{i} \leq s
\bigg\} \;,
\tag{\small{P-EG}}
\end{align}
\vspace{-.5cm}
\begin{align}
    \min_{\beta > 0}
    \bigg\{ H(\beta) = \int_\Theta  \Big(\max_{i\in[n]} \beta_i v_i(\theta) \Big) S(\diff \theta)  
    -  \sumiton b_i \log \beta_i 
    \bigg\}
    \;.
    \tag{\small P-DEG}
    \label{eq:pop_deg}
\end{align}
Concretely, the optimal primal variables in \cref{eq:pop_eg} corresponds to the set of equilibrium allocations $\xst$ and the unique equilibrium utilities $\ust$, and the unique optimal dual variable $\betast$ of \cref{eq:pop_deg} relates to the equilibrium utilities and prices  through 
\begin{align*}
    \betasti = b_i / \usti \;, \quad \pst (\theta) = \max_i \betasti \vithe \;.
\end{align*}
We call $\betast$ the \emph{pacing multiplier}.
Note equilibrium allocations might not be unique but equilibrium utilities and prices are unique. 
Given the above equivalence result, we use $(\xst,\ust)$ to denote both the equilibrium and the optimal variables. Another feature of linear Fisher market is full budget extraction: $\int \pst \diff S = \sumiton b_i$; we discuss quasilinear model in \cref{sec:quasilinear}.

We formally state the first-order conditions of infinite-dimensional EG programs and 
its relation to first-price auctions in \cref{fact:pop_eg} in appendix.
Also, we remark that there are two ways to specify the valuation component in this model: the functional form of $v_i(\cdot)$, or the distribution of values $v:\Theta \to \Rnp$ when view as a random vector. More on this in \cref{sec:opt_EG_programs}.

We are interested in estimating the following quantities of the long-run market equilibrium. 
(1)~\textbf{Individual utilities} at equilibrium, $u^*_i$.
    It directly reflects how much a buyer benefits from the market.
(2)~\textbf{Pacing multipliers} $\betasti =  b_i/ \usti$. 
    From an optimization perspective, it is simply the optimal dual variable of the EG program~\cref{eq:pop_eg}.
    However, its role deserves more explanation.
    Pacing multiplier has a two-fold interpretation.
    First, through the equation $\betasti = b_i/\usti$ it measures the price-per-utility that a buyer receives. 
    Second, through the equation $\pst(\theta) = \max_i \betasti \vithe$, $\beta$ can also be interpreted as the \emph{pacing policy\footnote{In the online budget management literature, pacing means buyers produce bids for items via multiplying his value by a constant.}} employed by the buyers in first-price auctions. 
    In our context,
    buyer $i$ produces a bid for item $\theta$ by multiplying the value by $\beta_i$, then the item is allocated via a first-price auction. 
    This connection is made precise in~\citet{conitzer2022pacing} from a game-theoretic point of view.
    The pacing multiplier $\beta$ serves as the bridge between Fisher market equilibrium and first price pacing equilibria and has important usage in online ad auction for characterizing the strategic behavior of advertisers.
(3)~The (logarithm of) \textbf{Nash social welfare} (NSW) at equilibrium
    \begin{align*}
        \LNSW^* \defeq \sumiton b_i \log \usti.
    \end{align*}
    NSW measures total utility of the buyers in a way that is more fair than the usual social welfare, which measures the sum of buyer utilities, because NSW incentivizes more balancing of buyer utilities.
(4)~\textbf{The revenue}. Linear Fisher market extract the budges fully, i.e., 
    $ \int \pst \diff S = \sumi b_i$ in the long-run market and $\sumtau \pgamtau = \sumi b_i$ in the observed market (see \cref{sec:opt_EG_programs}),    
    and therefore there is nothing to infer about revenue in this case. 
    However, in the quasilinear utility model where buyer's utility function is $u_i(x) = \lg x - p, v_i\rg$, buyers have the incentive to retain money and therefore one needs to study the statistical properties of revenues. This is discussed in \cref{sec:quasilinear}.

As we will see later, their counterparts in the observed market (to be introduced next) will be good estimators for these quantities.

\subsection{The Data}\label{sec:data}

Assume we are able to observe a market formed by a finite number of items.
We let $\gam = \{ \theta_1,\dots, \theta_t\} \subset \Theta^t$ be a set of items sampled i.i.d.\ from the supply distribution $S$.
We let $v_i(\gamma) = \big(v_i (\theta^1), \dots, v_i(\theta^t)\big)$ denote the valuation for agent $i$ of items in the set $\gamma$. 
For agent $i$, let $x_i = (x^1_i,\dots, x^t_i) \in \R^t$ denote the fraction of items given to agent $i$. 
With this notation, the total utility of agent $i$ is 
$\langle x_i , v_i(\gamma) \rangle$.

Similar to the long-run market, we assume the observed market is at equilibrium, which we now define.

\begin{defn}[Observed Market Equilibrium] \label{def:observed_market}
    The market equilibrium $\calMEgam(b,v, \sfs)$ given the item set $\gam$ and the supply vector $\sfs\in \R^t_+$ is an allocation-utitlity-price tuple $({x}^\gam, u^\gam, p^\gam) \in (\R^t_+)^n \times \Rnp \times \R^t_+$ such that the following holds.
(i) Supply feasibility and market clearance: $\sumiton {x}^\gam_i \leq \sfs$ and  $\langle p^\gam, {1}_t - \sumiton {x}^\gam_i \rangle = 0$.
(ii) Buyer optimality: ${x}^\gam_i \in D_i (p^\gam)$ and $u^\gam_i = \lg v_i(\gam),x_i\rg $ for all $i$, where (overloading notations) $$D_i (p) := \argmax \{   \langle {{v}}_i(\gam), {x}_i \rangle : {x}_i \geq 0,\,  \langle p,  {x}_i\rangle \leq b_i \}
        $$ is the demand set given the prices and the buyer's budget.
\label{defn:observed_market}
\end{defn}

Assume we have access to $(\xgam,\ugam,\pgam)$ along with the bid vector $b$, where $(\xgam,\ugam,\pgam)=\calMEgam (b, v, \frac1t 1_t)$ is the market equilibrium
(we explain the scaling of $1/t$ in \cref{sec:opt_EG_programs}).
Note the budget vector $b$ and value functions $v=\{v_i(\cdot)\}_i$ are the same as those in the long-run ME. We emphasize two high-lights in this model of observation.

\textbf{Dependency on realized values $\{\vithetau\}_{i,\tau}$ and value functions $v_i(\cdot)$.} In contrast to several online methods for computing long-run market equilibrium with convex optimization methods~\citep{gao2021online,liao2022dualaveraging,azar2016allocate}
where one needs knowledge of the values of items from buyers to produce an estimate of $\betast$, here we only need to observe the equilibrium allocation, utilities and prices.

\textbf{No convex program solving.} The quantities observed are natural estimators of their counterparts in the long-run market, and so we do not need to perform iterative updates or solve optimization problems. 
One interpretation of this is that the actual computation is done when equilibrium is reached via the utility maximizing property of buyers; the work of computation has thus implicitly been delegated to the buyers.

For finite-dimensional Fisher market, it is well-known that 
the observed market equilibrium $\calMEgam(b,v,\frac1t 1_t)$ can be captured by the following sample EG programs.
\begin{align}    \label{eq:sample_eg}
    \max_{x\geq0, u\geq 0}
    \bigg\{ 
    \sumiton    b_i  \log  (u_i)
    \;\bigg|\;  
    u_i   \leq \big\langle v_i(\gamma), {x}_{i} \big\rangle 
      \;\; \forall i \in [n] 
    \;,\;\;
        \sumiton {x}_{i}^\tau \leq  \tfrac1t  1_t \;\; \forall \tau \in [t]
    \bigg\} \;,
    \tag{\small S-EG}
\end{align}
\vspace{-.5cm}
\begin{align}
    \min_{\beta > 0}
    \bigg\{ H_t(\beta) = \frac1t \sumtau \max_{i\in[n]} \beta_i v_i(\theta^\tau)  -   \sumiton b_i \log \beta_i 
    \bigg\} \;.
    \tag{\small S-DEG}
    \label{eq:sample_deg}
\end{align}
We list the KKT conditions in~\cref{sec:opt_EG_programs}. 
Completely parallel to the long-run market, optimal solutions to \cref{eq:sample_eg} correspond to the equilibrium allocations and utilities, and the optimal variable $\betagam$ to \cref{eq:sample_deg} relates to equilibrium prices and utilities through $\ugami = b_i/\betagami$ and $\pgamtau = \max_i \betagami \vithetau$.
By the equivalence between market equilibrium and EG programs, we use $u^\gam$ and $x^\gam$ to denote the equilibrium and the optimal variables. 
Let 
\begin{align*}
    \LNSW^\gam \defeq \sumiton b_i \log \ugami 
    \;.
\end{align*}
All budgets in the observed market is extracted, i.e., $\sumtau \pgamtau = \sumiton b_i$. 

\subsection{Dual Programs: Bridging Data and the Estimands}

Given the convex program characterization, a natural idea is to study the concentration behavior of observed market equilibria through these convex programs. 
Such an approach is closely related to \emph{M-estimation} in the statistics literature (see, e.g.,~\citet{van2000asymptotic,newey1994large}) and sample average approximation (SSA) in the stochastic programming literature (see, e.g.,~\citet[Chapter 5]{shapiro2021lectures},~\citet{shapiro2003monte} and~\citet{Kim2015}).
However, for the primal program \cref{eq:sample_eg}, the dimension of the optimization variables is changing as the market grows, and therefore it is harder to use existing tools. 
On the other hand, the dual programs \cref{eq:sample_deg,eq:pop_deg} are defined in a fixed dimension, and moreover the constraint set is also fixed.

Define the sample function $F = f + \Psi$, where  $f(\beta,\theta)=\max_i\{ v_i(\theta) \betai \}$, 
and $\Psi(\beta) = -\sumiton b_i \log \beta_i$; 
the function $f$ is the source of non-smoothness, while $\Psi$ provides local strong convexity. 
Then
the sample dual objective in \cref{eq:sample_deg} can be expressed as $H_t(\beta) = \frac1t \sumtau F(\beta,\thetau)$
and 
the population dual objective \cref{eq:pop_deg} can be compactly written as $H = \E[F(\beta,\theta)] = \fbar + \Psi$ 
where $\fbar(\beta) = \E[f(\beta,\theta)]$ is the expectation of $f$. 
We call $\beta_i v_i(\theta)$ \emph{the bid of buyer $i$ for item $\theta$}.
The rest of the paper is devoted to studying concentration of the convex programs in the sense that as $t$ grows
\begin{align*}
    \min_{\beta > 0}  H_t(\beta)
    \quad`` \!\! \implies \!\!" \quad 
    \min_{\beta > 0}  H(\beta) 
    \; .
\end{align*}
The local strong convexity of the dual objective motivates us to do the analysis work in the neighborhood of the optimal solution $\betast$.
In particular, the function $x \mapsto -\log x$ is not strongly convex on the positive reals, but it is on any compact subset. 
By working on a compact subset, we can exploit 
strong convexity of the dual objective and obtain better theoretical results.
Recall that $\ubarbetai \leq \betasti \leq \bar \beta$ where $\ubarbetai =  b_i/ \int v_i \diff S$ 
and $\bar \beta = \sumiton b_i /\min_i \int v_i \diff S$.
Define the compact set 
$
    C \defeq \prod_{i=1}^n \big[\ubarbetai/2, 2\bar \beta\big]  \subset \R^n,
$
which must be a neighborhood of $\betast$.
Moreover, for large-enough $t$ we further have $\betagam \in C$ with high probability.
\begin{lemma} 
    \label{lm:value_concentration} 
    Define the event $A_t = \{ \betagam \in C   \}$
    .
    (i) If $t \geq 2\vbarsq {\log(2n/\eta)}$, then $\P(A_t) \geq \P(\frac12 \leq \frac1t \sumtau \vithetau \leq 2,\forall i ) \geq 1- \eta$.
    \label{it:lm:value_concentration:1}
    (ii) It holds $\P( A_t \text{ eventually}) = 1$.
    \label{it:lm:value_concentration:2}
    Proof in \cref{sec:tech_lemmas}.
\end{lemma}

We will also be interested in concentration of approximate market equilibria. 
For any utility vector $u$ achieved by a feasible allocation, we define $\beta_u = [\frac{b_1}{u_1},\dots, \frac{b_n}{u_n}]$. We say that a utility vector $u$ is an $\epsilon$-approximate equilibrium utility vector if $H_t(\beta_u) \leq \inf_\beta H_t(\beta)+\epsilon$. It can be shown that for any feasible utilities $u$, we have $H_t(\beta_u)\geq H_t(\betagam)$, and $u$ is the equilibrium utility vector if and only if $H_t(\beta_u) = H_t(\betagam)$.
To that end, let 
\begin{align}
    \cB^\gam(\epsilon)  \defeq \{ \beta > 0: H_t(\beta) \leq \inf_\beta H_t(\beta) + \epsilon \} 
    , \; 
    \cB^*(\epsilon)  \defeq \{ \beta > 0: H(\beta) \leq \inf_\beta H(\beta) + \epsilon \}
    \label{eq:def:cBstar}
    \;.
\end{align}
be the sets of $\epsilon$-approximate solutions to \cref{eq:sample_deg,eq:pop_deg}, respectively. 

Blanket assumptions.
Recall the total supply in the long-run market is one:
$\int s \diff \mu= 1$.
Assume 
the total item set produce one unit of utility in total, i.e., 
$\int v_i s \diff \mu =1$.
Suppose budges of all buyers sum to one, i.e., $\sumiton b_i = 1$.
Let $\ubar{b}\defeq\min_i b_i$. Note the previous budget normalization implies $\ubar{b} \leq 1/n$.
Finally, 
for easy of exposition, we assume the values are bounded $\sup_\Theta v_i (\theta) < \vbar$, for all~$i$.
By the normalization of values and budgets, we know $\ubarbetai = b_i /2 $ and $\bar \beta = 2$.

\section{Consistency and Finite-Sample Bounds}
In this section we introduce several natural empirical estimators based on the observed market equilibrium, and show that they satisfy both consistency and high-probability bounds.

\paragraph{Consistency}


Thanks to the convexity of the dual objectives $H$ and $H_t$, we can 
provide a set of consistency results based on the theory of epi-convergence~\citep{rockafellar2009variational}.

\begin{theorem}[Consistency]\label{thm:consistency}
    It holds that 
    \begin{enumthmresult}
        \item Empirical NSW and empirical individual utilities converge almost surely to their long-run market counterparts, i.e., $\sumiton b_i \log (u^\gam_i) \toas  \sumiton b_i \log (u^*_i)$ and $u^\gam_i \toas u^*_i.$
        \label{it:thm:consistency:1}
        
        \item The empirical pacing multiplier converges almost surely, i.e., $\beta^\gam_i \toas \beta^*_i$.
        \label{it:thm:consistency:2}


        \item Convergence of approximate market equilibrium: 
        $\limsup _{t} \cB^\gam(\epsilon) \subset \cB^*(\epsilon) \text { for all } \epsilon \geq 0$ 
        and 
        $\limsup _{t} \cB^\gam(\epsilon_t) \subset \cB^*(0) = \{\betast\} \text { for all } \epsilon_t \downarrow 0$. Recall the approximate solutions set, $\cBgam$ and $\cBst$, are defined in  \cref{eq:def:cBstar}.
        \label{it:thm:consistency:4}
    \end{enumthmresult}
    Proof in \cref{sec:proof:thm:consistency}.

\end{theorem}

We briefly comment on \cref{it:thm:consistency:4}.
The set limit result can be interpreted from a set distance point of view.
We define the inclusion distance from a set $A$ to a set $B$ by $d_\subset(A,B ) \defeq \inf _\epsilon\{\epsilon \geq 0: A \subset\{y: \operatorname{dist}(y, B) \leq \epsilon\}\}$ where $\operatorname{dist}(y, B) 
\defeq \inf\{\|y - b\| : b\in B \}$. 
Intuitively, $d_\subset(A,B )$ measures how much one should enlarge $B$ such that it covers $A$.
Then for any sequence $\epsilon_n \downarrow 0$, by the second claim in \cref{it:thm:consistency:4}, we know $d_\subset(\cB^\gam(\epsilon_t), \{\betast\}) \to 0$. This shows that the set of approximate solutions of $H_t$ with increasing accuracy centers around $\betast$ as market size grows.

\paragraph{High Probability Bounds}
Next, we refine the consistency results and 
provide finite sample guarantees. We start by focusing on Nash social welfare
and the set of approximate market equilibria. The convergence of utilities and 
pacing multiplier will then be derived from the latter result.

\begin{theorem} \label{thm:lnsw_concentration}
    For any failure probability $0< \eta < 1$, let $t \geq 2 {\vbarsq {\log(4n/\eta)}}$. Then with probablity greater than $1-\eta$, it holds
    \begin{align*}
        \big|\LNSW^\gam - \LNSW^* \big| 
        \leq O(1) {\vbar \big(\sqrt{n\log ((n+\vbar)t)} +  \sqrt{\log(1/\eta)} \big)}{ t^{-1/2}}
        \;.
    \end{align*}
    where $O(1)$ hides only constants.
    Proof in \cref{sec:proof:thm:lnsw_concentration}.

\end{theorem}

\cref{thm:lnsw_concentration} establishes a convergence rate $|\LNSW^\gam - \LNSW^* | = \tilde O _p(\vbar \sqrt n t^{-1/2})$. The proof proceeds by first establishing a pointwise concentration inequality and then applies a discretization argument.

\begin{theorem}[Concentration of Approximate Market Equilibrium] \label{thm:high_prob_containment}
    Let $\eps > 0$ be a tolerance parameter and $\alpha \in (0,1)$ be a failure probability. Then for any $0\leq \delta \leq \eps/2$, to ensure 
        $ \P \big(
        C\cap \cBgam(\delta) 
        \subset 
        C\cap \cBst(\epsilon)
        \big) \geq 1 -2 \alpha$
    it suffices to set
    \begin{align}
        \label{eq:t_condition}
        t \geq O(1)(n\sq + \vbarsq) \min\bigg\{  \frac{1}{ 
            \ubar{b}\epsilon} , \frac{1}{\epsilon\sq} \bigg\}\bigg(n \log\Big(\tfrac{16(2n+\vbar)}{\eps - \delta}\Big) + \log\tfrac{1}{\alpha}\bigg)
            \;,
    \end{align}
    where 
    the set $C= \prod_{i=1}^n [\ubarbetai / 2, \betabar]$, and  
    $O(1)$ hides only absolute constants.
    Proof in \cref{sec:proof:thm:high_prob_containment}.

\end{theorem}

By construction of $C$ we know $\betast \in C$ holds, and so $C\cap \cBst(\epsilon)$ is not empty. 
By \cref{lm:value_concentration} we know that for $t$ sufficiently large, $\betagam \in C$ with high probability, in which case the set 
$C\cap \cBgam(\delta)$ is not empty.

\begin{corollary}
    \label{cor:H_concentration}
    Let $t$ satisfy \cref{eq:t_condition}. Then with probability $\geq 1 - 2 \alpha$ it holds $H(\betagam) \leq H(\betast) + \epsilon$.
\end{corollary}
By simply taking $\delta = 0$ in \cref{thm:high_prob_containment} we obtain the above corollary.
More importantly, it establishes the fast statistical rate 
$    H(\betagam ) - H(\betast) = \tilde{O}_p (t\inv )
$
for $t$ sufficiently large, where we use $\tilde{O}_p$ to ignore logarithmic factors.
In words, when measured in the population dual objective where we take expectation w.r.t.\ the item supply, $\betagam$ converges to $\betast$ with the fast rate $1/t$.
This is in contrast to the usual $1/\sqrt{t}$ rate obtained in \cref{thm:lnsw_concentration}, where $\betagam$ is measured in the sample dual objective. There the $1/\sqrt{t}$ rate is the best obtainable.

By the strong-convexity of dual objective, the containment result can be translated to high-probability convergence of the pacing multipliers and the utility vector.

\begin{corollary}
    \label{cor:beta_u_concentration}
    Let $t$ satisfy \cref{eq:t_condition}. Then with probability $\geq 1 - 2 \alpha$ it holds 
    $\| \betagam - \betast \|_2 \leq \sqrt{\frac{8\epsilon }{ \ubar{b}}}$ and $\| \ugam - \ust \|_2 \leq \frac{4}{\ubar{b}}\sqrt{8\epsilon / \ubar{b}}$.
\end{corollary}

We compare the above corollary with Theorem~9 from~\citet{gao2022infinite} which establishes the convergence rate of the stochastic approximation estimator based on dual averaging algorithm~\citep{xiao2010dual}. 
In particular, they show that the average of the iterates, denoted $\beta_{\mathrm{DA}}$, enjoys a convergence rate of 
$\| \beta_{\mathrm{DA}} - \betast\|_2^2 = \tilde O_p \big(\frac{\vbar\sq }{ \ubar{b}\sq} \frac{1}{t}\big)$,
where $t$ is the number of sampled items.
The rate achieved in \cref{cor:beta_u_concentration} is 
$\|\betagam - \betast\|_2 \sq = \tilde O_p \big(\frac{n(n^2 + \vbar \sq)}{\ubar{b}\sq} \frac{1}{t} \big)$ 
for $t$ sufficiently large. 
Noting $n \leq \ubar{b}\inv$ due to the normalization $\sumiton b_i = 1$,
we see that our rate is worse off by a factor of $n(1+\frac{n\sq}{\vbar \sq})$. 
And yet our estimates is produced by the strategic behavior of the agents without any extra computation at all. 
Moreover, in the computation of the dual averaging estimator the knowledge of values $v_i(\theta)$ is required, while again $\betagam$ can be just observed naturally.

\section{Asymptotics and Inference}

\subsection{Asymptotics}

In this section we derive asymptotic normality results for Nash
social welfare, utilities and pacing multipliers.
As we will see, a central limit theorem (CLT) for Nash social welfare holds under basically no additional assumptions. 
However, the CLTs of pacing multipliers and utilities will require twice continuous differentiability of the population dual objective $H$, with a nonsingular Hessian matrix.
We present CLT results under such a premise, and then provide 
three sufficient conditions under which $H$ is $C^2$ at the optimum.

\begin{theorem}[Asymptotic Normality of Nash Social Welfare]\label{thm:normality}
    It holds that
    \begin{align}  \label{it:thm:normality:1}
        \sqrt{t}(\LNSW^\gam - \LNSW^*) \tod N 
        (0, \sigma\sq_{\NSW})
        \;,
    \end{align}    
    where $ \sigma\sq_{\NSW} 
    =\int_\Theta (p^*)\sq \diff S (\theta) -\big(\int_\Theta p^* \diff S (\theta)\big)\sq 
    =\int_\Theta (p^*)\sq \diff S (\theta) - 1 $.
    Proof in \cref{sec:proof:thm:normality}.
\end{theorem}

To present asymptotics for $\beta$ and $u$ we need a bit more notation.
Let $\Theta_i (\beta) \defeq \{\theta \in \Theta: v_i(\theta)\betai \geq v_k(\theta)\beta_k, \forall k\neq i \}$, i.e., the \emph{potential} winning set of buyer $i$ when the pacing multiplier are $\beta$. Let $\Theta_i^* \defeq \Theta_i(\betast)$.
We will see later that if the dual objective is sufficiently smooth at $\betast$, then the winning sets, 
$\Theta_i^*$, $i\in [n]$, will be disjoint (up to a measure-zero set).
Now we define a map $\must: \Theta \to \Rnp$, which represents the utility all buyers obtain from the item $\theta$ at equilibrium. Formally,
\begin{align}
    \must(\t) = [\xst_1 (\t) v_1 (\t), \dots, \xst_n(\t) v_n(\t)]\tp.
\end{align}
Since $\xst$ is pure, only one entry of $\must(\t)$ is nonzero.

\begin{theorem}[Asymptotic Normality of Individual Behavior] \label{thm:clt_beta_u}
    Assume $H$ is $C^2$ at $\betast$ with non-singular Hessian matrix $\cH=\nabla \sq H(\betast)$. 
Then 
$\sqrt{t} (\beta^\gam - \beta^*) \tod N \big(0, 
        \Sigma_\beta
        \big) $
and
$\sqrt{t}(u^\gam - u^*) \tod N\big
        (0, 
        \Sigma_u
        \big)
        $,       
    where 
         $\Sigma_\beta = \cH\inv\allowbreak 
        \cov(\must)
        \allowbreak
        \cH\inv $ and
        $\Sigma_u = 
        \Diag(-b_i/(\betasti)\sq )
        \allowbreak 
        \cH\inv\allowbreak 
        \cov(\must)
        \allowbreak 
        \cH\inv\allowbreak 
        \Diag(-b_i/(\betasti)\sq ) .$ 
    Here $\cov \must = \int \must (\must)\tp \diff S - (\int \must  \diff S )(\int \must  \diff S)\tp$.
    Proof in \cref{sec:proof:thm:normality}.

\end{theorem}
In \cref{thm:clt_beta_u} we require a strong regularity condition: twice differentiability of $H$, which seems hard to interpret at first sight.
In the next section we derive a set of simpler sufficient conditions for the twice differentiability of the dual objective.

\subsection{Analytical Properties of the Dual Objective} \label{sec:analytical_properties_of_dual_obj}

Intuitively, the expectation operator will smooth out the kinks in the piecewise linear function $f(\cdot, \theta)$; even if $f$ is non-smooth, it is reasonable to hope the expectation counterpart $\fbar$ is smooth, facilitating statistical analysis. First we introduce notation for characterizing smoothness of $\fbar$.

Define the gap between the highest and the second-highest bid under pacing multiplier $\beta$ by
\begin{align}
        \textsf{bidgap}(\beta,\theta) \defeq \big(v(\theta) {\cdot} \beta\big)_{(1)} - \big(v(\theta) {\cdot} \beta\big)_{(2)} \;,
\end{align} 
here $ v(\theta) {\cdot} \beta$ is the elementwise product of $v(\theta)$ and $\beta$, and $(v(\theta) {\cdot} \beta)_{(1)}$ and $(v(\theta) {\cdot} \beta)_{(2)}$
are the greatest and second-greatest entries of $ v(\theta) {\cdot} \beta$, respectively. 
When there is a tie for an item $\theta$, we have $\textsf{bidgap}(\beta,\theta) = 0$. 
When there is no tie for an item $\theta$, the gap $\textsf{bidgap}(\beta,\theta)$ is strictly positive.
Let $G(\beta,\theta) \in \partial f(\beta,\theta)$ be an element in the subgradient set.
The gap function characterizes smoothness of $f$:
$f(\cdot, \theta)$ is differentiable at $
\beta$
$\Leftrightarrow$ $\textsf{bidgap}(\beta,\theta)$ is strictly positive,
in which case $G(\beta,\theta) = \nabla_\beta f(\beta,\theta) = e_{i(\beta,\theta)}v_{i{(\beta,\theta)}}$ with $e_i$ being the $i$-th unit vector and $i(\beta,\theta) = \argmax_i \betai\vithe$.
 When $f(\cdot, \theta)$ is differentiable at $\beta$ a.s., the potential winning sets $\{\Theta_i(\beta) \}_i$ are disjoint (up to a measure-zero set).

\begin{theorem}[First-order differentiability]
    \label{thm:first_differentiability}
    The dual objective $H$ is differentiable at a point $\beta$ if and only if
    \begin{align} \label{eq:as:notie}
        \frac{1}{\textsf{bidgap}(\beta,\theta)} < \infty, \quad \text{for $S$-almost every } \theta 
        \;.
        \tag{\small{NO-TIE}}
     \end{align}
     When \cref{eq:as:notie} holds, $\nabla \fbar (\beta) = \E[G(\beta,\theta)]$.
     Proof and further technical remarks in \cref{proof:sec:analytical_properties_of_dual_obj}.
\end{theorem}

Given the neat characterization of differentiability of dual objective via the 
gap function $\textsf{bidgap}(\beta,\theta)$, 
it is then natural to explore higher-order smoothness, which was needed for some asymptotic normality results.
We provide three classes of markets whose dual objective $H$ enjoys twice differentiability.

\begin{theorem}[Second-order differentiability, Informal]
    \label{thm:second_order_informal}
    If any one of the following holds, then $H$ is $C^2$ at $\betast$. (i) A stronger form of \cref{eq:as:notie} holds, e.g., $\E[\textsf{bidgap}(\beta,\epsilon)\inv]$ or $\esssup_{\theta}\{\textsf{bidgap}(\beta,\theta)\inv\}$ is finite in a neighborhood of $\betast$. (ii) The distribution of $v=(v_1,\dots,v_n):\Theta \to \Rnp$ is smooth enough. 
    (iii) $\Theta = [0,1]$ and the valuations $v_i(\cdot)$'s are linear functions.
\end{theorem}

We briefly comment on the three candidate sufficient conditions; for a rigorous statement we refer readers to \cref{sec:analytical_formal}.
Based on the differentiability characterization, it is natural to search for a stronger form of \cref{eq:as:notie} and hope that such a refinement could lead to second-order differentiability. Condition (i) gives two such refinements. 
Condition (ii) is motivated by the idea that expectation operator tends to produce smooth functions. 
Given that the dual objective $H$ is the expectation of the non-smooth function $f$ (plus a smooth term $\Psi$), we expect under certain conditions on the expectation operator $H$ is twice differentiable. The exact smoothness requirement is presented in the appendix, which we show is easy to verify for 
several common distributions.
Finally, Condition (iii) considers the linear-valuations setting of~\citet{gao2022infinite}, where the authors provide tractable convex programs for computing the infinite-dimensional equilibrium. 
Here we give another interesting properties of this setup by showing that the dual objective is $C^2$. We also discuss 
how this can be extended to piecewise linear value functions in the appendix.

\subsection{Inference}\label{sec:inference}
In this section we discuss constructing confidence intervals for Nash social welfare, the pacing multipliers, and the utilities. 
We remark that 
the observed NSW, $\LNSW^\gam$, is a negatively-biased estimate of the NSW, $\LNSW^*$, of the long-run ME, i.e., $ \E[\LNSW^\gam] - \LNSW^* \leq 0$.\footnote{Note 
$ \E[\LNSW^\gam] - \LNSW^* = \E[\min_\beta H_t(\beta)] - H(\betast) \leq \min_\beta \E[H_t(\beta)]  - H(\betast)  = 0.$}
Moreover, it can be shown that, when the items are i.i.d.\ $\E[\min H_t] \leq \E[\min H_{t+1}]$ using Proposition~16 from~\citet{shapiro2003monte}. Monotonicity tells us that increasing the size of market produces on average less biased estimates of the long-run NSW.

To construct a confidence interval for Nash social welfare one needs to estimate the asymptotic variance. We let 
  $  \hat\sigma_{\NSW}\sq \defeq \frac{1}{t} \sumtau\big( F(\betagam, \thetau)  - H_t(\beta^\gam)\big)\sq 
    =\big( \frac{1}{t}\sumtau 
    (p^{\gam,\tau})\sq \big)- 1.
    $
where $p^{\gam,\tau}$ is the price of item $\theta^\tau$ in the observed market. We emphasize that in the computation of the variance estimator $ \hat\sigma_{\NSW}\sq$ one does not need knowledge of values $\{\vithetau\}_{i,\tau}$.
All that is needed is the equilibrium prices $\pgam = (p^{\gam,1},\dots, p^{\gam,t})$ of the items. 
Given the variance estimator, we 
construct the confidence interval 
$ [ \LNSW^\gam \pm z_{\alpha/2}  \frac{\hat \sigma_{\NSW}}{\sqrt t}]$,
where $z_\alpha$ is the $\alpha$-th quantile of a standard normal.
The next theorem establishes validity of the variance estimator.
\begin{theorem} \label{thm:ci_lnsw}
    It holds that $\hat \sigma _{\NSW} \toprob \sigma_{\NSW}\sq$.
    Given $0< \alpha < 1$, it holds that
   $ 
    \lim_{t\to\infty} \P\big( \LNSW^* \in [ \LNSW^\gam \pm z_{\alpha/2}  \hat \sigma_{\NSW}/\sqrt t \,] \big) 
        =
        1-\alpha 
$.  Proof in \cref{sec:proof_var_est}.

\end{theorem}

Estimation of the variance matrices for $\beta$ and $u$ is more complicated. The main difficulty lies in estimating the inverse Hessian matrix. Due to the non-smoothness of the sample function, we cannot exchange the twice differential operator and expectation, and thus the plug-in estimator, i.e., 
the sample average Hessian, is a biased estimator for the Hessian of the population function in general.

We provide a brief discussion of variance estimation under the following two simplified scenarios in \cref{sec:hessian_estimation}. First, in the case where $\E [\textsf{bidgap}(\beta, \theta)\inv ]< \infty$ holds in a neighborhood of $\betast$, which we recall is a stronger form \cref{eq:as:notie}, we prove that a plug-in type variance estimator is valid. 
Second, if we have knowledge of $\{\vithetau\}_{i,\tau}$, then we give a numerical difference estimator for the Hessian which is consistent.

\section{Extension: Revenue Inference in Quasilinear Fisher Market} \label{sec:quasilinear}
As we mentioned previously, 
in a linear Fisher market all buyer budgets are extracted, 
i.e., $\sumtau \pgamtau$ 
equals $\sumiton b_i$ in the observed market (and similarly for the underlying market), and there is thus nothing to infer about revenue if we know the budgets of each buyer.
A quasilinear (QL) utility is one such that the cost of purchasing goods is deducted from the utility, i.e., $u_i(x) = \lg x - p, v_i \rg$.
This may give buyers an incentive to leave some budget unspent.
In the finite-dimensional case,~\citet{chen2007note} and ~\citet{cole2017convex} show that there is an variant of EG program that captures the market equilibrium with QL utility. 
Furthermore,~\citet{conitzer2022pacing} showed that budget management in ad auctions with first-price auctions can be computed by Fisher markets with QL utilities. 
A QL variant of infinite-dimensional markets and an EG program are given by~\citet{gao2022infinite}.

Quaislinear market equilibria (QME) are defined analogously to the linear variant via market clearance conditions and buyer optimality;
we present the formal finite and infinite-dimensional definitions in \cref{sec:proof:thm:rev_convergence}. 
The demand sets are
$\argmax \{ \langle v_i - p, x_i \rangle : x_i \in L^\infty_+,\, \langle p, x_i\rangle \leq b_i \}$ 
in 
the long-run QME 
and 
$ \argmax \{   \langle {{v}}_i(\gam) - p , {x}_i \rangle : {x}_i \geq 0,\,  \langle p,  {x}_i\rangle \leq b_i \}$ in the observed QME.
QME has several distinctions from the linear ME.
First, in QME we cannot normalize both valuations and budgets, since buyers' budgets have value
outside the current market.
Second, budgets are not fully extracted in QME, which motivates the need for statistical analysis. 
Third, the pacing multipliers are restricted to $\beta \leq 1$, and may lie on the resulting boundary.

Define the revenues from the observed and the long-run market as follows: 
$
    \REV^\gam \defeq \frac1t \sumtau \pgamtau 
    ,
    \REV^* \defeq \int_\Theta \pst \diff S(\theta).
$
Assume $\sumiton b_i  = 1$ and unit supply $\int s\diff \mu = 1$.
Let $\nui \defeq \int v_i\diff S $ be the average value of buyer $i$. Let $\nubar = \max_i \nui$. 
Assume we observe the market $\QMEgam(b,v,\tfrac1t 1_t) = (\xgam,\ugam,\pgam)$.
Then we show that consistency and high-probability bounds hold for the revenue estimator.

\begin{theorem}[Revenue Convergence]
    \label{thm:rev_convergence}
    It holds that $\REV^\gam\toas \REV^*$ and $|\REV^\gam - \REV^*| = 
        \tilde{O}_p
        \Big(\frac{\vbar \sqrt{n} (\vbar + 2\nubar n + 1 ) }{ \ubar{b} }\frac{1}{\sqrt{t}}\Big)
        $ for $t$ sufficiently large.
    Proofs are in \cref{sec:proof:thm:rev_convergence}.
\end{theorem}

We leave CLT results for revenue estimates in quasilinear markets as an open problem. The main challenge compared to the linear case is that the optimal pacing multipliers can lie on the boundary of the constraint set.
More precisely, if the equilibrium pacing multiplier of a buyer is in the interior, then his budget is fully extracted. On the other hand, if it is on the boundary, the buyer retains a portion of his budget at equilibrium. 
When the optimum of the expectation function lies on the boundary of the constraint set, the asymptotic variance of the sample average optimum takes on a complicated expression~\cite[Theorem 3.3]{shapiro1989asymptotic}, which makes variance estimation difficult.


\appendix
\bibliographystyle{plainnat}
\bibliography{./refs.bib}

\begin{thebibliography}{82}
\providecommand{\natexlab}[1]{#1}
\providecommand{\url}[1]{\texttt{#1}}
\expandafter\ifx\csname urlstyle\endcsname\relax
  \providecommand{\doi}[1]{doi: #1}\else
  \providecommand{\doi}{doi: \begingroup \urlstyle{rm}\Url}\fi

\bibitem[Aleksandrov et~al.(2015)Aleksandrov, Aziz, Gaspers, and
  Walsh]{aleksandrov2015online}
Martin Aleksandrov, Haris Aziz, Serge Gaspers, and Toby Walsh.
\newblock Online fair division: Analysing a food bank problem.
\newblock \emph{arXiv preprint arXiv:1502.07571}, 2015.

\bibitem[Allouah et~al.(2022)Allouah, Kroer, Zhang, Avadhanula, Dania, Gocmen,
  Pupyrev, Shah, and Stier]{allouah2022robust}
Amine Allouah, Christian Kroer, Xuan Zhang, Vashist Avadhanula, Anil Dania,
  Caner Gocmen, Sergey Pupyrev, Parikshit Shah, and Nicolas Stier.
\newblock Robust and fair work allocation, 2022.
\newblock URL \url{https://arxiv.org/abs/2202.05194}.

\bibitem[Aronow and Samii(2017)]{aronow2017estimating}
Peter~M Aronow and Cyrus Samii.
\newblock Estimating average causal effects under general interference, with
  application to a social network experiment.
\newblock \emph{The Annals of Applied Statistics}, 11\penalty0 (4):\penalty0
  1912--1947, 2017.

\bibitem[Athey et~al.(2018)Athey, Eckles, and Imbens]{athey2018exact}
Susan Athey, Dean Eckles, and Guido~W Imbens.
\newblock Exact p-values for network interference.
\newblock \emph{Journal of the American Statistical Association}, 113\penalty0
  (521):\penalty0 230--240, 2018.

\bibitem[Azar et~al.(2016)Azar, Buchbinder, and Jain]{azar2016allocate}
Yossi Azar, Niv Buchbinder, and Kamal Jain.
\newblock How to allocate goods in an online market?
\newblock \emph{Algorithmica}, 74\penalty0 (2):\penalty0 589--601, 2016.

\bibitem[Banerjee et~al.(2022)Banerjee, Gkatzelis, Gorokh, and
  Jin]{banerjee2022online}
Siddhartha Banerjee, Vasilis Gkatzelis, Artur Gorokh, and Billy Jin.
\newblock Online nash social welfare maximization with predictions.
\newblock In \emph{Proceedings of the 2022 Annual ACM-SIAM Symposium on
  Discrete Algorithms (SODA)}, pages 1--19. SIAM, 2022.

\bibitem[Bauschke et~al.(2011)Bauschke, Combettes, et~al.]{bauschke2011convex}
Heinz~H Bauschke, Patrick~L Combettes, et~al.
\newblock \emph{Convex analysis and monotone operator theory in Hilbert
  spaces}, volume 408.
\newblock Springer, 2011.

\bibitem[Beck(2017)]{beck2017first}
Amir Beck.
\newblock \emph{First-Order Methods in Optimization}, volume~25.
\newblock SIAM, 2017.

\bibitem[Bei et~al.(2019{\natexlab{a}})Bei, Garg, and Hoefer]{bei2019ascending}
Xiaohui Bei, Jugal Garg, and Martin Hoefer.
\newblock Ascending-price algorithms for unknown markets.
\newblock \emph{ACM Transactions on Algorithms (TALG)}, 15\penalty0
  (3):\penalty0 1--33, 2019{\natexlab{a}}.

\bibitem[Bei et~al.(2019{\natexlab{b}})Bei, Garg, Hoefer, and
  Mehlhorn]{bei2019earning}
Xiaohui Bei, Jugal Garg, Martin Hoefer, and Kurt Mehlhorn.
\newblock Earning and utility limits in fisher markets.
\newblock \emph{ACM Transactions on Economics and Computation (TEAC)},
  7\penalty0 (2):\penalty0 1--35, 2019{\natexlab{b}}.

\bibitem[Bertsekas(1973)]{bertsekas1973stochastic}
Dimitri~P Bertsekas.
\newblock Stochastic optimization problems with nondifferentiable cost
  functionals.
\newblock \emph{Journal of Optimization Theory and Applications}, 12\penalty0
  (2):\penalty0 218--231, 1973.

\bibitem[Birnbaum et~al.(2011)Birnbaum, Devanur, and
  Xiao]{birnbaum2011distributed}
Benjamin Birnbaum, Nikhil~R Devanur, and Lin Xiao.
\newblock Distributed algorithms via gradient descent for fisher markets.
\newblock In \emph{Proceedings of the 12th ACM Conference on Electronic
  Commerce}, pages 127--136. ACM, 2011.

\bibitem[Borgs et~al.(2007)Borgs, Chayes, Immorlica, Jain, Etesami, and
  Mahdian]{borgs2007dynamics}
Christian Borgs, Jennifer Chayes, Nicole Immorlica, Kamal Jain, Omid Etesami,
  and Mohammad Mahdian.
\newblock Dynamics of bid optimization in online advertisement auctions.
\newblock In \emph{Proceedings of the 16th international conference on World
  Wide Web}, pages 531--540, 2007.

\bibitem[Budish(2011)]{budish2011combinatorial}
Eric Budish.
\newblock The combinatorial assignment problem: Approximate competitive
  equilibrium from equal incomes.
\newblock \emph{Journal of Political Economy}, 119\penalty0 (6):\penalty0
  1061--1103, 2011.

\bibitem[Budish et~al.(2016)Budish, Cachon, Kessler, and
  Othman]{budish2016course}
Eric Budish, G{\'e}rard~P Cachon, Judd~B Kessler, and Abraham Othman.
\newblock Course match: A large-scale implementation of approximate competitive
  equilibrium from equal incomes for combinatorial allocation.
\newblock \emph{Operations Research}, 65\penalty0 (2):\penalty0 314--336, 2016.

\bibitem[Caragiannis et~al.(2016)Caragiannis, Kurokawa, Moulin, Procaccia,
  Shah, and Wang]{caragiannis2016unreasonable}
Ioannis Caragiannis, David Kurokawa, Herv{\'e} Moulin, Ariel~D Procaccia,
  Nisarg Shah, and Junxing Wang.
\newblock The unreasonable fairness of maximum {Nash} welfare.
\newblock In \emph{Proceedings of the 2016 ACM Conference on Economics and
  Computation}, pages 305--322. ACM, 2016.

\bibitem[Caragiannis et~al.(2019)Caragiannis, Kurokawa, Moulin, Procaccia,
  Shah, and Wang]{caragiannis2019unreasonable}
Ioannis Caragiannis, David Kurokawa, Herv{\'e} Moulin, Ariel~D Procaccia,
  Nisarg Shah, and Junxing Wang.
\newblock The unreasonable fairness of maximum nash welfare.
\newblock \emph{ACM Transactions on Economics and Computation (TEAC)},
  7\penalty0 (3):\penalty0 1--32, 2019.

\bibitem[Cen and Shah(2022)]{cen2022regret}
Sarah~H Cen and Devavrat Shah.
\newblock Regret, stability \& fairness in matching markets with bandit
  learners.
\newblock In \emph{International Conference on Artificial Intelligence and
  Statistics}, pages 8938--8968. PMLR, 2022.

\bibitem[Chen et~al.(2007)Chen, Ye, and Zhang]{chen2007note}
Lihua Chen, Yinyu Ye, and Jiawei Zhang.
\newblock A note on equilibrium pricing as convex optimization.
\newblock In \emph{International Workshop on Web and Internet Economics}, pages
  7--16. Springer, 2007.

\bibitem[Cheung et~al.(2020)Cheung, Cole, and Devanur]{cheung2020tatonnement}
Yun~Kuen Cheung, Richard Cole, and Nikhil~R Devanur.
\newblock Tatonnement beyond gross substitutes? {G}radient descent to the
  rescue.
\newblock \emph{Games and Economic Behavior}, 123:\penalty0 295--326, 2020.

\bibitem[Cole and Fleischer(2008)]{cole2008fast}
Richard Cole and Lisa Fleischer.
\newblock Fast-converging tatonnement algorithms for one-time and ongoing
  market problems.
\newblock In \emph{Proceedings of the fortieth annual ACM symposium on Theory
  of computing}, pages 315--324, 2008.

\bibitem[Cole and Gkatzelis(2018)]{cole2018approximating}
Richard Cole and Vasilis Gkatzelis.
\newblock Approximating the {Nash} social welfare with indivisible items.
\newblock \emph{SIAM Journal on Computing}, 47\penalty0 (3):\penalty0
  1211--1236, 2018.

\bibitem[Cole et~al.(2017)Cole, Devanur, Gkatzelis, Jain, Mai, Vazirani, and
  Yazdanbod]{cole2017convex}
Richard Cole, Nikhil~R Devanur, Vasilis Gkatzelis, Kamal Jain, Tung Mai,
  Vijay~V Vazirani, and Sadra Yazdanbod.
\newblock Convex program duality, fisher markets, and {Nash} social welfare.
\newblock In \emph{18th ACM Conference on Economics and Computation, EC 2017}.
  Association for Computing Machinery, Inc, 2017.

\bibitem[Conitzer et~al.(2022{\natexlab{a}})Conitzer, Kroer, Panigrahi,
  Schrijvers, Stier-Moses, Sodomka, and Wilkens]{conitzer2022pacing}
Vincent Conitzer, Christian Kroer, Debmalya Panigrahi, Okke Schrijvers,
  Nicolas~E Stier-Moses, Eric Sodomka, and Christopher~A Wilkens.
\newblock Pacing equilibrium in first price auction markets.
\newblock \emph{Management Science}, 2022{\natexlab{a}}.

\bibitem[Conitzer et~al.(2022{\natexlab{b}})Conitzer, Kroer, Sodomka, and
  Stier-Moses]{conitzer2022multiplicative}
Vincent Conitzer, Christian Kroer, Eric Sodomka, and Nicolas~E Stier-Moses.
\newblock Multiplicative pacing equilibria in auction markets.
\newblock \emph{Operations Research}, 70\penalty0 (2):\penalty0 963--989,
  2022{\natexlab{b}}.

\bibitem[Dai and Jordan(2021)]{dai2021learning}
Xiaowu Dai and Michael Jordan.
\newblock Learning in multi-stage decentralized matching markets.
\newblock In M.~Ranzato, A.~Beygelzimer, Y.~Dauphin, P.S. Liang, and J.~Wortman
  Vaughan, editors, \emph{Advances in Neural Information Processing Systems},
  volume~34, pages 12798--12809. Curran Associates, Inc., 2021.
\newblock URL
  \url{https://proceedings.neurips.cc/paper/2021/file/6a571fe98a2ba453e84923b447d79cff-Paper.pdf}.

\bibitem[Deng et~al.(2003)Deng, Papadimitriou, and Safra]{deng2003complexity}
Xiaotie Deng, Christos Papadimitriou, and Shmuel Safra.
\newblock On the complexity of price equilibria.
\newblock \emph{Journal of Computer and System Sciences}, 67\penalty0
  (2):\penalty0 311--324, 2003.

\bibitem[Devanur et~al.(2008)Devanur, Papadimitriou, Saberi, and
  Vazirani]{devanur2008market}
Nikhil~R Devanur, Christos~H Papadimitriou, Amin Saberi, and Vijay~V Vazirani.
\newblock Competitive equilibrium via a primal-dual algorithm for a convex
  program.
\newblock \emph{Journal of the ACM (JACM)}, 55\penalty0 (5):\penalty0 1--18,
  2008.

\bibitem[Devanur et~al.(2018)Devanur, Garg, Mehta, Vaziranb, and
  Yazdanbod]{devanur2018new}
Nikhil~R Devanur, Jugal Garg, Ruta Mehta, Vijay~V Vaziranb, and Sadra
  Yazdanbod.
\newblock A new class of combinatorial markets with covering constraints:
  Algorithms and applications.
\newblock In \emph{Proceedings of the Twenty-Ninth Annual ACM-SIAM Symposium on
  Discrete Algorithms}, pages 2311--2325. SIAM, 2018.

\bibitem[Eisenberg(1961)]{eisenberg1961aggregation}
Edmund Eisenberg.
\newblock Aggregation of utility functions.
\newblock \emph{Management Science}, 7\penalty0 (4):\penalty0 337--350, 1961.

\bibitem[Eisenberg and Gale(1959)]{eisenberg1959consensus}
Edmund Eisenberg and David Gale.
\newblock Consensus of subjective probabilities: The pari-mutuel method.
\newblock \emph{The Annals of Mathematical Statistics}, 30\penalty0
  (1):\penalty0 165--168, 1959.

\bibitem[Gao and Kroer(2020)]{gao2020first}
Yuan Gao and Christian Kroer.
\newblock First-order methods for large-scale market equilibrium computation.
\newblock In \emph{Neural Information Processing Systems 2020, NeurIPS 2020},
  2020.

\bibitem[Gao and Kroer(2022)]{gao2022infinite}
Yuan Gao and Christian Kroer.
\newblock Infinite-dimensional fisher markets and tractable fair division.
\newblock \emph{Operation Research, Forthcoming}, 2022.

\bibitem[Gao et~al.(2021)Gao, Kroer, and Peysakhovich]{gao2021online}
Yuan Gao, Christian Kroer, and Alex Peysakhovich.
\newblock Online market equilibrium with application to fair division.
\newblock \emph{arXiv preprint arXiv:2103.12936}, 2021.

\bibitem[Ghodsi et~al.(2011)Ghodsi, Zaharia, Hindman, Konwinski, Shenker, and
  Stoica]{ghodsi2011dominant}
Ali Ghodsi, Matei Zaharia, Benjamin Hindman, Andy Konwinski, Scott Shenker, and
  Ion Stoica.
\newblock Dominant resource fairness: Fair allocation of multiple resource
  types.
\newblock In \emph{Nsdi}, volume~11, pages 24--24, 2011.

\bibitem[Gorokh et~al.(2019)Gorokh, Banerjee, and Iyer]{gorokh2019remarkable}
Artur Gorokh, Siddhartha Banerjee, and Krishnamurthy Iyer.
\newblock The remarkable robustness of the repeated fisher market.
\newblock \emph{Available at SSRN 3411444}, 2019.

\bibitem[Guo et~al.(2021)Guo, Kandasamy, Gonzalez, Jordan, and
  Stoica]{guo2021online}
Wenshuo Guo, Kirthevasan Kandasamy, Joseph~E Gonzalez, Michael~I Jordan, and
  Ion Stoica.
\newblock Online learning of competitive equilibria in exchange economies.
\newblock \emph{arXiv preprint arXiv:2106.06616}, 2021.

\bibitem[Hjort and Pollard(2011)]{hjort2011asymptotics}
Nils~Lid Hjort and David Pollard.
\newblock Asymptotics for minimisers of convex processes.
\newblock \emph{arXiv preprint arXiv:1107.3806}, 2011.

\bibitem[Hu et~al.(2022)Hu, Li, and Wager]{hu2022average}
Yuchen Hu, Shuangning Li, and Stefan Wager.
\newblock Average direct and indirect causal effects under interference.
\newblock \emph{Biometrika}, 2022.

\bibitem[Hudgens and Halloran(2008)]{hudgens2008toward}
Michael~G Hudgens and M~Elizabeth Halloran.
\newblock Toward causal inference with interference.
\newblock \emph{Journal of the American Statistical Association}, 103\penalty0
  (482):\penalty0 832--842, 2008.

\bibitem[Im et~al.(2017)Im, Kulkarni, and Munagala]{im2017competitive}
Sungjin Im, Janardhan Kulkarni, and Kamesh Munagala.
\newblock Competitive algorithms from competitive equilibria: Non-clairvoyant
  scheduling under polyhedral constraints.
\newblock \emph{Journal of the ACM (JACM)}, 65\penalty0 (1):\penalty0 1--33,
  2017.

\bibitem[Jagadeesan et~al.(2021)Jagadeesan, Wei, Wang, Jordan, and
  Steinhardt]{jagadeesan2021learning}
Meena Jagadeesan, Alexander Wei, Yixin Wang, Michael Jordan, and Jacob
  Steinhardt.
\newblock Learning equilibria in matching markets from bandit feedback.
\newblock \emph{Advances in Neural Information Processing Systems},
  34:\penalty0 3323--3335, 2021.

\bibitem[Jain(2007)]{jain2007polynomial}
Kamal Jain.
\newblock A polynomial time algorithm for computing an arrow--debreu market
  equilibrium for linear utilities.
\newblock \emph{SIAM Journal on Computing}, 37\penalty0 (1):\penalty0 303--318,
  2007.

\bibitem[Kash et~al.(2014)Kash, Procaccia, and Shah]{kash2014no}
Ian Kash, Ariel~D Procaccia, and Nisarg Shah.
\newblock No agent left behind: Dynamic fair division of multiple resources.
\newblock \emph{Journal of Artificial Intelligence Research}, 51:\penalty0
  579--603, 2014.

\bibitem[Kim et~al.(2015)Kim, Pasupathy, and Henderson]{Kim2015}
Sujin Kim, Raghu Pasupathy, and Shane~G. Henderson.
\newblock \emph{A Guide to Sample Average Approximation}, pages 207--243.
\newblock Springer New York, New York, NY, 2015.
\newblock \doi{10.1007/978-1-4939-1384-8_8}.

\bibitem[Kroer and Peysakhovich(2019)]{kroer2019scalable}
Christian Kroer and Alexander Peysakhovich.
\newblock Scalable fair division for 'at most one' preferences.
\newblock \emph{arXiv preprint arXiv:1909.10925}, 2019.

\bibitem[Kroer and Stier-Moses(2022)]{kroer2022market}
Christian Kroer and Nicolas~E Stier-Moses.
\newblock Market equilibrium models in large-scale internet markets.
\newblock \emph{Innovative Technology at the Interface of Finance and
  Operations. Springer Series in Supply Chain Management. Springer Natures,
  Forthcoming}, 2022.

\bibitem[Kroer et~al.(2021)Kroer, Peysakhovich, Sodomka, and
  Stier-Moses]{kroer2021computing}
Christian Kroer, Alexander Peysakhovich, Eric Sodomka, and Nicolas~E
  Stier-Moses.
\newblock Computing large competitive equilibria using abstractions.
\newblock \emph{Operations Research}, 2021.

\bibitem[Leung(2020)]{leung2020treatment}
Michael~P Leung.
\newblock Treatment and spillover effects under network interference.
\newblock \emph{Review of Economics and Statistics}, 102\penalty0 (2):\penalty0
  368--380, 2020.

\bibitem[Li and Wager(2022)]{li2022random}
Shuangning Li and Stefan Wager.
\newblock Random graph asymptotics for treatment effect estimation under
  network interference.
\newblock \emph{The Annals of Statistics}, 50\penalty0 (4):\penalty0
  2334--2358, 2022.

\bibitem[Liao et~al.(2022)Liao, Gao, and Kroer]{liao2022dualaveraging}
Luofeng Liao, Yuan Gao, and Christian Kroer.
\newblock Nonstationary dual averaging and online fair allocation.
\newblock \emph{arXiv preprint arXiv:2202.11614v1}, 2022.

\bibitem[Liu et~al.(2021)Liu, Ruan, Mania, and Jordan]{liu2021bandit}
Lydia~T Liu, Feng Ruan, Horia Mania, and Michael~I Jordan.
\newblock Bandit learning in decentralized matching markets.
\newblock \emph{J. Mach. Learn. Res.}, 22:\penalty0 211--1, 2021.

\bibitem[Liu et~al.(2022)Liu, Lu, Wang, Jordan, and Yang]{liu2022welfare}
Zhihan Liu, Miao Lu, Zhaoran Wang, Michael Jordan, and Zhuoran Yang.
\newblock Welfare maximization in competitive equilibrium: Reinforcement
  learning for markov exchange economy.
\newblock In \emph{International Conference on Machine Learning}, pages
  13870--13911. PMLR, 2022.

\bibitem[McElfresh et~al.(2020)McElfresh, Kroer, Pupyrev, Sodomka,
  Sankararaman, Chauvin, Dexter, and Dickerson]{mcelfresh2020matching}
Duncan~C McElfresh, Christian Kroer, Sergey Pupyrev, Eric Sodomka,
  Karthik~Abinav Sankararaman, Zack Chauvin, Neil Dexter, and John~P Dickerson.
\newblock Matching algorithms for blood donation.
\newblock In \emph{Proceedings of the 21st ACM Conference on Economics and
  Computation}, pages 463--464, 2020.

\bibitem[Min et~al.(2022)Min, Wang, Xu, Wang, Jordan, and Yang]{min2022learn}
Yifei Min, Tianhao Wang, Ruitu Xu, Zhaoran Wang, Michael~I Jordan, and Zhuoran
  Yang.
\newblock Learn to match with no regret: Reinforcement learning in markov
  matching markets.
\newblock \emph{arXiv preprint arXiv:2203.03684}, 2022.

\bibitem[Munro et~al.(2021)Munro, Wager, and Xu]{munro2021treatment}
Evan Munro, Stefan Wager, and Kuang Xu.
\newblock Treatment effects in market equilibrium.
\newblock \emph{arXiv preprint arXiv:2109.11647}, 2021.

\bibitem[Murray et~al.(2020{\natexlab{a}})Murray, Kroer, Peysakhovich, and
  Shah]{murray2020robust}
Riley Murray, Christian Kroer, Alex Peysakhovich, and Parikshit Shah.
\newblock Robust market equilibria with uncertain preferences.
\newblock In \emph{Proceedings of the AAAI Conference on Artificial
  Intelligence}, volume~34, pages 2192--2199, 2020{\natexlab{a}}.

\bibitem[Murray et~al.(2020{\natexlab{b}})Murray, Kroer, Peysakhovich, and
  Shah]{robust_blog}
Riley Murray, Christian Kroer, Alex Peysakhovich, and Parikshit Shah.
\newblock
  https://research.fb.com/blog/2020/09/robust-market-equilibria-how-to-model-uncertain-buyer-preferences/,
  Sep 2020{\natexlab{b}}.

\bibitem[Nesterov and Shikhman(2018)]{nesterov2018computation}
Yurii Nesterov and Vladimir Shikhman.
\newblock Computation of fisher--gale equilibrium by auction.
\newblock \emph{Journal of the Operations Research Society of China},
  6\penalty0 (3):\penalty0 349--389, 2018.

\bibitem[Newey and McFadden(1994)]{newey1994large}
Whitney~K Newey and Daniel McFadden.
\newblock Large sample estimation and hypothesis testing.
\newblock \emph{Handbook of econometrics}, 4:\penalty0 2111--2245, 1994.

\bibitem[Nisan et~al.(2007)Nisan, Roughgarden, Tardos, and
  Vazirani]{nisan2007algorithmic}
Noam Nisan, Tim Roughgarden, Eva Tardos, and Vijay~V Vazirani.
\newblock \emph{Algorithmic game theory}.
\newblock Cambridge University Press, 2007.

\bibitem[Othman et~al.(2010)Othman, Sandholm, and Budish]{othman2010finding}
Abraham Othman, Tuomas Sandholm, and Eric Budish.
\newblock Finding approximate competitive equilibria: efficient and fair course
  allocation.
\newblock In \emph{AAMAS}, volume~10, pages 873--880, 2010.

\bibitem[Othman et~al.(2016)Othman, Papadimitriou, and
  Rubinstein]{othman2016complexity}
Abraham Othman, Christos Papadimitriou, and Aviad Rubinstein.
\newblock The complexity of fairness through equilibrium.
\newblock \emph{ACM Transactions on Economics and Computation (TEAC)},
  4\penalty0 (4):\penalty0 1--19, 2016.

\bibitem[Parkes et~al.(2015)Parkes, Procaccia, and Shah]{parkes2015beyond}
David~C Parkes, Ariel~D Procaccia, and Nisarg Shah.
\newblock Beyond dominant resource fairness: Extensions, limitations, and
  indivisibilities.
\newblock \emph{ACM Transactions on Economics and Computation (TEAC)},
  3\penalty0 (1):\penalty0 1--22, 2015.

\bibitem[Peysakhovich and Kroer(2019)]{peysakhovich2019fair}
Alexander Peysakhovich and Christian Kroer.
\newblock Fair division without disparate impact.
\newblock \emph{Mechanism Design for Social Good}, 2019.

\bibitem[Rockafellar(1970)]{rockafellar1970convex}
R~Tyrrell Rockafellar.
\newblock \emph{Convex analysis}, volume~18.
\newblock Princeton university press, 1970.

\bibitem[Rockafellar and Wets(2009)]{rockafellar2009variational}
R~Tyrrell Rockafellar and Roger J-B Wets.
\newblock \emph{Variational analysis}, volume 317.
\newblock Springer Science \& Business Media, 2009.

\bibitem[Sahoo and Wager(2022)]{sahoo2022policy}
Roshni Sahoo and Stefan Wager.
\newblock Policy learning with competing agents.
\newblock \emph{arXiv preprint arXiv:2204.01884}, 2022.

\bibitem[Shapiro(1989)]{shapiro1989asymptotic}
Alexander Shapiro.
\newblock Asymptotic properties of statistical estimators in stochastic
  programming.
\newblock \emph{The Annals of Statistics}, 17\penalty0 (2):\penalty0 841--858,
  1989.

\bibitem[Shapiro(2003)]{shapiro2003monte}
Alexander Shapiro.
\newblock Monte carlo sampling methods.
\newblock \emph{Handbooks in operations research and management science},
  10:\penalty0 353--425, 2003.

\bibitem[Shapiro et~al.(2021)Shapiro, Dentcheva, and
  Ruszczynski]{shapiro2021lectures}
Alexander Shapiro, Darinka Dentcheva, and Andrzej Ruszczynski.
\newblock \emph{Lectures on stochastic programming: modeling and theory}.
\newblock SIAM, 2021.

\bibitem[Shmyrev(2009)]{shmyrev2009algorithm}
Vadim~I Shmyrev.
\newblock An algorithm for finding equilibrium in the linear exchange model
  with fixed budgets.
\newblock \emph{Journal of Applied and Industrial Mathematics}, 3\penalty0
  (4):\penalty0 505, 2009.

\bibitem[Sinclair et~al.(2022)Sinclair, Banerjee, and Yu]{sinclair2021fairness}
Sean~R Sinclair, Siddhartha Banerjee, and Christina~Lee Yu.
\newblock Sequential fair allocation: Achieving the optimal envy-efficiency
  tradeoff curve.
\newblock In \emph{Abstracts of the 2022 SIGMETRICS/Performance Joint
  International Conference on Measurement and Modeling of Computer Systems},
  2022.

\bibitem[Van~der Vaart(2000)]{van2000asymptotic}
Aad~W Van~der Vaart.
\newblock \emph{Asymptotic statistics}, volume~3.
\newblock Cambridge university press, 2000.

\bibitem[Varian(1974)]{varian1974equity}
Hal~R Varian.
\newblock Equity, envy, and efficiency.
\newblock \emph{Journal of Economic Theory}, 9\penalty0 (1):\penalty0 63--91,
  1974.

\bibitem[Vazirani(2007)]{vazirani_2007}
Vijay~V. Vazirani.
\newblock \emph{Combinatorial Algorithms for Market Equilibria}, pages
  103--134.
\newblock Cambridge University Press, 2007.
\newblock \doi{10.1017/CBO9780511800481.007}.

\bibitem[Wager and Xu(2021)]{Wager2021}
Stefan Wager and Kuang Xu.
\newblock Experimenting in equilibrium.
\newblock \emph{Management Science}, 67\penalty0 (11):\penalty0 6694--6715,
  November 2021.
\newblock \doi{10.1287/mnsc.2020.3844}.
\newblock URL \url{https://doi.org/10.1287/mnsc.2020.3844}.

\bibitem[Wang(1985)]{wang1985distribution}
Jinde Wang.
\newblock Distribution sensitivity analysis for stochastic programs with
  complete recourse.
\newblock \emph{Mathematical Programming}, 31\penalty0 (3):\penalty0 286--297,
  1985.

\bibitem[Wu and Zhang(2007)]{wu2007proportional}
Fang Wu and Li~Zhang.
\newblock Proportional response dynamics leads to market equilibrium.
\newblock In \emph{Proceedings of the thirty-ninth annual ACM symposium on
  Theory of computing}, pages 354--363, 2007.

\bibitem[Xiao(2010)]{xiao2010dual}
Lin Xiao.
\newblock Dual averaging methods for regularized stochastic learning and online
  optimization.
\newblock \emph{Journal of Machine Learning Research}, 11:\penalty0 2543--2596,
  2010.

\bibitem[Ye(2008)]{ye2008path}
Yinyu Ye.
\newblock A path to the arrow--debreu competitive market equilibrium.
\newblock \emph{Mathematical Programming}, 111\penalty0 (1):\penalty0 315--348,
  2008.

\bibitem[Zhang(2011)]{zhang2011proportional}
Li~Zhang.
\newblock Proportional response dynamics in the fisher farket.
\newblock \emph{Theoretical Computer Science}, 412\penalty0 (24):\penalty0
  2691--2698, 2011.

\end{thebibliography}

\section{Related Work} 
\label{sec:related_works}

Our paper is related to following lines of research.

\noindent\textbf{Statistical inference for SAA.} 
Asymptotics of sample average function minimizers are well-studied by the stochastic programming community (see, e.g.,~\citet[Chapter 5]{shapiro2021lectures},~\citet{shapiro2003monte} and~\citet{Kim2015}) and the statistics community (see, e.g.,~\citet{van2000asymptotic} and~\citet{newey1994large}). 
Despite the powerful tools developed by researchers, the problem of developing asymptotics for the convex EG program exhibits special challenges. The sample function consists of two parts: a non-smooth part coming from a piecewise linear function, and a locally-strongly convex part, which as we will see comes from the function $x \mapsto -\log x$. The non-smoothness of sample function requires us to investigate sufficient conditions for the second-order differentiability of the expected function and verify several technical regularity conditions, both of which are key hypotheses for most asymptotic normality of minimizers of non-smooth sample function. Moreover, the strong convexity of the sample function requires us to develop sharp finite-sample guarantees that exploit the strong convexity structure.

\noindent\textbf{Applications of Fisher Market Equilibrium}
Fisher market equilibrium is related to a game-theoretic solution concept called pacing equilibrium which is a useful model for online ad auction platforms~\citep{borgs2007dynamics,conitzer2022multiplicative,conitzer2022pacing}.
In addition to ad auction markets,
Fisher market equilibrium model has other usages in the tech industry, such as
the allocation of impressions to content in certain recommender systems~\citep{robust_blog},
robust 
and fair work allocation in content review~\citep{allouah2022robust}.
We refer readers to~\citet{kroer2022market} for a comprehensive review.
Outside the tech industry, Fisher market equilibria also have applications to
scheduling problems~\citep{im2017competitive},
fair course seat allocation~\citep{othman2010finding,budish2016course},
allocating donations to food banks~\citep{aleksandrov2015online},
sharing scarce compute resources~\citep{ghodsi2011dominant,parkes2015beyond,kash2014no,devanur2018new}, 
and allocating blood donations to blood banks~\citep{mcelfresh2020matching}.

The statistical framework developed in this paper provides a guideline to quantify the uncertainty in equilibrium-based allocations in the above-mentioned applications.

\noindent\textbf{Algorithmic Results for Fisher Market}
The problem of equilibrium computation has been of interest in economics for a long time (see, e.g.,~\citet{nisan2007algorithmic}). There is a large literature focusing on computation of equilibrium in Fisher markets through 
combinatorial algorithms (\citet{vazirani_2007,devanur2008market,jain2007polynomial,ye2008path,deng2003complexity}),
convex optimization formulations~\citep{eisenberg1959consensus,eisenberg1961aggregation,shmyrev2009algorithm,cole2017convex}, 
gradient-based methods~\citep{wu2007proportional,zhang2011proportional,aleksandrov2015online,birnbaum2011distributed,nesterov2018computation,gao2020first},
t\^{a}tonnement process-based methods~\citep{borgs2007dynamics,bei2019ascending,cole2008fast,cheung2020tatonnement},
and abstraction methods~\citep{kroer2021computing}. 
Extensions to settings such as quasilinear utilities~\citep{chen2007note, cole2017convex}, limited utilities~\citep{bei2019earning}, indivisible items~\citep{cole2018approximating}, or imperfectly specified utility functions~\citep{caragiannis2016unreasonable,murray2020robust,kroer2019scalable,peysakhovich2019fair} are also available. 
Several works study fair online allocation of divisible goods~\citep{azar2016allocate,sinclair2021fairness,banerjee2022online,liao2022dualaveraging} and indivisible goods~\citep{budish2011combinatorial, othman2016complexity, gorokh2019remarkable} 
by Fisher market equilibrium-based methods.

Most related to our work is~\cite{gao2022infinite}, where the authors extend the classical Fisher market model to a measurable (possibly continuous) item space and shows that infinite-dimensional EG-type convex programs capture ME under this setting. 
This paper proposes a statistical model based on their infinite-dimensional Fisher market and investigate the statistical inference problem.

\noindent\textbf{Statistical Learning and Inference in Equilibrium Models}
\citet{Wager2021,munro2021treatment,sahoo2022policy}
take a mean-field game modeling approach and perform policy learning with a gradient descent method.
In particular, 
\citet{munro2021treatment} study the causal effects of binary intervention on the supply-demand market equilibrium. 
\citet{Wager2021} study the effect of supply-side payments on the market equilibrium. 
\citet{sahoo2022policy} study the learning of capacity-constrained treatment assignment while accounting for strategic behavior of agents.
Different from the above mean-field modeling papers, 
in the linear or quaislinear Fisher market models, equilibrium concept is defined for a finite number of agents, allowing us to avoid a mean-field modeling approach. 
Moreover, the Fisher markets equilibria we study are captured by convex programs, 
so we can leverage well-established tools from the stochastic programming and the empirical processes literature. 
Finally, in Fisher market equilibria, 
a concrete parametric model of demand is imposed as opposed to previous works that take a more or less nonparametric approach, 
and therefore we could obtain results that characterize each agent's behavior, e.g., a central limit result for as market size grows (c.f.\ \cref{thm:clt_beta_u}).

By a different group of researchers, the question of statistical learning and inference 
has been investigated for other equilibrium models, such as
general exchange economy~\citep{guo2021online,liu2022welfare} and
matching markets~\citep{cen2022regret,dai2021learning,liu2021bandit,jagadeesan2021learning,min2022learn}.
Our paper focuses on a specific type of exchange economy called infinite-dimensional Fisher market, which is a model for the long-run market behavior.

Our work is also related to the rich literature of inference under interference~\citep{hudgens2008toward,aronow2017estimating,athey2018exact,leung2020treatment,hu2022average,li2022random}.
In the Fisher market model, the interference among agents is caused by the supply constraint and the utility-maximizing behavior of agents given the price signal. 
In other words, in Fisher markets we put a parametric model on the interference structure which allows us to derive a rich collection of results.
\section{Extended Analytical Properties of the Dual Objective}
\label{sec:analytical_formal}

Let $I(\beta, \theta) = \argmax_i{\beta_i v_i(\theta)}$ be the set of maximizing indices, which could be non-unique. 
We say there is \emph{no tie for item $\theta$ at $\beta$} if $I(\beta, \theta)$ is single-valued, in which case we use $i(\beta,\theta)$ to denote the unique maximizing index. 
Moreover, by Theorem 3.50 from~\citet{beck2017first}, the subgradient $\partial_\beta f(\beta,\theta)$ is the convex hull of the set $\{v_ie_i, i\in I(\beta,\theta) \}$. When $I(\beta,\theta)$ is single-valued, the subgradient set is a singleton, and thus $f$ is differentiable.

Now we have different equivalent ways to describe when $f$ is differentiable:
$f(\cdot, \theta)$ is differentiable at $
\beta$ 
$\Leftrightarrow$ $I(\beta,\theta)$ is single-valued 
$\Leftrightarrow$ $\textsf{bidgap}(\beta,\theta)$ is strictly positive
$\Leftrightarrow$ the sets $\{ \Theta_i(\beta) = \{\theta: \beta_i v_i(\theta) \geq \beta_k v_k (\theta),\forall k\neq i \} \}$ are disjoint.
When $f(\cdot,\theta)$ is differentiable,
we have $G(\beta,\theta) = \nabla_\beta f(\beta,\theta) = e_{i(\beta,\theta)}v_{i{(\beta,\theta)}}$.

\subsubsection*{Markets with stability}

A natural idea is to search for a stronger form of \cref{eq:as:notie} and hope that such a refinement could lead to second-order differentiability. 
In particular, this section is concerned with statement (i) of \cref{thm:second_order_informal}. First we show the condition based on the expectation.

\begin{theorem}
    \label{thm:int_implies_hessian}
    If the following integrability condition holds in a neighborhood of $\betast$
    \begin{align}
        \E \bigg[\frac{1}{\textsf{bidgap}(\beta,\theta)^{}}\bigg] = \int_\Theta \frac{1}{\textsf{bidgap}(\beta,\theta)^{}} \diff S(\theta) < \infty
        \;,
        \tag{\small\textsc{INT}}
        \label{eq:as:UI}
    \end{align}
    then $H$ is twice continuously differentiable at $\betast$. Furthermore, it holds $\nabla\sq \fbar (\betast) = 0 $.  

    Proof in \cref{proof:sec:analytical_properties_of_dual_obj}.
\end{theorem}

By the above theorem, if \cref{eq:as:UI} holds, then the variance matrices in \cref{thm:clt_beta_u} can be simplified as
\begin{align} \label{eq:simplified_variance_beta_u}
    \Sigma_\beta =  \Diag( \{ \Omega_i\sq (\beta^*_i) ^ 4 / ( b_i)\sq \}_{i=1}^n )
    \;,
    \; \Sigma_u = \Diag(\{ \Omega_i\sq \}_i)
    \;.
\end{align}
In this case components of $\betagam$ are asymptotically independent.

We compare the integrability condition in the above theorem with \cref{eq:as:notie}.
Both \cref{eq:as:UI} and \cref{eq:as:notie} can be interpreted as a form of robustness of the market equilibrium. The quantity $\textsf{bidgap}(\beta,\theta)$ measures the advantage the winner of item $\theta$ has over other losing bidders. The larger $\textsf{bidgap}(\beta,\theta)$ is, the more slack there is in terms of perturbing the pacing multiplier before affecting the allocation at $\theta$.
In contrast to \cref{eq:as:notie} which only imposes an item-wise requirement on the winning margin,
the above assumption requires the margin exists in a stronger sense. Concretely, such a moment condition on the margin function $\epsilon$ represents a balance between how small the margin could be and the size of item sets for which there is a small winning margin.

Second we consider the condition based on the essential supremum.
For any buyer $i$ and his winning set $\Theta_i^*$, there exists a positive constant $\eps_i > 0$ such that
        \begin{align} 
            \label{eq:as:win_margin}
            \betasti  v_i(\theta) \geq \max_{k\neq i} \beta^*_k v_k(\theta) + \eps_i \, , \forall \theta \in \Theta_i^*
            \tag{\small{GAP}} 
            \quad \Leftrightarrow \quad
            \esssup_{\theta \in\Theta} 1/\textsf{bidgap}(\beta,\theta) < K < \infty  
        \end{align} 
It requires that the buyer wins the items without tying bids uniformly over the winning item set. 
The existence of a constant $K<\infty$ such that $1/\textsf{bidgap}(\beta,\theta) < K$ for almost all items makes a stronger requirement than \cref{eq:as:UI}.
From a practical perspective, it is also evidently a very strong assumption: for example, it won't occur with many natural continuous valuation functions. Instead, the condition requires the valuation functions to be discontinuous at the points in $\Theta$ where the allocation changes.
Empirically, since $\betagam$ is a good approximation of $\betast$ for a market of sufficiently large size, 
\cref{eq:as:win_margin} can be approximately verified by replacing $\betast$ with $\betagam$.
As a trade-off, \cref{eq:as:UI} is a weaker condition than \cref{eq:as:win_margin} but is harder to verify in practical application.

Below we present two examples where \cref{eq:as:UI} holds.
\begin{example}[Discrete Values]
    Suppose the values are supported on a discrete set, i.e., $[v_1,\dots, v_n]\in  \{V_1,\dots, V_K  \} \subset \Rn$ a.s.\ 
    Suppose there is no tie for each item at $\betast$. Then \cref{eq:as:win_margin} and thus \cref{eq:as:UI} hold.
    $\blacksquare$
\end{example}

\begin{example}[Continuous Values]
    Here we give a numeric example of market with two buyers where \cref{eq:as:UI} holds.
    Suppose the values are uniformly distributed over the sets $A_1 =\{v \in \R^2_+: v_2 \leq 1, v_2 \geq 2v_1 \}$ and $A_2 = \{v \in \R^2_+: v_1 \leq 1, v_2 \leq  \frac12 v_1 \}$. One can verify on $B = \{ \beta\in \R^2: \frac12 \beta_1 < \beta_2 < 2 \beta_1 \}$ \cref{eq:as:UI} holds.
    To further verify this,
    by calculus, we can show the map $\fbar(\beta) = \E[\max\{v_1\beta_1, v_2 \beta_2 \}]$ is 
    \begin{align*}
       \fbar (\beta)
       = 
       \begin{cases}
        \big(\frac{5}{12} - \frac13 \frac{\beta_1}{\beta_2}\big) \beta_1 + \frac{2\beta_2}{3 \beta_1} \beta_2
        &\text{if $ \beta_2 \geq 2\beta_1$}
        \\
        \frac13 (\beta_1 + \beta_2) 
        &\text{if $\beta \in B$, i.e., $\tfrac12  \beta_1 < \beta_2 < 2\beta_1$}
        \\
        \big(\frac{5}{12} - \frac13 \frac{ \beta_2}{\beta_1}\big) \beta_2 + \frac{2\beta_1 }{3 \beta_2}\beta_1 
        &\text{if $\beta_2 \leq \frac12 \beta_1$}
       \end{cases}
       \;.
    \end{align*}
    We see that 
    $\nabla \sq \fbar = 0$ on $B$ which agrees with \cref{thm:int_implies_hessian}.

    However, \cref{eq:as:UI} fails to capture the fact that $\fbar$ is $C^2$ in other regions as well. To see this, note that in the region $\{\beta \in \R^2_{++}: \beta_2 > 2\beta_1\}$, the Hessian is
    \begin{align*}
        \nabla\sq \fbar(\beta)= \begin{bmatrix} \frac{2\,{\beta_{2}}^2}{3\,{\beta_{1}}^3} & -\frac{2\,\beta_{2}}{3\,{\beta_{1}}^2}\\ -\frac{2\,\beta_{2}}{3\,{\beta_{1}}^2} & \frac{2}{3\,\beta_{1}} \end{bmatrix}.
    \end{align*}
    The Hessian on the region $\{ \beta_2 < \frac12 \beta_1 \}$ has a completely symmetric expression by switching $\beta_1$ and $\beta_2$.
    From here we can see the function $\fbar$ is $C^2$ except on the lines $\beta_2=2\beta_1$ and $\beta_2 = \beta_1/2$. 
    Thus, the condition in \cref{eq:as:UI} does not provide the full picture of when twice differentiability holds.
    $\blacksquare$
\end{example}

\subsubsection*{Markets with linear values}
Now we consider the condition (iii) of \cref{thm:second_order_informal}: linear valuations.
To study linear valuations, we adopt the setup in Section~4 from~\citet{gao2022infinite}. 
Suppose the item space is $\Theta = [0,1]$ with supply $s(\theta) = 1$. The valuation of each buyer $i$ is linear and nonnegative: $v_i(\theta) = c_i \theta + d_i\geq 0$. Moreover, assume the valuations are normalized so that $\int_{[0,1]}v_i\diff \theta  = 1 \Leftrightarrow c_i/2 + d_i = 1$.
Assume the intercepts of $v_i$ are ordered such that $2 \geq d_{1}>\cdots>d_{n} \geq 0$. 

We briefly review the structure of equilibrium allocation in this setting. By Lemma~5 from~\citet{gao2022infinite}, there is a unique partition $0 = a^*_0 < a_1^*<\cdots<a_n^*=1$ such that buyer $i$ receives $\Theta_i=\left[a_{i-1}^*, a_i^*\right]$. In words, the item set $[0,1]$ will be partitioned into $n$ segments and assigned to buyers $1$ to $n$ one by one starting from the leftmost segments. Intuitively, buyer 1 values items on the left of the interval more than those on the right, which explains the allocation structure.
Moreover, the equilibrium prices $p^*(\cdot)$ are convex piecewise linear with exactly $n$ linear pieces, corresponding
to intervals that are the pure equilibrium allocations to the buyers.

\begin{theorem}\label{thm:linear_value_implies_hessian}
    In the market set up as above, the dual objective $H$ is $C^2$ at $\betast$.

    Proof in \cref{proof:sec:analytical_properties_of_dual_obj}.

\end{theorem}

The above result also extends to most cases of piecewise linear (PWL) valuations discussed in Section~4.3 of~\citet{gao2022infinite}). 
In the PWL setup there is a partition of $[0,1]$, $A_{0}=0 \leq A_{1} \leq \cdots \leq A_{K-1} \leq A_{K}=1$, such that all $v_{i}(\theta)$'s are linear on $ [A_{k-1}, A_{k}]$. 
At the equilibrium of a market with PWL valuations, we call an item $\theta$ an \emph{allocation breakpoint} if there is a tie, i.e., $I(\betast, \theta)$ is multivalued.
Now suppose the following two conditions hold:
(i) none of the allocation breakpoints coincide with any of the valuation breakpoints $\{A_k\}$, 
and 
(ii) at any allocation breakpoint there are exactly two buyers in a tie.
Under these two conditions, one can show that in a small enough neighborhood of the optimal pacing multiplier $\betast$, 
the allocation breakpoints are differentiable functions of the pacing multiplier.
This in turn implies twice differentiability of the dual objective by repeating the argument in the proof of \cref{thm:linear_value_implies_hessian}. 
However, if either condition (i) or (ii) mentioned above breaks, the dual objective is not twice differentiable.


\subsubsection*{Markets with smoothly distributed values}
Now we consider condition (ii) of \cref{thm:second_order_informal}: smoothing via the expectation operator.
Given that the dual objective $H$ is the expectation of the non-smooth function $f$ (plus a smooth term $\Psi$), we expect that under certain conditions on the expectation operator $H$ will be twice differentiable.
In this section, we make this precise.
First we introduce some extra notations. For each $i\in[n]$, define the map $\sigma_i: \Rnp \to \Rnp$,
$$\sigma_i(v) = [v_1v_i,\,  \dots,\,v_{i-1}v_i,\, v_i,\, v_{i+1}v_i,\, \dots,\, v_nv_i ] \tp $$ for $i\in[n]$, which multiplies all except the $i$-th entry of $v$ by $v_i$. 

\begin{defn}[Regularity]
    Let $f: \Rnp \to \Rp$ be the probability density function (w.r.t.\ the Lebesgue measure) of a positive-valued random vector with finite first moment.
    We say the density $f$ is {regular} if for all $ h_i(v_{-i} )\defeq \int_0^\infty f \big(\sigma_i(v)\big) v_i^n \diff v_i\,$, $ i \in [n]$, it holds
    (i) $h_i$ is continuous on $\R^{n-1}_{++}$, and
    (ii) all lower dimensional density functions of $h_i$ are continuous (treating $h_i$ as a scaled probability density function).  
\end{defn}

\begin{theorem}\label{thm:smooth_density_implies_hessian}
    Let $H$ be differentiable in a neighborhood of $\betast$.
    Assume the random vector $[v_1,\dots, v_n]: \Theta \to \Rnp$ has a distribution absolutely continuous w.r.t.\ the Lebesgue measure on $\Rn$ with density function $f_v$.
    If $f_v$ is regular, then $H$ is twice continuously differentiable on $\Rnpp$.

    Proof in \cref{proof:sec:analytical_properties_of_dual_obj}.
\end{theorem}

The above regularity conditions are easy to verify when the values are i.i.d.\ draws from a distribution. In that case, many smooth distributions supported on the positive reals fall under the umbrella of the described regularity. Below we examine three cases: the truncated Gaussian distribution, the exponential distribution and the uniform distribution.

When values are i.i.d.\ truncated standard Gaussians, the joint density $f(v) =c_1 \prod_{i=1}^n \exp(- v_i\sq / 2)$ and 
$h_i (v_{-i}) 
= c_1  \int_\Rp v_i^n \exp ({- \frac12 v_i\sq ( 1 + \sum_{k \neq i} v_k\sq)}) \diff v_i = c_2 (\sum_{k \neq i} v_k\sq)^{-n/2},$
which are regular. Here $c_i$, $i=1,2$, are appropriate constants.
Similarly, for the i.i.d.\ exponential case with the rate parameter equal to one, the density $f(v) = \prod_{i=1}^n \exp(- v_i)$ and $h_i (v_{-i}) =  (\sum_{k\neq i} v_k)^{-n}$ satisfy the required continuity conditions. 
Finally, suppose the values are i.i.d.\ uniforms on $[0,1]$. The joint density is $f(v) = \prod_{i=1}^n \indi\{0<v_i < 1\}$ and for example, if $i=1$, $h_1(v_{-1}) = (\min\{1, v_2\inv, \dots, v_n\inv \})^{n+1}/(n+1)$, which also satisfies the required continuity conditions.

\section{Variance Estimation for $\beta$ and $u$}
\label{sec:hessian_estimation}

First, we discuss the restrictive case where a stronger form of \cref{eq:as:notie} holds:
$\E [\textsf{bidgap}(\beta, \theta)\inv ]< \infty$ in a neighborhood of $\betast$,
which we recall is a  sufficient condition for twice differentiability (\cref{thm:second_order_informal}).
Note that in the observed market the equilibrium allocation $\xgam$ might not be unique, and for our purpose we let $\xgam$ be any equilibrium allocation.
We construct the following estimator for $\Omega_i$. Let $\ugamtaui \defeq \xtaui \vithetau$ be the utility of buyer $i$ obtained from item $\thetau$. 
Then $\ugami = \sumtau \ugamtaui$.
Under the assumption that $H$ is differentiable at $\betast$ (c.f.\ \cref{thm:first_differentiability}), the equilibrium allocation is unique and pure, i.e., $\xsti = \indi \{ \Theta^*_i \}$.
By rewriting $\Omega_i\sq  = \int \big(v_i \xsti - (\int v_i \xsti \diff S)\big)\sq \diff S $, it is natural to consider the estimator $\hatOmesqi \defeq  \frac1t  \sumtau ( t \ugamtaui - \ugami)\sq$. 

\begin{theorem}\label{thm:estiamtion_Omega}
    If $\E [\textsf{bidgap}(\beta, \theta)\inv ]< \infty$  in a neighborhood of $\betast$, then $\hatOmesqi\toprob  \Omega_i\sq$. Proof in \cref{sec:proof_var_est}.
\end{theorem}

Having derived a consistent estimator for $\Omega_i\sq$, we can construct confidence interval for $\betagam$ and $\ugam$. By \cref{thm:int_implies_hessian} and \cref{eq:simplified_variance_beta_u}, the plug-in type estimators for $\Sigma_\beta$ and $\Sigma_u$ take the form $\hat \Sigma_\beta = \Diag( \{ \hatOmesqi (\betagami)^4 / b_i\sq \})$ and $\hat \Sigma_u = \Diag(\{ \hatOmesqi \})$.

Second, we discuss the case where we have the knowledge of $\{\vithetau\}_{i,\tau}$ under the general assumption that $H$ is $C^2$ at $\betast$. 
Following the discussion in Section~7.3 of~\citet{newey1994large}, we estimate the Hessian matrix by computing numerical difference. 
We choose a smoothing level $\eta_t$ and define 
\begin{align*}
      (\hat \cH) _{ij} & \defeq
       \frac{
       1 
        }{4 \eta_t\sq }\Big(H_t (\betagam + \eta_t(e_i + e_j)) 
        - H_t (\betagam + \eta_t(-e_i + e_j))
        \\
        & \quad \quad - H_t (\betagam + \eta_t(e_i - e_j)) 
        + H_t (\betagam + \eta_t(- e_i - e_j))  \Big)
        \;,
\end{align*}
which serves as an estimator of the $(i,j)$-th entry of the Hessian $\cH = \nabla \sq H(\betast)$. By Theorem 7.4 from~\citet{newey1994large}, one can show that if $\eta_t \to 0$ and $ \sqrt t\eta_t \to \infty$, then the estimator is consistent, i.e, $\hat \cH \toprob \cH$. 
Note in order to compute the value of $H_t$ at a perturbed $\betagam$ we need access to the values of the buyers.
Since those values may not always be available in practice, this estimator is not as practical as our other estimators which rely purely on equilibrium quantities.
Moreover, the estimator requires tuning of the smoothing parameter $\eta_t$.

\section{Further Properties of EG programs}
\label{sec:opt_EG_programs}

\begin{fact}\label{fact:pop_eg}
    Both optima in \cref{eq:pop_eg,eq:pop_deg} are attained.
    Let $(x^*_{\EG}, u^*_{\EG})$ and $\betast$ attain the optima the EG programs \cref{eq:pop_eg,eq:pop_deg}, respectively.
    
        \begin{itemize}
            \setlength\itemsep{-.5em}

            \item First-order conditions. Given $(x_{\EG}, u_{\EG})$ feasible to \eqref{eq:pop_eg} and $\beta$ feasible to \eqref{eq:pop_deg}, they are both optimal if and only if the following KKT conditions hold: 
            (i) $\langle p_{\EG}, s - \sum_i x_{\EG,i} \rangle = 0$ where $p_{\EG} = \max_i \betai  v_i$, 
            (ii) $ \langle p_{\EG} - \beta_i v_i, x_{\EG,i} \rangle = 0$, and
            (iii) $\langle v_i, x_{\EG,i} \rangle = u_{\EG,i} = b_i / \beta_i$. 
            
            \item Uniqueness. The equilibrium utility and prices are unique. The optimal solution $\betast$ to \cref{eq:pop_deg} is unique.
            
            \item Strong duality.
            $ \sumiton b_i \log u^*_{\EG,i} = H(\betast ) + \sumiton b_i(\log b_i - 1).$
            
            \item Equilibrium.
            Given any optimal solutions $(x^*_{\EG}, u^*_{\EG}, \betast)$ to \cref{eq:pop_eg,eq:pop_deg}, let $p^*_{\EG}(\cdot) = \max_i \betasti v_i(\cdot)$. Then $(x^*_{\EG}, u^*_{\EG}, p^*_{\EG})$   
            is a ME.
            Conversely, for a ME $(x^*, u^*, p^*)$, it holds that (i) $(x^*,u^*)$ is an optimal solution of \eqref{eq:pop_eg} and (ii) $\beta^*_{{\ME}}:= b_i / \langle v_i, x^*_i \rangle$ is the optimal solution of \eqref{eq:pop_deg}.

            \item Bounds on $\betast$. Define $\ubarbetai \defeq  b_i/ \int v_i \diff S$ 
            and $\bar \beta \defeq \sumiton b_i /\min_i \{\int v_i \diff S\}$. Then $\ubarbetai \leq \betasti \leq \bar \beta$.
        \end{itemize}
    \end{fact}

    Based on the set of KKT conditions we comment on the structure of market equilibrium. 
    Condition~(ii) describes how the pacing multiplier relates to equilibrium allocation; buyer $i$ only receives items within its `winning set' $\{\theta: p^*(\theta)= \beta^*_i v_i(\theta) \}$. 
    This also hints at a connection between Fisher market and first-price auction: the equilibrium allocation can be thought of as the result of a first-price auction where each buyer bids $\betasti v_i(\theta)$ and then item goes to the highest bidder (with appropriate tie-break).
    Condition~(iii) shows pacing multipliers $\betast$ can be interpreted as price-per-utility. 
    Finally, all budgets are extracted, i.e., 
    $\lg \pst, s \rg = \sumiton b_i$.
    To see this, we apply all three KKT conditions and obtain $\lg \pst, s \rg = \lg \pst, \sumiton \xsti \rg = \sumiton \betasti \lg v_i,\xsti\rg = \sumiton \betasti (b_i/\betasti) = \sumiton b_i$.
    Intuitively, this is due to the fact that buyers only receives utilities from obtaining goods but not retaining money. In \cref{sec:quasilinear} we study an extension called \emph{quasilinear market} where buyers have the incentive to retain money.
    
\begin{fact} \label{fact:sample_eg}

Parallel to the population EG programs, we state the optimality conditions for the sample EG programs.
Consider the sample EG programs \cref{eq:sample_eg,eq:sample_deg}.
It holds 

    \begin{itemize}
     
   \item First-order conditions. Given $(x^\gam_{\EG}, u^\gam_{\EG})$ feasible to \eqref{eq:sample_eg} and $\beta^\gam$ feasible to \eqref{eq:sample_deg}, they are both optimal if and only if the following KKT conditions hold: 
        (i) $\langle p^\gam_{\EG}, \mathsf{s} - \sum_i x^\gam_{\EG,i} \rangle = 0$ where $p^\gam_{\EG} = \max_i \betagami v_i$, 
        (ii) $ \langle p^\gam_{\EG} - \beta^\gam_i v_i, x^\gam_{\EG,i} \rangle = 0$,
        (iii) and $\langle v_i, x^\gam_{\EG,i} \rangle = u^\gam_{\EG,i} = b_i / \beta^\gam_i$. 

        \item Strong duality.
        $ \sumiton b_i \log u^\gam_{\EG,i} = H_t(\betagam ) + \sumiton b_i(\log b_i - 1).$

        \item Uniqueness. The equilibrium utility and prices are unique. The optimal solution $\betagam$ to \cref{eq:sample_deg} is unique.

        \item Equilibrium.
        Any optimal solutions $(x^\gam_{\EG}, u^\gam_{\EG}, p^\gam_{\EG})$ to \cref{eq:sample_eg,eq:sample_deg}   
        is a ME.
        Conversely, for a ME $(x^\gam, u^\gam, p^\gam)$, it holds that (i) $(x^\gam,u^\gam)$ is an optimal solution of \eqref{eq:sample_eg} and (ii) $\beta^\gam_{{\ME}}:= b_i / \langle v_i, x^\gam_i \rangle$ is the optimal solution of \eqref{eq:sample_deg}.

        \item Bounds on $\betagam$ and $\ugam$. It holds 
        $ \frac{b_i}{ \sumiton b_i } \sumtau \mathsf{s}^\tau \vithetau \leq \ugami \leq\sumtau \mathsf{s}^\tau \vithetau  $ and 
        $\frac{b_i}{\sumtau \mathsf{s}^\tau \vithetau}\leq \betagami \leq \frac{\sumi b_i }{  \sumtau \mathsf{s}^\tau\vithetau }$ 
    \end{itemize}
\end{fact}

We comment on the scaling $1/t$ in the observed market $\calMEgam(b,v, \frac1t 1_t)$.
Recall that in the long-run market the total supply of items is one, while each buyer $i$ has budget $b_i$. To match budget sizes and markets sizes, we require that in the observed market, the ratio between total supply and buyer $i$'s budget is also $1:b_i$.
Note that in a linear Fisher market, we can scale all budgets by any positive constant and the equilibrium does not change, except that prices are scaled by the same amount. 
Formally, if $\calMEgam(b,v,\sfs)=(x,u,p)$, then $\calMEgam (\delta b, (\alpha_1 v_1,\dots, \alpha_n v_n),\beta \sfs) = (x,(\alpha_1\beta u_1, \dots,\alpha_n\beta u_n), \delta p)$ for positive scalars $\delta,\beta,\{\alpha_i\}$.
Thus, our particular choice of supply normalization is not crucial. 
For example, we could equivalently work with the market $\calMEgam(tb, v, 1_t)$, with trivial scaling adjustments in derivation of the results.

We remark that there are two ways to specify the valuation component 
of infinite-dimensional Fisher market.
The first one is simply imposing functional form assumptions on $v_i(\cdot)$. For example, one could take $\Theta=[0,1]$ and let $v_i(\cdot)$ be a linear function or a piecewise linear function~\cite[Section 4]{gao2022infinite}.
If we view $v = (v_1,\dots,v_n):\Theta\to\Rnp$ as a random vector, then an alternative way is to simply specify the distribution of $v$.
Formally, let $v$ and $v'$ be identically distributed random vectors representing the values of buyers.
By the form of EG programs \cref{eq:pop_eg,eq:pop_deg},
if $(u^*,\betast)$ are equilibrium utilities and pacing multipliers of the market $\calME(b,v,s)$ and $( u^{*'}, \beta^{*  '}  )$ are those of  $\calME(b,v',s)$, then $(u^*,\beta^*)=  ( u^{*'}, \beta^{* '})$. 
Even though the equilibrium allocations and prices are different in the two markets,
the quantities we care about, e.g., 
individual utilities and  
Nash social welfare are the same.
Moreover, applying the same reasoning to the case of quasilinear market (see~\cref{sec:quasilinear}), it will be clear that identical value distributions implies identical revenues in market equilibrium.
When the distribution of $v$ is absolutely continuous w.r.t.\ the Lebesgue measure on $\Rn$ we use $f_v$ to denote the density function.

\section{Technical Lemmas}
\label{sec:tech_lemmas}

\begin{proof}[Proof of \cref{lm:value_concentration}]
    Recall the event $A_t = \{ \betagam \in C   \}$.
Define $\vbarit = \frac1t \sumtau \vithetau$.

        First we notice concentration of values implies membership of $\betagam$ to $C$, i.e.,
        $\{ 1/2 \leq \vbarit \leq 2,\forall i \} \subset \{\betagam \in C \}$ due to 
        \cref{fact:sample_eg}. Concretely, 
        $\ugami \leq \frac1t \sumtau\vithetau$ and $\ugami \geq \frac1t\frac{b_i}{\sumiton b_i} \sumtau \vithetau$, and through the equation $\betagami = b_i / \ugami$ the inclusion follows.
        Note $0\leq \vithetau\leq \vbar$ is a bounded random variable with mean $\E[\vithetau] =1 $. By Hoeffding's inequality we have
        $   
            \P(|\vbarit - 1| \geq \delta)\leq 2\exp(-\frac{2\delta\sq t}{\vbar\sq})
        $.
        Next we use a union bound and obtain
        \begin{align}
            \label{eq:hoeffding_average_values}
            \P(\betagam \notin C) \leq \P \bigg(
            \bigcup_{i=1}^n \big\{ |\vbarit - 1| \geq \delta \big\} \bigg) \leq 2n \exp \Big(-\frac{2\delta\sq t}{\vbar\sq}\Big)
            \;.
        \end{align}

        By setting $2n \exp(-\frac{2\delta\sq t}{\vbar\sq}) = \eta$ and $\delta = 1/2$ and solving for $t$ we obtain item (i) in claim.
    
        To show item (ii),
        we use the Borel-Cantelli lemma. By choosing $\delta  = 1/2$ in the \cref{eq:hoeffding_average_values} we know $\P(A_t^c) \leq \P(\{ 1/2 \leq \vbarit \leq 2,\forall i \}^c)\leq 2n\exp(-t/(2\vbarsq))$. Then we have
        \begin{align*}
            \sum_{t=1}^\infty \P(A_t^c) < \infty
            \;.
        \end{align*}
        By the Borel-Cantelli lemma it follows that $\P(\{A_t^c \text{ infinitely often}\}) = 0$, or equivalently $\P(A_t \text{ eventually}) = 1$.
    \end{proof}
    
\begin{lemma}[Smoothness and Curvature] \label{lm:smooth_curvature}
    It holds that 
both $H$ and $H_t$ are $L$-Lipschitz and $\lambda$-strongly convex w.r.t\ the $\ell_\infty$-norm on $C$ with $L = 2n+\vbar$ and $\lambda = \ubar{b}/4$. Moreover, $H_t$ and $H$ are $(\vbar + 2\sqrt{n})$-Lipschitz w.r.t.\ $\ell_2$-norm.
\end{lemma}
\begin{proof}[Proof of \cref{lm:smooth_curvature}]

    Now we verify that $H_t$ and $H$ are $(\vbar + 2n)$-Lipschitz on the compact set $C$ w.r.t.\ the $\ell_\infty$-norm. For $\beta,\beta' \in C$,
    \begin{align*}
        & |H_t(\beta)  - H_t(\beta')| 
        \\
        & \leq \frac1t\sumtau \big|\max_i \{\vithetau \beta_i \} - \max_i \{\vithetau \beta_i' \}\big| + 
        \sumiton b_i \big| \log \beta_i -  \log \beta_i'\big|
        \\
        & \leq \vbar \| \beta - \beta'\|_\infty + \sumiton b_i \cdot \frac{1}{\ubarbetai/2} |\beta_i - \beta_i'|
        \\
        & = (\vbar + 2n) \| \beta - \beta'\|_\infty
        \;.
    \end{align*}
    This concludes the $(\vbar + 2n)$-Lipschitzness of $H_t$ on $C$. Similar argument goes through for $H$. 
    From the above reasoning we can also conclude $|H_t(\beta)  - H_t(\beta')|\leq \vbar \|\beta-\beta'\|_2 + 2 \|\beta - \beta'\|_1 \leq (\vbar + 2\sqrt{n})\|\beta - \beta'\|_2$. This concludes $(\vbar + 2\sqrt n)$-Lipschitzness of $H_t$ w.r.t.\ $\ell_2$-norm.

    Recall $H = \fbar + \Psi$ where $\fbar(\beta) = \E[\max_i\{ v_i(\theta) \betai \}]$ and $\Psi(\beta) = -\sumiton b_i \log \beta_i$. The function $\Psi$ is smooth with the first two derivatives 
    \begin{align*}
        \nabla\Psi(\beta) =  - [b_1/\beta_1, \dots, b_n/\beta_n]\tp, \quad \nabla\sq \Psi(\beta) =  \Diag( \{ b_i/(\betai)\sq \})
        \;.
    \end{align*}
    It is clear that for all $\beta \in C$ it holds $\betai\leq 2$. So 
    $\nabla\sq \Psi(\beta) \succ \min_i\{b_i/4\} I = \lambda I$.
    To verify the storng-convexity w.r.t\ $\|\cdot \|_\infty$ norm, we note
    for all $\beta', \beta \in C$.
    \begin{align*}
        H(\beta ')  - H(\beta) - \lg z + \nabla \Psi(\beta), \beta' - \beta  \rg \geq (\lambda/2) \|\beta' - \beta \|_2\sq \geq (\lambda/2) \|\beta' - \beta \|_\infty\sq
        \;,
    \end{align*}
    where $z \in \partial \fbar(\beta)$ and $z + \nabla \Psi(\beta) \in \partial H(\beta)$.
    This completes the proof.
\end{proof}

\begin{defn}[Definition 7.29 in~\citet{shapiro2021lectures}]
  \label{def:epiconvergence}
    A sequence $f_k:\Rn \to \bar \R$, $k=1,\dots$, of extended real valued functions epi-converge to a function $f:\Rn \to \bar \R$, if for any point $x\in \Rn$ the following conditions hold

    (1) For any sequence $x_k \to x$, it holds $\liminf _{k \rightarrow \infty} f_{k}\left(x_{k}\right) \geq f(x)$,

    (2) There exists a sequence $x_k \to x$ such that $\limsup _{k \rightarrow \infty} f_{k}\left(x_{k}\right) \leq f(x)$.
\end{defn}



\begin{defn} [Definition 3.25, \citet{rockafellar2009variational}, see also Definition 11.11 and Proposition 14.16 from \citet{bauschke2011convex}]
    A function $f:\Rn \to \bar \R$ is level-coercive if $\liminf_{\|x\|\to \infty} f(x) / \|x\| > 0$. It is equivalent to $\lim _{\|x\| \rightarrow+\infty} f(x)=+\infty$.
\end{defn}
\begin{lemma}[Corollary 11.13, \citet{rockafellar2009variational}]
    \label{lm:coercivity_level_boundedness}
    For any proper, lsc function $f$ on $\Rn$, level coercivity implies level boundedness. When $f$ is convex the two
properties are equivalent.
\end{lemma}

\begin{lemma}[Theorem 7.17,~\cite{rockafellar2009variational}] \label{lm:def_of_epiconv}
    Let $h_n: \R^d \to \bar \R$, $h:\R^d\to \bar \R$ be closed convex and proper. Then $h_n \toepi h$ is equivalent to either of the following conditions.

    (1) There exists a dense set $A \subset \R^d$ such that $h_n(v) \to h(v)$ for all $v\in A$.

    (2) For all compact $C \subset \dom h$ not containing a boundary point of $\dom h$, it holds 
    $$\lim_{n\to \infty} \sup_{v\in C} |h_n(v) - h(v)| = 0
    \;.
    $$
\end{lemma}

\begin{lemma}[Proposition 7.33,~\citet{rockafellar2009variational}] \label{lm:inf_conv}
    Let $h_n: \R^d \to \bar \R$, $h:\R^d\to \bar \R$ be closed and proper. If $h_n$ has bounded sublevel sets and $h_n \toepi h$, then $\inf_v h_n(v) \to \inf_v h(v)$.
\end{lemma}

\begin{lemma}[Theorem 7.31,~\citet{rockafellar2009variational}] \label{lm:appro_sol_conv}
    Let $h_n: \R^d \to \bar \R$, $h:\R^d\to \bar \R$ satisfy $h_n \toepi$ and $-\infty < \inf h < \infty$. Let $S_{n}(\varepsilon)=\left\{\theta \mid h_{n}(\theta) \leq \inf h_{n}+\varepsilon\right\}$ and $S(\varepsilon)=\left\{\theta \mid h(\theta) \leq \inf h+\varepsilon\right\}$. Then $\limsup_n S_n(\varepsilon) \subset S(\varepsilon)$ for all $\varepsilon \geq 0$, and $\limsup_n S_n(\varepsilon_n) \subset S(0)$ whenever $\varepsilon_n \downarrow 0$.

\end{lemma}

\begin{lemma}[Theorem 5.7,~\citet{shapiro2021lectures}, Asymptotics of SAA Optimal Value] \label{lm:clt_optimal_value}
    Consider the problem $$\min_{x\in X}  f(x) = \E[F(x,\xi)]$$ where $X$ is a nonempty closed subset of $\R^n$, $\xi$ is a random vector with probability distribution
    $P$ on a set $\Xi$ and $F:X\times \Xi \to \R$. Assume the expectation is well-defined, i.e., $f(x)<\infty$ for all $x\in X$. Define the sample average approxiamtion (SAA) problem $$\min_{x\in X}  f_N(x) = \frac1N \sum_{i=1}^N F(x,\xi_i)$$ where $\xi_i$ are i.i.d.\ copies of the random vector $\xi$. Let $v_N$ (resp., $v^*$) be the optimal value of the SAA problem (resp., the original problem). Assume the following.
\begin{enumconditions}
    \item The set $X$ is compact. \label{it:lm:clt_optimal_value:1}
    \item For some point $x\in X$ the expectation $\E[F(x,\xi)\sq]$ is finite. \label{it:lm:clt_optimal_value:2}
    \item There is a measurable function $C:\Xi \to \R_+$
    such that $\E[C(\xi)\sq] < \infty$ and $\left|F(x, \xi)-F\left(x^{\prime}, \xi\right)\right| \leq C(\xi)\left\|x-x^{\prime}\right\|$ for all $x, x' \in X$ and almost every $\xi \in \Xi$.
    \label{it:lm:clt_optimal_value:3}
    \item The function $f$ has a unique minimizer $x^*$ on $X$. \label{it:lm:clt_optimal_value:4}
\end{enumconditions}
Then $${v}_{N}=  {f}_{N}(x^*)+o_{p}(N^{-1 / 2})
\;, 
\quad \sqrt{N} ( v_N - v^*) \tod N\big(0, \var
\big({F(x^*, \xi)}\big) \big)
\;.
$$ 
\end{lemma}

\section{Proof of Theorem~\ref{thm:consistency}} 
\label{sec:proof:thm:consistency}
\begin{proof}[Proof of \cref{thm:consistency}]

    We show epi-convergence (see \cref{def:epiconvergence}) of $H_t$ to $H$.
    Epi-convergence is closely related to the question of whether we have convergence of the set of minimizers. In particular, epi-convergence is a suitable notion of convergence under which one can guarantee that the set of minimizers of the sequence of approximate optimization problems converges to the minimizers of the original problem.

    To work under the framework of epi-convergence, we extend the definition of $H_t$ and $H$ to the entire Euclidean space as follows. 
    We extend $\log$ to the entire real by defining $\log(x) = -\infty$ if $x<0$.
    Let
    \begin{align*}
        \tilde F(\beta,\theta) = 
        \begin{cases}
            F(\beta,\theta) = \max_i v_i(\theta)\betai - \sumiton b_i \log \betai & 
            \text{if $\beta \in \Rnpp$}
            \\
            + \infty & \text{else}
        \end{cases},
    \end{align*}
    and
    \begin{align*}
            \tilde H(\beta): \R^n \to \bar \R,  \beta \mapsto \begin{cases} H(\beta) & \text{if $\beta \in \Rnpp$} \\
            +\infty & \text{else} \end{cases},
            \quad 
            \tilde H_t(\beta): \R^n \to \bar \R,  \beta \mapsto \begin{cases} H_t(\beta) & \text{if $\beta \in \Rnpp$} \\
            +\infty & \text{else} \end{cases}.
    \end{align*}
    It is clear that for $\beta \in \Rn$ it holds $\tilde H(\beta) = \E[\tilde F(\beta,\theta)]$ and $\tilde H_t(\beta) = \frac1t \sumtau \tilde F(\beta,\thetau)$.
    In order to prove the result, we will invoke Lemmas~\ref{lm:def_of_epiconv}, \ref{lm:inf_conv}, and \ref{lm:appro_sol_conv}.
    To invoke those lemmas, we will need the following four properties that we each prove immediately after stating them.

    \begin{enumerate}
        \item {Check that $ \tilde H$ is closed, proper and convex, and $\tilde H_t$ is closed, proper and convex almost surely.} Convexity and properness of the functions $ \tilde H_t$ and $\tilde H$ is obvious.
        Recall for a proper convex function, closedness is equivalent to lower semicontinuity \citep[Page 52]{rockafellar1970convex}. 
        It is obvious that $\tilde H_t$ is continuous and thus closed almost surely. 

        It remains to verify lower semicontinuity of $\tilde H$,
        i.e., for all $\beta \in \Rn$, $\liminf_{\beta' \to \beta}\tilde H(\beta') \geq \tilde H(\beta)$.
        For any $\beta \in \Rn$, we have that 
        $\tilde f(\beta,\theta) \defeq \max_i \vithe \beta_i  + \delta_{\Rnp}(\beta)\geq 0$, where $\delta_A(\beta) = \infty $ if $\beta\notin A$ and $0$ if $\beta\in A$. 
        With this definition of $\tilde f$ we have 
        $\tilde F(\beta,\theta) =\tilde f(\beta,\theta)  - \sumiton b_i \log \beta_i$.
        Applying Fatou's lemma (for extended real-valued random variables), we get 
        $\liminf_{\beta'\to\beta} \E[\tilde f(\beta',\theta) ]\geq \E[\liminf_{\beta' \to \beta} \tilde f(\beta',\theta) ] \geq  \E[\tilde f(\beta,\theta) ]$ where in the last step we used lower semicontinuity of $\beta \mapsto \tilde f(\beta,\theta) $.
        And thus 
        \begin{align*}
            & \liminf_{\beta'\to\beta} \tilde H(\beta') 
            \\
            & = \liminf_{\beta'\to\beta}\E \bigg[\tilde f(\beta',\theta)   - \sumiton b_i \log \beta'_i \bigg] 
            \\
            & \geq \liminf_{\beta'\to\beta}\E[ \tilde f(\beta',\theta)  ] - \sumiton b_i \log \betai 
            \\
            & \geq  \E \bigg[ \tilde f(\beta,\theta)   - \sumiton b_i \log \betai \bigg]
            \\
            & =  \tilde H(\beta) 
        \end{align*}
        This shows $\tilde H$ is lower semicontinuous.
        \label{it:consistency:1}
        
        \item {Check $\tilde H_t$ pointwise converges to $\tilde H$ on $\Q^n$.}
        Let $\Q^n$ be the set of $n$-dimensional vectors with rational entries.
        For a fixed $\beta \in \Rn$, define the event $E_\beta \defeq\{   \lim_{t\to \infty} \tilde H_t(\beta)  = \tilde H (\beta)\}$.
        Since $\vithetau \leq \vbar$ almost surely by assumption, the strong law of large numbers implies that $\P( E_\beta) = 1$. Define 
        \begin{align*}
            E\defeq \Big\{ \lim_{t\to \infty} \tilde H_t(\beta)  = \tilde H (\beta), \text{ for all } \beta \in \Q^n \Big\} = \bigcap_{\beta \in \Q^n} E_\beta.
        \end{align*}
        Then by a union bound we obtain $\P(E^c) = \P(\bigcup_{\beta \in \Q^n} E_\beta^c) \leq \sum_{\beta \in \Q^n} \P(E_\beta^c)= 0$, implying $E$
        has measure one.
        \label{it:consistency:2}

        \item {Check $-\infty < \inf_\beta \tilde H < \infty$.} This is obviously true since valuations are bounded.
        \label{it:consistency:3}

        \item {Check that for almost every sample path $\omega$, $\tilde H_t$ has bounded sublevel sets (eventually).} 
        By \cref{lm:coercivity_level_boundedness}, this property is equivalent to eventual coerciveness of $\tilde H_t$, i.e., there is a (random) $N$ such that for all $t \geq N$, it holds $\lim_{\|\beta \| \to \infty} \tilde H_t (\beta) = +\infty $.
        By \cref{lm:value_concentration}, we know for almost every $\omega$, there is a finite constant $N_\omega$ such that for all $t\geq N_\omega$ it holds $\vbarit \geq 1/2$. Then it holds for this $\omega$, all $t\geq N_\omega$, and all $\beta\in \Rn$, 
        \$
            \tilde H_t(\beta) 
            & = 
            \frac1t \sumtau \max_i \vithetau  \beta_i -\sumiton b_i\log \beta_i 
            \\
            & \geq \max_i (\vbarit \beta_i)-\sumiton b_i\log \beta_i  
            \\
            & \geq \frac12 \|\beta\|_\infty -\sumiton b_i\log \beta_i  \to +\infty \quad \text{as $\|\beta\|\to \infty$}
            \;.
        \$
        This implies $\tilde H_t$ has bounded sublevel sets.
        \label{it:consistency:4}

    \end{enumerate}

    With the above \cref{it:consistency:1} and \cref{it:consistency:2}  we invoke \cref{lm:def_of_epiconv} and obtain that 
    \#\P\big( \tilde H_t(\beta) \toepi \tilde H(\beta) \big) = 1
    \;,
    \label{eq:epi_conv}\# 
    and that the convergence is uniform on any compact set.

    The epi-convergence result \cref{eq:epi_conv} along with \cref{it:consistency:4} allows us to invoke \cref{lm:inf_conv} and obtain 
    \#\inf_{\beta \in \Rn} \tilde H_t(\beta) \to \inf_{\beta \in \Rn} \tilde H(\beta) \text{ a.s.} \label{eq:min_conv}\#
    which also implies $\inf_{\Rnpp} H_t \to \inf_\Rnpp H$ a.s.

    With the epi-convergence result \cref{eq:epi_conv} along with \cref{it:consistency:3} we invoke \cref{lm:appro_sol_conv} and obtain 
    \begin{equation} \label{eq:containment}
        \begin{split} 
            \limsup _{t} \cB^\gam(\epsilon) \subset \cB^*(\epsilon) \text { for all } \epsilon \geq 0
            \;,\\
        \limsup _{t} \cB^\gam(\epsilon_t) \subset \cB^*(0) \text { for all } \epsilon_t \downarrow 0
        \;.
        \end{split}
    \end{equation}

    \underline{Putting together.}
    At this stage all statements in the theorem are direct implications of the above results. 
    
    \emph{Proof of \cref{it:thm:consistency:1}}
    
    Convergence of Nash social welfare follows from \cref{eq:min_conv} and strong duality, i.e., $\LNSW^\gam = \inf_{\beta \in \Rnpp } H_t(\beta) + \sumiton (b_i\log b_i - b_i)$ and $\LNSW^* = \inf_{\beta \in \Rnpp } H(\beta)+\sumiton (b_i\log b_i - b_i)$. 

    \emph{Proof of \cref{it:thm:consistency:2}}
    
    Now we show consistency of the pacing multiplier via 
    \cref{lm:def_of_epiconv} and \cref{lm:inf_conv}. 
    Recall the compact set $C = \prod_{i=1}^n  [\ubarbetai/2, 2\betabar] = \prod_{i=1}^n [b_i/2, 2] \subset \R^n$. 
    By construction, $\betast \in C$. 
    First note that for almost every sample path~$\omega$, $1/2 \leq \vbarit \leq 2$ eventually, and thus $\beta^\gamma_i = b_i / u^\gam_i \leq b_i / (b_i \vbarit ) \leq 2$ and $\beta^\gam_i \geq b_i/2$ eventually. 
    So $\betagam \in C$ eventually.
    Now we can invoke \cref{lm:def_of_epiconv} Item (2) to get 
    \begin{align}
        \lim_{t\to\infty}\sup_{\beta \in C} | H_t(\beta) - H(\beta)| \to 1 \quad\text{a.s.}
    \label{eq:as conv ht h}
    \end{align}
    Now we can show that the value of $H$ on the sequence $\beta^\gam$ converges to the value at $\beta^*$:
    \$ 0\leq \lim_{t \to \infty}  H(\beta^\gam) - H(\betast) = \lim_{t \to \infty} [H (\beta^\gam) - H_t(\beta^\gam)] + \lim_{t \to \infty} [H_t(\beta^\gam) - H(\betast)] = 0
    \;.
    \$
    Here the first term tends to zero due to \eqref{eq:as conv ht h},
    and the second term by \cref{eq:min_conv}. 
    For any limit point of the sequence $\{\beta^\gam\}_t$, $\beta^\infty$, by 
    lower semicontinuity of $H$,
    \$ 0 &\leq H(\beta^\infty) - H(\betast) \leq \liminf_\ttinf H(\beta^\gam) - H(\betast) = 0
    \;.
    \$
    So it holds that $H(\beta^\infty) = H(\betast)$ for all limit points $\beta^\infty$. By uniqueness of the optimal solution $\betast$ (see \cref{fact:pop_eg}), we have $\beta^\gam \to \betast$ a.s. 
    
    \emph{Proof of \cref{it:thm:consistency:4}}
    
    Convergence of approximate equilibrium follows from \cref{eq:containment}.
    \end{proof}

\section{Proof of Theorem~\ref{thm:lnsw_concentration}}
\label{sec:proof:thm:lnsw_concentration}
\begin{proof}[Proof of \cref{thm:lnsw_concentration}]

    Recall 
    the set $C 
    = \prod_{i=1}^n  [\ubarbetai/2, 2\betabar]
    = \prod_{i=1}^n [b_i/2, 2] \subset \R^n$ and
    the event $A_t = \{\betagam \in C \}$. 
    By \cref{lm:value_concentration} we know that if $t \geq 2 {\vbarsq {\log(4n/\eta)}}$ then 
    event $A_t$ happens
    with probability $\geq 1-\eta/2$. 
    Now the proof proceeds in two steps.

    \paragraph{Step 1. A covering number argument. }
    Let $\cB^o$ be an $\epsilon$-covering of the compact set $C$, i.e, for all $\beta \in C$ there is a $\beta^o(\beta) \in \cB^o$ such that $\| \beta - \beta^o(\beta) \|_\infty \leq \epsilon$. It is easy to see that such a set can be chosen with cardinality bounded by $|\cB^o|\leq (2/\epsilon)^n$. 

    Recall $H_t$ and $H$ are $L$-Lipschitz w.r.t.\ $\ell_\infty$-norm on $C$.
    Using this fact we get the following uniform concentration bound over the compact set $C$. 
    \begin{align*}
        & \sup_{\beta \in C} |H_t (\beta) - H(\beta)|
        \\
        & \leq \sup_{\beta \in C} \big\{  |H_t(\beta) - H_t(\beta^o(\beta))|
        + |H(\beta) - H(\beta^o(\beta))|
        + |H_t(\beta^o(\beta)) - H(\beta^o(\beta))|\big\}
        \\
        & \leq 2(\vbar + 2n) \epsilon + \sup_{\beta^o \in \cB^o}|H_t(\beta^o) - H(\beta^o)|
        \;.
    \end{align*}

    Next we bound the second term in the last expression. For some fixed $\beta \in C$, let $X^\tau \defeq \max_i \vithetau \beta_i$ and let its mean be $\mu$. Note $0 \leq X^\tau  \leq \vbar \|\beta\|_\infty \leq 2\vbar$ due to $\beta \in C$. So $X^\tau$'s are bounded random variables. By Hoeffding's inequality we have
    \begin{align*}
        \P\big(|H_t(\beta) - H(\beta)| \geq \delta\big)
        = \P\bigg( \Big|\frac1t \sumtau X^\tau - \mu\Big| \geq \delta\bigg)
        \leq 2\exp\bigg(-\frac{\delta\sq t }{2\vbarsq}\bigg)
        \;.
    \end{align*}
    By a union bound we get
    \begin{align*}
        \P\bigg(\sup_{\beta^o \in \cB^o}|H_t(\beta^o) - H(\beta^o)| \geq \delta\bigg)
        \leq 2|\cB^o| \exp\bigg(-\frac{\delta\sq t }{2\vbarsq}\bigg) \leq 2\exp\bigg(-\frac{\delta \sq t }{2\vbarsq} + n\log(2/\epsilon)\bigg)
        \;.
    \end{align*}
    Define the event 
    \begin{align}
        \label{eq:def:eventEt}
    E_t \defeq \Big\{\sup_{\beta^o \in \cB^o}|H_t(\beta^o) - H(\beta^o)| 
    \leq \frac{2\vbar}{\sqrt{t}} \sqrt{\log(4/\eta) + n\log(2/\epsilon)} 
    =: \iota \Big\} 
    \;.
    \end{align}
    By setting $2\exp(-{\delta \sq t }/{(2\vbarsq)} + n\log(2/\epsilon)) = \eta/2$ and solving for $\eta$, we have that $\P(E_t) \geq 1-\eta/2$.

    \paragraph{Step 2. Putting together. }
    Recall the event $A_t = \{ \betagam \in C   \}$.
    Now let events $A_t$ and $E_t$ hold. Note $\P(A_t \cap E_t) \geq 1-\eta$ if $t \geq 2 {\vbarsq {\log(4n/\eta)}}$. Then 
    \begin{align}
        & \Big| \sup_{\beta\in\Rnpp} H_t(\beta) - \sup_{\beta\in\Rnpp} H(\beta) \Big| 
        \notag
        \\
        & = \Big| \sup_{\beta\in C} H_t(\beta) - \sup_{\beta\in C} H(\beta)\Big|  
        \notag
        \\
        & \leq \sup_{\beta \in C} |H_t(\beta)-H(\beta)| 
        \notag
        \\
        & \leq 2(\vbar + 2n) \epsilon + \iota 
        \;,
        \label{eq:08251314}
    \end{align}
    where the first equality is due to event $A_t$ and the last inequality is due to event $E_t$ defined in \cref{eq:def:eventEt}.
    Now we choose the discretization error as $\epsilon = \frac{1}{\sqrt{t} (\vbar + 2n)}$. 
    Then, the expression in \cref{eq:08251314} can be upper bounded as follows.
    \begin{align*}
        & 2(\vbar + 2n) \epsilon + \iota 
        \\
        & = \frac{2}{\sqrt{t}} + 
        \frac{2\vbar}{\sqrt t}  \sqrt{\log(4/\eta) + {n} {\log(2\sqrt{t} (\vbar + 2n))}}
        \;.
    \end{align*}
    This completes the proof.
\end{proof}

\section{Proof of Theorem~\ref{thm:high_prob_containment}}
\label{sec:proof:thm:high_prob_containment}
    \begin{proof}[Proof of \cref{thm:high_prob_containment}]

    The proof idea of this theorem closely follows Section 5.3 of~\citet{shapiro2021lectures}.

    We first need some additional notations. Define the approximate solutions sets of surrogate problems as follows: For a closed set $ A\subset \Rnpp$, let 
    \begin{align*}
        \cB^*_A(\epsilon) &\defeq \{ \beta \in A: H(\beta) \leq \min_A H  +\eps  \}
        \;,
        \\
        \cB^\gam_A(\epsilon) &\defeq\{ \beta \in A: H_t(\beta) \leq \min_A H_t  +\eps \}
        \;.
    \end{align*}

    In words, they solve the surrogate optimization problems which are defined with a new constraint set $A$. Note that if $\betast \in A$ then $\cB^*_A(\epsilon) = A\cap \cB^*(\eps)$. Recall on the compact set $C$, both $H_t$ and $H$ are $L$-Lipschitz and $\lambda$-strongly convex w.r.t\ the $\ell_\infty$-norm, where $L = (\vbar + 2n)$ and $\lambda = \ubar{b}/4$.

    Let $r \defeq \sup\{ H(\beta) - H^*: \beta \in C \}$. Then if $\epsilon \geq r$ then $C\subset \cBst(\epsilon)$ and the claim is trivial. Now we assume $\epsilon < r$. 
    
    Define $a = \min\{2\epsilon, (r + \eps)/2 \}$. Note $\eps < a < r$. Define $S = C\cap \cBst(a)$. The role of $S$ will be evident as follows. We will show that, with high probability, the following chain of inclusions holds 
    \begin{align*}
        \cBgamC(\delta) \overset{(1)}{\subset} \cBgamS (\delta) \overset{(2)}{\subset} \cBstS(\epsilon) \overset{(3)}{\subset} \cBstC(\epsilon) \, .
    \end{align*}
   
    \textbf{Step 1. Reduction to discretized problems.}
    We let $S'$ be a $\nu$-cover of the set $S = \cBst(a) \cap C$.
    Let $X = S' \cup \{\betast \}$. 
    In this part the goal is to show 
    \begin{align*}
        \P\big(     \cBgamC(\delta)\subset \cBstC(\eps)
        \big) 
        \geq 
        \P\big(\cB^\gam_X(\delta')  \subset \cB^*_X(\epsilon') \big)
    \end{align*}
    where 
    \begin{align*}
        \nu = (\eps' -\delta')/4 > 0 \,, \quad \delta' =  \delta + L\nu > 0
        \,, \quad \epsilon' = \epsilon - L\nu >0
        \,.
    \end{align*}

    First, we claim
\begin{claim} \label{claim:1}
It holds    $
    \cB^\gam_X(\delta')  \subset \cB^*_X(\epsilon') 
    \implies \cBgamS (\delta) \overset{}{\subset} \cBstS(\epsilon)
    $ (Inclusion (2)).
\end{claim}

    Next, we show 
    \begin{claim}\label{claim:2}
        Inclusion (2) implies Inclusion (1):
    $  \cBgamS (\delta) \overset{}{\subset} \cBstS(\epsilon)
    \implies \cBgamC(\delta) \overset{}{\subset} \cBgamS (\delta)\,.$   
    \end{claim}

    Proofs of \cref{claim:1} and \cref{claim:2} are deferred after the proof of \cref{thm:high_prob_containment}. At a high level, \cref{claim:1} uses the covering property of the set $X$. \cref{claim:2} exploits convexity of the problem.

    Finally, we show Inclusion (3) $ \cBstS(\epsilon) \overset{}{\subset} \cBstC(\epsilon)$.  
    Note that $\betast$ belongs to both $C$ and $S$. And thus for any $\beta \in \cBstS (\epsilon)$, it holds $H(\beta) \leq \min_X H + \eps = \Hst + \eps = \min_S H + \eps \,.$ We obtain $\beta \in \cBstC(\epsilon)$.

    To summarize, \cref{claim:1} shows that $\cB^\gam_X(\delta')  \subset \cB^*_X(\epsilon') $ implies Inclusion (2).
    Inclusion (3) holds automatically.
    By \cref{claim:2} we know Inclusion (2) implies Inclusion (1). So it holds deterministically that 
    \begin{align*}
       \{\cB^\gam_X(\delta')  \subset \cB^*_X(\epsilon') \} 
       \subset \{ \cBgamC(\delta)\subset  \cBstC(\eps) \} 
       \;.
    \end{align*}

    \textbf{Step 2. Probability of inclusion for discretized problems.}
    Now we bound the probability $\P(\cB^\gam_X(\delta')  \subset \cB^*_X(\epsilon') )$.

    For now, we forget the construction $X = S' + \{ \betast\}$ where $S'$ is a $\nu$-cover of $S$.
    Let $X\subset C$ be any discrete set with cardinality $|X|$.

    Let $\betastX \in \argmin_X H$ be a minimizer of $H$ over the set $X$. For $\beta \in X$ define the random variable $Y^\tau_\beta \defeq F(\betastX, \thetau) - F(\beta, \thetau)$.  Also let $\mu_\beta \defeq \E[Y^\tau_\beta]$, which is well-defined by the i.i.d.\ item assumption. Let $D\defeq \sup_{\beta \in X} \| \beta - \betastX\|_\infty$.

    Consider any $0 \leq \delta ' < \epsilon'$. If $X- \cB^*_X(\epsilon')$ is empty, then all elements in $X$ are $\epsilon'$-optimal for the problem $\min_X H$. Next assume $X- \cB^*_X(\epsilon')$ is not empty. We upper bound the probability of the event $\cB^\gam_X(\delta') \not \subset \cB^*_X(\epsilon')$.
    \begin{align}
        &\P \big(\cB^\gam_X(\delta') \not \subset \cB^*_X(\epsilon')\big)
        \notag
        \\
        &= \P \big(\text{there exists } \beta \in X- \cB^*_X(\epsilon'),\,  H_t(\beta)\leq H_t(\betastX) + \delta'\big)
        \notag
        \\
        &\leq \sum_{\beta \in X- \cB^*_X(\epsilon')} \P \big(H_t(\beta)\leq H_t(\betastX) + \delta'\big)
        \notag
        \\
        &= \sum_{\beta \in X- \cB^*_X(\epsilon')} \P\bigg(\frac1t \sumtau Y^\tau_\beta \geq - \delta'\bigg)
        \notag
        \\
        &\leq \sum_{\beta \in X- \cB^*_X(\epsilon')} \P\bigg(\frac1t \sumtau Y^\tau_\beta - \mu_\beta \geq \eps' - \delta'\bigg) \tag{A}
        \\
        &\leq \sum_{\beta \in X- \cB^*_X(\epsilon')} \exp\Big(-\frac{2t (\eps'-\delta')\sq}{L\sq \|\beta - \betast \|_\infty\sq}\Big) \tag{B}
        \\ 
        &\leq |X| \exp\Big(-\frac{2t (\eps'-\delta')\sq}{L\sq \|\beta - \betastX \|_\infty\sq}\Big) \,.
        \label{eq:bound_discretization}
    \end{align}
    Here in (A) we use the fact that $\mu_\beta = H(\betastX) - H(\beta) > - \eps' $ for $\beta \in X- \cB^*_X(\epsilon')$. In (B), using $L$-Lipschitzness of $H$ on the set $C$, we obtain $|Y^\tau_\beta| \leq L \|\beta - \betastX \|_\infty$ and then apply Hoeffding's inequality for bounded random variables. Setting \cref{eq:bound_discretization} equal to $\alpha$ and solving for $t$, we have that if 
    \begin{align}\label{eq:t_for_discrete_X}
        t \geq \frac{L\sq D\sq}{2(\eps' - \delta')\sq } \Big(\log|X| + \log\frac1\alpha\Big) \,,
    \end{align}
    then $\P\big(\cB^\gam_X(\delta') \not \subset \cB^*_X(\epsilon')\big) \leq \alpha$. Note the above derivation applies to any finite set $X\subset S$.

    Now we use the construction $X= S' + \{ \betast\}$.
    Then the cardinality of $X$ can be upper bounded by $(4/\nu)^n$. Note since $\betast \in X$ it holds $\betast = \betastX$.
    We apply the result in \cref{eq:t_for_discrete_X} with the following parameters
    \begin{align*}
       & \nu = (\eps' -\delta')/(4L)
       \,, \quad 
       \delta' =  \delta + L\nu
       \,, \quad 
       \epsilon' = \epsilon - L\nu
       \,, \quad 
       \epsilon' - \delta' = \frac12(\epsilon - \delta)\,,
    \\
       & D =  \min \{\sqrt{2a/\lambda} , 2 \}
       \,, 
       \quad |X| \leq 
       \Big(\frac{16L}{\eps - \delta}\Big)^n 
       \,.
    \end{align*}
    We justify the choice of $D$.   First, $S \subset C$ implies $D \leq 2$.
    By the $\lambda$-strong convexity of $H$ on $C$: for all $\beta \in X\subset S \subset \cBst(a)$, it holds
    \begin{align*}
        &(1/2) \lambda \|\beta - \betastX\|^2_\infty = (1/2) \lambda \|\beta - \betast\|^2_\infty \leq H(\beta) - \Hst \leq a
        \\
        \implies 
       & D = \sup_{\beta\in X}\|X - \betastX\|_\infty \leq \sqrt{2a/\lambda} \,. 
    \end{align*}
    Substituting these quantities into the bound \cref{eq:t_for_discrete_X} the expression becomes
    \begin{align*}
        t 
        & \geq  c'\cdot \frac{ L\sq }{ (\eps -\delta)\sq } 
        \cdot 
        \min \bigg\{\frac{2a}{\lambda },4\bigg\} 
        \cdot 
        \bigg(n\log\Big(\frac{16 L }{\eps - \delta}\Big) + \log \frac1\alpha\bigg) 
       \,.
    \end{align*}
    Here $c'$ is an absolute constant that changes from line to line.
    Moreover, noting that $a \leq 2\epsilon$ and $\delta \leq \epsilon/2$ implies $a / (\epsilon - \delta)\sq \leq 8/\epsilon$, we know that if 
    \begin{align} \label{eq:final_t_bound}
        t & \geq c'\cdot  L \sq  \min \bigg\{ \frac{1}{ \lambda \epsilon} , \frac{1}{\epsilon\sq} \bigg\} \cdot \bigg(n\log\Big(\frac{16L}{\eps - \delta}\Big) + \log \frac1\alpha\bigg)
        \;,
    \end{align}
    then
    $\P\big(\cB^\gam_X(\delta')  \subset \cB^*_X(\epsilon')\big) \geq 1- \alpha$. 
    By plugging in $L = (2n + \vbar)$ and $\lambda = \ubar{b} / 4$, we know
    $\P \Big( \cBgamS (\delta) \overset{}{\subset} \cBstS(\epsilon)\Big) \geq 1-\alpha $ as long as 
    \begin{align*}
        t \geq c'\cdot  (2n+\vbar)\sq  \min \bigg\{ \frac{1}{ \ubar{b}\epsilon} , \frac{1}{\epsilon\sq} \bigg\} \cdot \bigg(n\log\Big(\frac{16(2n+\vbar)}{\eps - \delta}\Big) + \log \frac1\alpha\bigg) 
        \;.
    \end{align*}

    \textbf{Step 3. Putting together.}
    By \cref{lm:value_concentration},
    if $t \geq 2\vbarsq {\log(2n/\alpha)}$ then $\betagam \in C$ with probability $\geq 1 - \alpha$. Under the event $\betagam \in C$, it holds $\cBgamC(\delta) = C \cap \cBgam(\delta)$. Since $\betast\in C$ it holds that $\cBstC(\eps) = C\cap \cBst(\eps)$.
    Moreover, if $t$ satisfies the bound in \cref{eq:final_t_bound}, we know Inclusion (2) holds with probability $\geq 1- \alpha$, which then implies Inclusion (1). 
    So if $t$ satisfies the two requirements, $t \geq 2\vbarsq {\log(2n/\alpha)}$ and \cref{eq:final_t_bound}, then with probability $\geq 1- 2\alpha$,
    \begin{align*}
        C \cap \cBgam(\delta) = \cBgamC(\delta ) \subset \cBstC(\eps) = C\cap \cBst(\eps)\,.
    \end{align*}

\end{proof}

\begin{proof}[Proof of \cref{claim:1}]
    To see this, for $\beta \in \cBgamS(\delta)$ let $\beta' \in X$ be such that $\|\beta - \beta' \|_\infty \leq \nu$. By Lipschitzness of $H_t$ on $C$, we know
    \begin{align*}
        H_t(\beta')
        & \leq H_t(\beta)+L\nu \tag{Lipschitzness of $H_t$}
        \\
        &\leq \min_S H_t + \delta + L\nu  \tag{$\beta \in \cBgamS(\delta)$}
        \\ 
        &\leq \min_X H_t + \delta + L\nu \tag{$X\subset S$} 
        \\
        & =  \min_X H_t + \delta'\,.
    \end{align*}
    This implies the membership $\beta' \in \cBgamX(\delta')$. Furthermore, we have 
    \begin{align*}
        \cBgamX(\delta') \subset \cBstX(\eps') \subset \cBstC (\epsilon') \,.
    \end{align*} 
    Here the first inclusion is simply the assumption that $\cB^\gam_X(\delta')  \subset \cB^*_X(\epsilon')$. The second inclusion follows by the construction of $X$; since $\betast \in X$, we know $\cBstX(\epsilon') \subset \cBstC (\epsilon')$ and thus $\min_X H = \min_X H= \Hst$.  We now obtain
    \begin{align*}
        \beta' \in \cBstC (\epsilon') \,.
    \end{align*}

    Using the Lipschitzness of $H$ on $C$, we have for all $\beta \in \cBgamS(\delta)$
    \begin{align*}
        H(\beta) &\leq H(\beta') + L\nu \tag{Lipschitzness of $H$}
        \\
        &\leq \min_C H + \epsilon' + L\nu \tag{$\beta' \in \cBstC (\epsilon') $}
        \\
        & = \min_C H + \epsilon \,.
    \end{align*}
    So we conclude $\beta \in \cBstC(\epsilon)$, implying $\cBgamS(\delta) \subset  \cBstC(\epsilon)$.     This completes the proof of \cref{claim:1}.
\end{proof}

\begin{proof}[Proof of \cref{claim:2}]

    This claim relies on convexity of the problem.
   
       Assume, for the sake of contradiction, there exists $\betadia \in \cBgamC(\delta)$ but $\betadia \not \in \cBgamS(\delta)$. The only possibility this can happen is $\betadia \in C$ but $\betadia \not \in S = C\cap \cBst(a)$. So $\betadia \not \in \cBst(a)$ (note $a < r$ implies the set $C- \cBst(a)$ is not empty), which by definition means 
       \begin{align}\label{eq:betadia_gt_a}
           H(\betadia) - \Hst > a \,.
       \end{align}
       Now define 
       \begin{align*}
           \betabar =  \argmin_{\beta \in S} H_t(\beta) \in \cBgamS(\delta) \,.
       \end{align*}
       By the assumption $\cBgamS (\delta) \overset{}{\subset} \cBstS(\epsilon)$, we know $\betabar \in \cBstS(\eps)$ and so
       \begin{align}\label{eq:betabar_leq_eps}
           H(\betabar) - \Hst \leq \epsilon \,.
       \end{align}
       
       Next, let $\beta^c = c \betabar + (1-c) \betadia$ with $c\in [0,1]$, which is a point lying on the line segment joining the two points $\betabar$ and $\betadia$. By the optimality of $\betadia \in \cBgamC(\delta)$ and $\betabar \in C$, we know $H_t(\betadia) \leq H(\betabar) + \delta$.
       By convexity of $H_t$, we have for all $c\in[0,1]$,
       \begin{align}\label{eq:line_segment_implications}
           H_t(\beta^c) \leq \max\{ H_t(\betabar), H_t(\betadia)\} \leq H_t(\betabar) + \delta \,.
       \end{align} 
   
       Now consider the map $K: [0,1] \to \R_+, c \mapsto H(\beta^c) - \Hst$. Since any convex function is continuous on its effective domain~\cite[Corollary 10.1.1]{rockafellar1970convex}, we know $H$ is continuous. Continuity of $H$ implies continuity of $K$.
       Note $K(0) = H(\betadia) - \Hst > a$ by \cref{eq:betadia_gt_a} and $K(1) = H(\betabar) - \Hst \leq \epsilon$ by \cref{eq:betabar_leq_eps}. By intermediate value theorem, there is $c^* \in [0,1]$ such that $\eps < H(\beta^{c*}) - \Hst < a$. Moreover, by $H(\beta^{c*}) - \Hst < a$ and $\beta^{c*} \in C$ we obtain $\beta^{c*} \in S = \cBst(a) \cap C$. In addition, recalling $H_t(\beta^{c*}) \leq H_t(\betabar) + \delta$ (\cref{eq:line_segment_implications}), we conclude by definition $\beta^{c*}\in \cBgamS(\delta)$.
       
       At this point we have shown the existence of a point $\beta^{c*}$ such that
       \begin{align*}
           \beta^{c*} \in \cBgamS(\delta) \,, \quad \beta^{c*} \not \in \cBst(\eps)\,.
       \end{align*}
       This clearly contradicts the assumption $\cBgamS (\delta) \overset{}{\subset} \cBstS(\epsilon) = \cBst(\epsilon) \cap S$.
       This completes the proof of \cref{claim:2}.
   
   \end{proof}

\begin{proof}[Proof of \cref{cor:H_concentration}]
    Under the event $\{ \betagam \in C\}$, the set $C \cap \cBgam (0) = \{ \betagam \}$. 
    Moreover, $\betagam \in C \cap \cBst (\epsilon)$ implies $H(\betagam) \leq H(\betast) + \epsilon$. This completes the proof.
\end{proof}

\begin{proof}[Proof of \cref{cor:beta_u_concentration}]
    Under the event $\{ \betagam \in C\}$, we use strong convexity of $H$ over $C$ w.r.t.\ $\ell_2$-norm and obtain $\frac{\lambda}{2} \| \betagam - \betast\|_2 \sq \leq H(\betagam) - H(\betast)$ where $\lambda = \ubar{b}/4$ is the strong-convexity parameter. 

    For the second claim we use the equality $\betagami = b_i / \ugami$ and $\betasti = b_i / \usti$. 
    For $\beta, \beta' \in C$, it holds $|\frac{1}{\betai} - \frac{1}{\beta'_i} | \leq \frac{4}{{b_i}^2} | \betai - \beta'_i|$.
    And so $\| \ugam - \ust \|_2 = \sumi (b_i)\sq (\frac{1}{\betagami} - \frac{1}{\betasti})\sq \leq \sumi \frac{16}{(b_i)\sq} | \betagami - \betasti|\sq \leq \frac{16}{(\ubar{b})\sq} \| \betagam - \betast\|_2\sq$. 
    So we obtain $\|\ugam -\ust\|_2 \leq \frac{4}{\ubar{b}} \|\betagam - \betast\|_2$.
    We complete the proof.
\end{proof}

\section{Proof of Theorems~\ref{thm:normality} and \ref{thm:clt_beta_u}} 
\label{sec:proof:thm:normality}
\begin{proof}[Proof of \cref{thm:normality}]

    We aim to apply \cref{lm:clt_optimal_value} to our problem. To do this we first introduce surrogate problems 
    \begin{align*}
        H^\gam_C \defeq \min_{\beta \in C} H_t(\beta) 
        \;,\quad  
        H^*_C   \defeq \min_{\beta \in C} H(\beta)
        \;.
    \end{align*}
    Since $\betast \in C$ we know $H^*_C = \Hst$. We write down the decomposition
    \begin{align*}
        \sqrt{t} (H^\gam - \Hst) = \sqrt{t} (H^\gam -  H^\gam_C ) + \sqrt{t}( H^\gam_C  - H^*_C)
        \;.
    \end{align*}

    For the first term we show that $\sqrt{t} (H^\gam -  H^\gam_C ) \toprob 0$ (which implies convergence in distribution). Choose any $\eps > 0$, define the event $A^\eps_t = \{ \sqrt{t} |H^\gam - H^\gam_C | \geq \eps\}$. By \cref{lm:value_concentration} we know that with probability 1, $\betagam \in C$ eventually and so $H^\gam - H^\gam_C = 0$ eventually. This implies 
    $\P(\limsup_{t\to\infty}A^\eps_t) = \P(A^\eps_t \text{ eventually}) = 0$. By Fatou's lemma, 
    \begin{align*}
        \limsup_{t\to\infty} \P (A^\eps_t) \leq \P \Big(\limsup_{t\to\infty}A^\eps_t \Big) = 0
        \;.
    \end{align*} 
    We conclude for all $\epsilon > 0$, $\lim_{t\to\infty} \P(\sqrt{t} | H^\gam -  H^\gam_C | > \epsilon)=0$.

    For the second term, we invoke \cref{lm:clt_optimal_value} and obtain $\sqrt{t}( H^\gam_C  - H^*_C) \tod N(0,\var[F(\betast,\theta)])$, where we recall $F(\beta,\theta) = \max_i \betai v_i(\theta) - \sumiton b_i \log \beta_i$.
    To do this we verify all hypotheses in \cref{lm:clt_optimal_value}.
    \begin{itemize}
        \item The set $C$ is compact and therefore \cref{it:lm:clt_optimal_value:1} is satisfied. 
        \item The function $F$ is finite for all $\beta \in \Rnpp$ and thus \cref{it:lm:clt_optimal_value:2} holds. 
        \item The function $F(\cdot, \theta)$ is $(2n+\vbar)$-Lipschitz on $C$ for all $\theta$, and thus \cref{it:lm:clt_optimal_value:3} holds. 
        \item \cref{it:lm:clt_optimal_value:4} holds because the function $H$ has a unique minimizer over $C$. 
    \end{itemize}
Now we calculate the variance term. 
    \begin{align*}
        \var(F(\betast, \theta)) &= \var \Big(\max_i \{ v_i(\theta) \betasti\}\Big)
        \\
        &= \var (p^*(\theta)) 
        \\
        &= \int_\Theta (p^*)\sq \diff S (\theta) - \bigg(\int_\Theta p^* \diff S (\theta) \bigg) \sq
        \\
        &= \int_\Theta (p^*)\sq \diff S (\theta) - 1
        \;,
    \end{align*}
    where in the last equality we used $\int p^* = \sumiton b_i = 1$.
    By Slutsky's theorem, we obtain the claimed result.

\end{proof}

\begin{proof}[Proof of Theorem~\ref{thm:clt_beta_u}]

    We verify all the conditions in Theorem 2.1 from~\citet{hjort2011asymptotics}.

    Because $H$ is $C^2$ at $\betast$, there exists a neighborhood $N$ of $\betast$ 
    such that $H$ is continuously differentiable on $N$. By \cref{thm:first_differentiability} this implies that 
    the random variable $\textsf{bidgap}(\beta,\cdot)\inv$ is finite almost surely for each $\beta \in N$. This implies $I(\beta,\theta)$ is single valued a.s.\ for $\beta \in N$.

    Define $D(\theta)\defeq \nabla F(\betast,\theta) = G(\betast,\theta) - \nabla \Psi(\betast)$ where we recall the subgradient $G(\betast,\theta) = e_{i(\betast,\theta)}v_{i(\betast,\theta)}$ and $i(\betast,\theta) = \argmax_i \betasti v_i(\theta)$ is the winner of item $\theta$ when the pacing multiplier of buyers is $\betast$. 
    Let $$R(h, \theta) \defeq [ F(\betast + h,\theta) - F(\betast,\theta) - D(\theta) \tp h ] / \| h \|_2$$ measure the first-order approximation error.
    By optimality of $\betast$ we know $\nabla H (\betast)=\E[D(\theta)] = 0$. Moreover, by twice differentiability of $H$ at $\betast$, the following expansion holds: 
    \begin{align*}
        H(\betast + h) - H(\betast) = \frac12 h \tp \big(\nabla \sq H(\betast)\big) h + o(\| h\|_2\sq)
        \;.
    \end{align*}

    To invoke Theorem 2.1 from~\citet{hjort2011asymptotics}, we check the following stochastic version of differentiability condition holds 
    \begin{align}
        \E [R(h,\theta) \sq ] = o(\| h \|_2\sq) \quad \text{as } \|h\|_2 \downarrow 0
        \;.
        \label{eq:diff_of_D}
    \end{align} 

By $H$ being differentiable at $\betast$, 
we know $R(\t, h)  \toas 0$. Since we assume $\max_i \esssup v_i(\t) < \infty$, we know the sequence of random variables $R(\t, h)$ is bounded. 
    We conclude \cref{eq:diff_of_D} holds true.

    At this stage we have verified all the conditions in Theorem 2.1 from~\citet{hjort2011asymptotics}. Invoking the theorem we obtain
    \begin{align*}
        \sqrt{t} (\betagam - \betast) = - [\nabla \sq H(\betast)]\inv \Bigg( \frac{1}{\sqrt t}\sumtau D(\thetau) \Bigg) + o_p(1)
        \;.
    \end{align*} 
    In particular,  $\sqrt{t} (\betagam - \betast) \tod \cN (0, [\nabla \sq H(\betast)]\inv \var(D) [\nabla \sq H(\betast)]\inv )$.
    Finally, noting $D = \must$ we obtain the claimed result.

    \emph{Proof of CLT for $\beta$}
    This follows from the discussion above.

    \emph{Proof of CLT for $u$.}  We use the delta method. Take $g(\beta) = [b_1/\beta_1,\dots, b_n/\beta_n]$. Then the asymptotic variance of $\sqrt{t}(g(\betagam) - g(\betast))$ is $\nabla g(\betast) 
    \tp 
    \Sigma_\beta
    \nabla g (\betast)$. Note $\nabla g(\betast)$ is the diagonal matrix $\Diag(\{-b_i/\betasti\sq \})$. 
    From here we obtain the expression for $\Sigma_u$.

\end{proof}

\section{Proofs for Analytical Properties of the Dual Objective}
\label{proof:sec:analytical_properties_of_dual_obj}
\begin{remark}[Comment on \cref{thm:first_differentiability}]
    We briefly discuss why differentiability is related to the gap in buyers' bids.
Recall $\fbar(\beta)= \E[\max_i \beta_i v_i(\theta)]$.
Let $\delta \in \Rnp$ be a direction with positive entries, and let $I(\beta, \theta) = \argmax_i{\beta_i v_i(\theta)}$ be the set of winners of item $\theta$ which could be multivalued. Consider the directional derivative of $\fbar$ at $\beta$ along the direction $\delta$:
\begin{align*}
    & \lim_{t \downarrow 0} 
     \E\bigg[\frac{\max_i (\beta_i + t \delta_i) v_i(\theta) - \max_i \beta_i v_i(\theta)}{t}\bigg] 
    \\
    & =  
    \E \bigg[\lim_{t \downarrow 0} \frac{\max_i (\beta_i + t \delta_i) v_i(\theta) - \max_i \beta_i v_i(\theta)}{t}\bigg] 
    \\
    & =
    \E\Big[\max_{i \in I(\beta,\theta)} v_i (\theta) \delta_i\Big] 
    ;,
\end{align*}
where the exchange of limit and expectation is justified by the dominated convergence theorem.
Similarly, the left limit is
\begin{align*}
    \lim_{t \uparrow 0} \E\bigg[\frac{\max_i (\beta_i + t \delta_i) v_i(\theta) - \max_i \beta_i v_i(\theta)}{t}\bigg] = \E\Big[\min_{i \in I(\beta,\theta)} v_i(\theta) \delta_i\Big] 
    \;.
\end{align*}
If there is a tie at $\beta$ with positive probability, i.e., the set $I(\beta,\theta)$ is multivalued for a non-zero measure set of items, then the left and right directional derivatives along the direction $\delta$ do not agree. Since differentiability at a point $\beta$ implies existence of directional derivatives, we conclude differentiability implies \cref{eq:as:notie}.

\end{remark}

\begin{proof}[Proof of \cref{thm:first_differentiability}]
    Recall $f(\beta,\theta) = \max_i \beta_i \vithe$. Note $f$ is differentiable at $\beta$ if and only if $ \textsf{bidgap}(\beta,\theta) > 0$. 
    Let $\Theta_{\mathrm{diff}}(\beta) \defeq \{ \theta: f(\beta,\theta) \text{ is continuously differentiable at $\beta$} \}$. Then 
    \begin{align*}
        \Theta_{\mathrm{diff}}(\beta) = \bigg\{\theta: \frac{1}{\textsf{bidgap}(\beta,\theta)} < \infty \bigg\} = \{ \theta: I(\beta,\theta) \text{ is single-valued}\} 
        \;.
    \end{align*}
    
    By Proposition 2.3 from~\citet{bertsekas1973stochastic} we know $\fbar(\beta) = \E[f(\beta,\theta)] = \int_\Theta f(\beta,\theta) \diff S(\theta)$ is differentiable at $\beta$ if and only if $S( \Theta_{\mathrm{diff}}(\beta)) = 1$.
    From here we obtain  \cref{thm:first_differentiability}.

\end{proof}
\begin{remark}
    Suppose \cref{eq:as:notie} holds in a neighborhood $N$ of $\betast$,
    i.e., $\frac{1}{\textsf{bidgap}(\beta,\theta)}$ is finite a.s.\ for each $\beta \in N$, then $H$ is \emph{continuously} differentiable on $N$.
    See Proposition 2.1 from \citet{shapiro1989asymptotic}.
\end{remark}

\begin{proof}[Proof of \cref{thm:int_implies_hessian}]
    \cref{eq:as:UI}
    holds in a neighborhood $N$ of $\betast$
    implies the \cref{eq:as:notie} holds on $N$, 
    and thus $H$ is differentiable on $N$ with gradient $\nabla H(\beta) = \E[v_{i(\beta,\theta)} e_{i(\beta,\theta)}] + \nabla\Psi = \E[G(\beta,\theta)] + \nabla\Psi$. To compute the Hessian w.r.t.\ the first term, we look at the limit 
    \begin{align}
        \label{eq:G_ratio}
        \lim_{\|h\| \downarrow 0} \E \bigg[\frac{G(\betast + h,\theta) - G(\betast, \theta)}{\|h\|}\bigg]
        \;.
    \end{align}

\begin{lemma}[$\textsf{bidgap}(\beta,\theta)$ as   Lipschitz parameter of $G$] 
        \label{lm:epsilon_as_lip}
        Suppose for some $\beta\in \Rnpp$ the gap function $\textsf{bidgap}(\beta,\theta) > 0$.
         Let $\beta' = \beta + h$. 
        \begin{itemize}
            \item If $\|h\|_\infty \leq \textsf{bidgap}(\beta,\theta)/\vbar$ then 
            $I(\beta',\theta)$ is single-valued and moreover 
            $i(\beta',\theta) = i(\beta,\theta)$, implying $G(\beta',\theta) = G(\beta,\theta)$.
            \item It holds ${ \| G(\beta + h,\theta) - G(\beta, \theta) \|_2}{} \leq 6 \vbar\sq \cdot \frac{1}{ \textsf{bidgap}(\beta,\theta)} \|h\|_2 $ for all $\beta + h \in\Rnp$.
        \end{itemize}
\end{lemma}

    Suppose we could exchange expectation and limit in \cref{eq:G_ratio}, then the above expression would become zero: for a fixed $\theta$, since \cref{eq:as:notie} holds at $\betast$, i.e, $\textsf{bidgap}(\betast,\theta) > 0$, we apply \cref{lm:epsilon_as_lip} and obtain
    $
        \lim_{\|h\| \downarrow 0}  {(G(\betast + h,\theta) - G(\betast, \theta))}/{\|h\|} = 0.
    $
    This implies that $H$ is twice differentiable at $\betast$ with Hessian $\nabla\sq H(\betast) = \nabla\sq \Psi(\betast)$.
    It is then natural to ask for sufficient conditions for exchanging limit and expectation.

    By \cref{lm:epsilon_as_lip}, we know the ratio 
    ${(G(\betast + h,\theta) - G(\betast, \theta))}/{\|h\|}$ is dominated by $6\vbar \textsf{bidgap}(\betast,\theta)\inv$, which by \cref{eq:as:UI} is integrable. 
    By dominated convergence theorem, we can exchange limit and expectation, and the claim follows.
\end{proof}
    
\begin{proof}[Proof of \cref{lm:epsilon_as_lip}]

    Note that for any $\beta$ and $\beta' = \beta + h$,
\begin{align} 
    \label{eq:lip_of_G}
    \frac{ \| G(\beta + h,\theta) - G(\beta, \theta) \|_2}{\|h\|_2} \leq 6 \vbar\sq \cdot \frac{1}{ \textsf{bidgap}(\beta,\theta)}
    \;.    
\end{align}
To see this, we notice that on one hand, if $\| h \|_\infty\leq \epsilon / (3\vbar )$ where $\epsilon = \textsf{bidgap}(\beta,\theta)$, then for $i = i(\beta,\theta)$ and all $\theta \in \Theta_i(\beta)$,
\begin{align*}
    \beta'_i v_i(\theta) 
    & = (\betai + h_i) v_i(\theta)
    \\
    &\geq \betai v_i(\theta) -  \epsilon /3 \tag{A}
    \\
    &\geq \beta_k v_k(\theta) + \epsilon -  \epsilon /3
    \tag{B}
    \\
    &\geq \beta'_k  v_k(\theta)  - \epsilon /3 + \epsilon -  \epsilon /3 \tag{C} 
    \;,
\end{align*}
where $(A)$ and $(C)$ use the fact $\|h\|_\infty \leq \epsilon / (3\vbar)$, and (B) uses the definition of $\epsilon$.
This implies $\argmax_i \beta'_i \vithe = \argmax_i \betai \vithe$ and thus
$G(\beta + h ,\theta) - G(\beta  ,\theta) = 0$.
On the other hand, if $\| h \|_\infty > \epsilon / (3\vbar)$, then $\| h\|_2\geq  \|h\|_\infty > \epsilon / (3\vbar)$. Using the bound $\| G \|_2 \leq \vbar$, we obtain \cref{eq:lip_of_G}.
This completes proof of \cref{lm:epsilon_as_lip}.
\end{proof}

\begin{proof}[Proof of \cref{thm:linear_value_implies_hessian}]

    By Lemma 5 from~\citet{gao2022infinite}, we know that at equilibrium the there exists unique breakpoints 
    $ 0 = a_0^* < a_{1}^{*}<\cdots<a_{n}^{*}=1$ 
    such that buyer $i$ receives the item set $[a_{i-1}^*, a_i^*] \subset \Theta$. 
    Moreover, it holds 
    \begin{align*}
        & \beta^*_1 d_1 > \beta^*_2 d_2 > \dots > \beta^*_n d_n \;,
        \\
        & \beta^*_1 c_1 < \beta^*_2 c_2 < \dots < \beta^*_n c_n \;.
    \end{align*}

    Now we consider a small enough neighborhood $N$ of $\betast$. For each $\beta\in N$, we define the 
    breakpoint $a_i^*(\beta) $ by solving for $\theta$ through 
    $\beta_i(c_i \theta + d_i) = \beta_{i+1}(c_{i+1} \theta + d_{i+1})$ for $i\in[n-1]$, $a_0^*(\beta) = 0$, and $a_n^*(\beta) = 1$. 
    As a sanity check, note $a_i^*(\betast) = a_i^*$ for all $i$.
    Recall $\Theta_i(\beta) = \{\theta \in \Theta: v_i(\theta)\beta_i \geq v_k(\theta)\beta_k, \forall k\neq i \}$ is the winning item set of buyer $i$ when pacing multiplier equal $\beta$. We will show later $\Theta_i(\beta) = [a_{i-1}^*(\beta), a_{i}^*(\beta) ]$ for $\beta \in N$ when $N$ is appropriately constructed.
    
    Now we recall the gradient expression 
    \begin{align*}
        \nabla \fbar(\beta) 
        & = \E[e_{i(\beta,\theta)} v_{i(\beta,\theta)}] 
        \\
        & = \sumiton e_i \int_{a_{i-1}^*(\beta) }^{a_i^*(\beta)} c_i \theta + d_i \diff \theta
        \\
        & = \sumiton e_i \bigg( \frac{c_i}{2}\big( [a_{i}^*(\beta)]\sq - [a_{i-1}^*(\beta)]\sq \big) + d_i \big(a_{i}^*(\beta) - a_{i-1}^*(\beta)\big) \bigg) 
        \;.
    \end{align*}
    On $N$, the breakpoints $a_i^*(\beta)= (-\beta_i d_i + \beta_{i+1}d_{i+1}) / (\beta_i c_i - \beta_{i+1} c_{i+1})$ is $C^1$ in the parameter $\beta$. We conclude $\nabla \fbar(\beta) $ is continuously differentiable. 

    What remains is to construct such a neighborhood $N$. Define $$\delta = \min
    \bigg\{
        \frac12 \ubar{\Delta}_{\beta d} / \bar{\Delta}_d
        , \frac12 \ubar{\Delta}_{\beta c} / \bar{\Delta}_c 
        ,
        \frac {1}{4} \ubar{\Delta}_a \ubar{\Delta}_{\beta c} / \vbar 
        \bigg\}\; ,$$ where 
        $\ubar{\Delta}_a\defeq \min |a_i - a_{i-1}|$,
    $ \ubar{\Delta}_{\beta c} \defeq \min \{  \beta^*_{i-1}c_{i-1} -\beta^*_i c_i \} > 0$, 
    $ \bar{\Delta}_c \defeq \max_i\{c_{i-1} - c_{i}\} > 0$, and $ \ubar{\Delta}_{\beta d} > 0$ and $ \bar{\Delta}_d > 0$ are similarly defined. 
    Let $N = \{ \beta: \|\beta - \betast\|_\infty \leq \delta \}$.
    The neighborhood $N$ is constructed so that on $N$ it holds 
    \begin{align*}
        & \beta_1 d_1 > \beta_2 d_2 > \dots > \beta_n d_n \; ,
        \\
        & \beta_1 c_1 < \beta_2 c_2 < \dots < \beta_n c_n \; ,
        \\
        & 0 = a_0^*(\beta) < a_{1}^{*} (\beta) <\cdots<a_{n}^{*}(\beta)=1 \; ,
    \end{align*}
    where the first inequality follows from 
    $\delta \leq \frac12 \ubar{\Delta}_{\beta d} / \bar{\Delta}_d
    $, the second inequality from $\delta \leq \frac12 \ubar{\Delta}_{\beta c} / \bar{\Delta}_c$, and the third inequality follows from the $\delta \leq \ubar{\Delta}_a \ubar{\Delta}_{\beta c} / (4\vbar)$, where $\vbar = \max_i \sup_{\theta \in [0,1]} c_i \theta + d_i$.
    From here we can see that the partition of item set $\Theta$ induced by these breakpoints corresponds to exactly the winning item sets when buyers' pacing multiplier is $\beta$. So we have shown $\Theta_i(\beta) = [a_{i-1}^*(\beta), a_{i}^*(\beta) ]$ for $\beta \in N$.
    This finishes the proof of \cref{thm:linear_value_implies_hessian}.
\end{proof}

\begin{proof}[Proof of \cref{thm:smooth_density_implies_hessian}]
    We need the following technical lemma on the continuous differentiability of integral functions.
    \begin{lemma} [Adapted from Lemma 2.5 from~\citet{wang1985distribution}]
        \label{lm:integral_is_cont_diff}
        Let $u = [u_1, \dots, u_n]\in \Rnpp$, and $$I(u) = \int_{0}^{u_1} \diff t_1 \dots \int_{0}^{u_n}  h (t_1, t_2,\dots, t_n)  \diff t_n \; ,$$ where $h$ is a continuous density function of a probabilistic distribution function on $\Rnp$ and such that all lower dimensional density functions are also continuous. Then the integral $I(u)$ is continuously differentiable.
    \end{lemma}
    \begin{remark}
        The difference between the above lemma and the original statement is that  the original theorem works with density $h$ and integral function $I(u)$ both defined on $\Rn$, while the adapted version works with density $h$ and integral $I(u)$ defined only on $\Rnpp$. 
    \end{remark}

    Recall the gradient expression $$\nabla \fbar(\beta) = \E[e_{i(\beta,\theta)} v_{i(\beta,\theta)}] = \sumiton e_i \int v_i \indi(V_i(\beta)) f(v)\diff v \; ,$$
    where the set $V_i(\beta) = \{ v \in \Rnp: v_i \beta_i \geq v_k \beta_k, k\neq i\}, i \in [n]$, is the values for which buyer $i$ wins. 
    For now, we focus on the first entry of the gradient, i.e., $\int v_1 \indi(V_1(\beta)) f(v)\diff v$.
    We write the integral more explicitly as follows. By Fubini's theorem,
    \begin{align}
        & \int v_1 \indi(V_1(\beta)) f(v)\diff v
        \\
        &  = \int_{0}^{\infty} \diff v_1 \int_0^{\frac{\beta_1 v_1}{\beta_2}} \diff v_2 \int_0^{\frac{\beta_1 v_1}{\beta_3}} \diff v_3 \dots 
        \int_0^{\frac{\beta_1 v_1}{\beta_n}}
        \underbrace{(v_1 f(v_1,\dots, v_n) )}_
        {=: A_1(v)}\diff v_n 
        \label{eq:integral_before_change_of_var}
        \;.
    \end{align}
    To apply the lemma we use a change of variable. Let $t = T(v) = [v_1,\frac{v_2}{v_1}, \dots, \frac{v_n}{v_1}]$ and $v = T\inv (t) = [t_1, t_2 t_1, \dots, t_n t_1]$. Then \cref{eq:integral_before_change_of_var} is equal to 
    \begin{align}
        \int_{0}^{\infty} \diff t_1 \int_0^{\frac{\beta_1}{\beta_2}} \diff t_2 \int_0^{\frac{\beta_1 }{\beta_3}} \diff t_3 \dots 
        \int_0^{\frac{\beta_1}{\beta_n}} \underbrace{\big(t_1^n f(t_1, t_2 t_1, \dots, t_n t_1)\big)}_
        {=: A_2(t)}
        \diff t_n 
        \;.
        \label{eq:after_change_of_var}
    \end{align}
    Note $\E[v_1(\theta)] = \int_\Rnpp A_1(v) \diff v= \int_\Rnpp A_2(t) \diff t = 1 $.
    We use Fubini's theorem and obtain 
    \begin{align*}
        \cref{eq:after_change_of_var} = \int_0^{\frac{\beta_1}{\beta_2}} \diff t_2 \int_0^{\frac{\beta_1 }{\beta_3}} \diff t_3 \dots 
        \int_0^{\frac{\beta_1}{\beta_n}} h(t_{-1}) \diff t_{-1}
        \;,
    \end{align*}
    where we have defined $h(t_{-1}) = \int_\Rp t_1^n f(t_1, t_2 t_1, \dots, t_n t_1) \diff t_1$. 
    By the smoothness assumption on $h$ and \cref{lm:integral_is_cont_diff}, 
    we know that the map $u_{-1} \mapsto \int_{0}^{u_2} \diff t_2 \dots \int_{0}^{u_n}  h (t_{-1})  \diff t_n $ is $C^1$ for all $u_{-1} \in \Rnmpp $. 
    Moreover, 
    the map $\beta \mapsto [\frac{\beta_1}{\beta_2}, \dots, \frac{\beta_1}{\beta_n}]$ is $C^1$. We conclude the first entry of $\nabla \fbar(\beta)$ is $C^1$ in the parameter $\beta$. 
    A similar argument applies to other entries of the gradient. We complete the proof of \cref{thm:smooth_density_implies_hessian}.
\end{proof}

\section{Proof of Theorem~\ref{thm:ci_lnsw}}
\label{sec:proof_var_est}
\begin{proof}[Proof of \cref{thm:ci_lnsw}]
    Define the functions 
    \begin{align*}
        \hat \sigma \sq (\beta) & \defeq \frac{1}{t} \sumtau\Big( F(\beta, \thetau)  - H_t(\beta)\Big)\sq
        \;,
        \\
        \sigma \sq (\beta) & \defeq \var(F(\beta,\theta)) = \E\big[(F(\beta,\theta) - H(\beta))\sq \big]
        \;.
    \end{align*}
    We will show uniform convergence of $\hat \sigma \sq$ to $\sigma\sq$ on $C$, i.e., $\sup_{\beta \in C} |\hat \sigma \sq - \sigma\sq|  \toas 0$.
    We first rewrite $\sighatsq$ as follows 
    \begin{align*}
        \sighatsq(\beta) = \underbrace{\frac{1}{t}   \sumtau \big(F(\beta,\thetau) - H(\beta)\big)\sq}_
        {=: \I(\beta)}
        - \underbrace{(H_t(\beta)  - H(\beta))\sq  }_
        {=: \II(\beta)}
        \;.
    \end{align*}
    By Theorem 7.53 of~\citet{shapiro2021lectures}, the following uniform convergence results hold
    \begin{align*}
        \sup_{\beta \in C} |\I(\beta)  - \sigma\sq(\beta) | \toas 0
        \;,
        \quad
        \sup_{\beta \in C} |\II(\beta)| \toas 0 
        \;.
    \end{align*}
    The above two inequalities imply $\sup_{\beta \in C} |\hat \sigma \sq - \sigma\sq|  \toas 0$. Note the variance estimator 
    $ \hat \sigma\sq_{\NSW} = \hat \sigma\sq (\betagam)$ 
    and the asymptotic variance 
    $ \sigma\sq_{\NSW} = \sigma \sq(\betast) $.
    By $\betagam \toas \betast$ we know, 
    \begin{align*}
        & 
        | \hat \sigma\sq_{\NSW} - \sigma\sq_{\NSW}| 
        \\
        & = |\hat \sigma\sq(\betagam) - \sigma \sq(\betast) | 
        \\
        & \leq |\hat \sigma\sq (\betagam) -  \sigma\sq (\betagam)|
        + |  \sigma\sq (\betagam) -  \sigma\sq(\betast) |
        \\
        & \to 0 \quad \text{ a.s.}
    \end{align*}
    where the first term vanishes by the uniform convergence just established, the second term by continuity of $\sigma \sq(\cdot)$ at $\betast$. 
    
    Now we have shown $\hat \sigma\sq_{\NSW}$ is a consistent variance estimator for the asymptotic variance. Then note 
    \begin{align*}
        \sqrt n \frac{\LNSW^\gam - \LNSW^*}{ \hat \sigma _\NSW} 
        = \sqrt n \frac{\LNSW^\gam - \LNSW^*}{  \sigma _\NSW} 
        \cdot \frac{ \sigma _\NSW}{\hat \sigma _\NSW} 
        \;. 
    \end{align*}
    Since $\sqrt n\frac{\LNSW^\gam - \LNSW^*}{  \sigma _\NSW} \tod N(0,1)$ by \cref{thm:normality}, and $\frac{ \sigma _\NSW}{\hat \sigma _\NSW}  \toprob 1$ which is a constant, by Slutsky's theorem we know $ \sqrt n \frac{\LNSW^\gam - \LNSW^*}{ \hat \sigma _\NSW} \tod N(0,1)$. This completes the proof of \cref{thm:ci_lnsw}.
\end{proof}
\begin{proof}[Proof of \cref{thm:estiamtion_Omega}]
    The assumption that $\E [\textsf{bidgap}(\betast, \theta)\inv ]< \infty$ on a neighborhood $N$ of $\betast$ implies that 
    the set $I(\beta,\theta)$ is single-valued for all $\beta \in N$ and almost all $\theta$. So $G(\beta,\theta)  = v_{i(\beta,\theta)}e_{i(\beta,\theta)}$ is well-defined. 
    Let $G_i(\beta,\theta) = v_i \indi\{i = i(\beta,\theta) \}$ be the $i$-th entry of the vector $G(\beta,\theta)$.
    For any $\beta'\in N$ define
    \begin{align*}
        \hat \Omega\sq_i  (\beta') \defeq \frac1t \sumtau \Bigg( G_i(\beta',\thetau) - \bigg(\frac1t \sumtau G_i(\beta',\thetau)\bigg) \Bigg)\sq
        \;.
    \end{align*}
    By $\betagam \toas \betast$, we know that for large enough $t$, $\betagam \in N$. 
    Moreover, we claim 
    \begin{align*}
        \betagam \in N \implies G_i(\betagam,\thetau) = t \ugamtaui \;. 
    \end{align*}
    To see this, 
    $\betagam \in N$ implies that 
    the set of items that incurs a tie is zero-measure, i.e, $S(\theta: I(\betagam,\theta) \text{ multivalued}) = 0$.
    By the first-order condition of finite sample EG (\cref{fact:sample_eg}),
    the equilibrium allocation in the observed market is then unique and pure (no splitting of items, $\xgamtaui \in \{0,\frac1t \}$), in which case $G_i(\betagam,\theta) =  \indi\{i = i(\beta,\theta) \} \vithetau = t \xgamtaui \vithetau    = t\ugamtaui $ and thus 
    $\hat \Omega\sq_i  (\cdot)|_{\betagam} = \hatOmesqi $. By Theorem 7.53 of~\citet{shapiro2021lectures} one can show the following uniform convergence result 
    \begin{align*}
        \sup_{\beta \in N} | \hat \Omega\sq_i (\beta) - \var[ G_i(\beta,\theta)] | \toas 0 \;. 
    \end{align*}
    Noting $\Omega_i \sq = \var[G_i(\betast,\theta)]$, the uniform convergence result implies $\hat \Omega\sq_i (\betagam) - \var[G_i(\betast,\theta)] \toprob 0$, which is equivalent to $\hatOmesqi \toprob \Omega_i\sq$.

    This completes the proof of \cref{thm:estiamtion_Omega}.
\end{proof}

\section{Proof of Theorem~\ref{thm:rev_convergence}}
\label{sec:proof:thm:rev_convergence}

We start by formally defining quasilinear market equilibrium.
\begin{defn}
    The {Long-run QME}, $\QME(b,v,s)$ is an allocation-utility-price tuple $(x^*, u^*, p^*) \in (L^\infty_+)^n \times \Rnp \times L^1_+$ such that conditions in \cref{def:long-run_market} hold,
    except that buyer optimality is now defined as  
    \begin{itemize}
        \item[2'] $x^*_i \in D_i (p^*)$ and $u^*_i = \lg v_i - \pst , x_i \rg$ for all $i$ where 
        the {demand} $D_i$ of buyer $i$ is its set of utility-maximizing allocations given the prices and budget:
        \[ D_i (p) := \argmax \{ \langle v_i - p, x_i \rangle : x_i \in L^\infty_+,\, \langle p, x_i\rangle \leq b_i \} \;.\] 
    \end{itemize}
\end{defn}

\begin{defn}
    The {observed QME}, $\QMEgam(b,v,\sfs)$, given an item sequence $\gamma$, is an allocation-utility-price tuple $({x}^\gam, u^\gam, p^\gam) \in (\R^t_+)^n \times \Rnp\times \R^t_+$ such that \cref{defn:observed_market} holds except that 
    buyer optimality is now defined as 
\begin{itemize}
    \item[2'] ${x}^\gam_i \in D_i (p^\gam)$ and $u^\gam_i = \lg v_i(\gam) - \pgam ,x_i\rg $ for all $i$, where (overloading notations) $$D_i (p) := \argmax \{   \langle {{v}}_i(\gam) - p , {x}_i \rangle : {x}_i \geq 0,\,  \langle p,  {x}_i\rangle \leq b_i \} \;.
    $$ 
\end{itemize}
\end{defn} 
We will need the following convex program characterizations of infinite-dimensional QME introduced in Secion~6 of~\citet{gao2022infinite}.
First we state the primal and dual convex programs.
\begin{align} 
    \tag{\small P-QEG}
    \label{eq:pop_qeg}
    \sup \; & \sumiton \left(b_{i} \log \mu_{i}-\delta_{i}\right) 
    \\
    \text { s.t. } & \mu_{i} \leq\left\langle v_{i}, x_{i}\right\rangle+\delta_{i}, \forall i \in[n] 
    \notag
    \\
    & \sumiton x_{i} \leq {s} 
    \notag
    \\
    & \mu_{i} \geq 0, \delta_{i} \geq 0, x_{i} \in L_{1}(\Theta)_{+}, \forall i \in[n]
    \notag
\end{align}
\begin{align}
     \inf \; & \langle p, {1}\rangle- \sumiton b_{i} \log \beta_{i} 
    \tag{\small P-DQEG}
    \label{eq:pop_qdeg}
    \\
     \text { s.t. } & p \geq \beta_{i} v_{i}, \beta_{i} \leq 1, \forall i \in[n] 
    \notag
    \\
    &  p \in L_{1}(\Theta)_{+}, \beta \in \mathbb{R}_{+}^{d}
    \notag
\end{align}

\begin{fact}[Theorem~10 from~\citet{gao2022infinite} and Appendix~C in~\citet{gao2021online}]
    \label{fact:qeg}
 The following holds.

    1. First-order conditions. For any feasible solutions $(x^*_{\QEG},\mu^*, \delta^*)$ to \cref{eq:pop_qeg} and a feasible solution $(p^*_\QEG, \beta^*)$ to \cref{eq:pop_qdeg}. They are optimal both to the respective convex programs if and only if 
\begin{align*}
    & p^{*}_\QEG=\max _{i} \beta_{i}^{*} v_{i}
    \\
    & \Big\langle p^{*}_\QEG, {1}-\textstyle{\sum_{i}} x_{\QEG,i}^{*}\Big \rangle=0 
    \\
    &
    \mu_{i}^{*} = \frac{b_{i}}{\beta_{i}^{*}}, \forall i 
    \\
    &
    \delta_{i}^{*}\left(1-\beta_{i}^{*}\right)=0, \forall i 
    \\
    & \left\langle p^{*}_\QEG-\beta_{ i}^{*} v_{i}, x_{i}^{*}\right\rangle=0, \forall i
\end{align*}

2. Equilibrium. A pair of allocations and prices $(\xst, \pst)$ is a QME if and only if there exists a $\delta^*$ and $\betast$ such that $(\xst,\delta^*)$ and $(\pst, \betast)$ are optimal solutions to  \cref{eq:pop_qeg} and \cref{eq:pop_qdeg}, respectively.

3. Bounds on $\betast$. It holds $\frac{b_i}{\nu_i+ b_i} \leq \betasti \leq 1$ for all $i$, where $\nui = \int v_i(\theta)\diff S(\theta)$.      
\end{fact} 

Note the variable $\mu$ in \cref{eq:pop_qeg} doest not correspond to the equilibrium utility of buyer $i$ at optimality. The equilibrium utility of buyer $i$ is $\lg v_i - \pst, x^*_i \rg$. By the discussion in Section~6 of~\citet{gao2022infinite}, if $\betasti < 1$, then $\lg v_i-\pst, \xsti\rg = (1-\betasti) \mu^*_i$.  If $\betasti = 1$,  then $\lg v_i-\pst, \xsti\rg = 0$. Moreover, $\delta^*_i$ represents the leftover budget in equilibrium~\citep{conitzer2022pacing}.

\begin{proof}[Proof of \cref{thm:rev_convergence}]

    Given the above equivalence results, we use $(\pst, \xst)$ to denote both the equilibrium prices and allocations, as well as the optimal $p$ and $x$ variables in the quasilinear EG programs.
 
    Now the study of the convergence of revenue is reduced to 
    that of the convergence behavior of convex programs 
    \begin{align*}
        \min_{0 < \beta \leq 1_n} H_t (\beta ) \quad ``\!\! \implies \!\!" \quad  \min_{0<\beta \leq 1_n} H (\beta)
    \end{align*}
    Note that in contrast to the EG programs for linear utilities (\cref{eq:pop_deg} and \cref{eq:sample_deg}), we now have an upper bound on the variables $\beta$. 
        
    By repeating the proof of \cref{thm:consistency} we obtain $\betagam \toas \betast$. To show almost sure convergence of revenue, we note 
    \begin{align*}
        & \Big|\frac1t \sumtau \ptau -  \int_\Theta p^*(\theta) s(\theta) \diff \mu(\theta)\Big|
        \\
        & \leq \frac1t \sumtau |\max_i \{ \vithetau \beta^\gam_i\} - \max_i \{ \vithetau \beta^*_i \} | + 
        \Big| \frac1t \sumtau  \max_i \{ \vithetau \beta^*_i \}  -  \int_\Theta p^*(\theta) s(\theta) \diff \mu(\theta)\Big|
        \\
        & \leq  \vbar \| \beta^\gam - \betast\|_\infty +    
        \Big|\frac1t \sumtau  \max_i \{ \vithetau \beta^*_i \}  -  \int_\Theta p^*(\theta)\diff S (\theta )\Big|
        \toas 0 \;.
    \end{align*}
    Here the first term converges to zero a.s.\ by $\beta^\gam \toas \betast$, and the second term converges to $0$ a.s.\ by strong law of large numbers and noting $\E[\max_i \{v_i(\theta)\beta^*_i \}] = \E[p^*(\theta)]$.
    This proves the first part of the statement.

    Define $\ubarbetaQi \defeq  \frac{b_i}{\nui + b_i}$ and $\betabarQ \defeq 1$. We know $\ubarbetaQi \leq \betasti \leq \betabarQ$ from \cref{fact:qeg}. 
    Here we use subscript $Q$ to denote quantities related to the quasilinear market. Define the set
    \begin{align*}
        C_Q \defeq \prod_{i=1}^n \bigg[\frac{b_i}{2 \nu_i + b_i}, 1 \bigg]
        \;.
    \end{align*}
    Clearly we have $\betast \in C_Q$. Furthermore, for $t$ large enough $\betagam \in C_Q$ with high probability.
    To see this, if $t$ satisfies $t \geq 2 (\vbar / \min_i \nui) \sq \log(2n/\eta)$, then $\frac1t \sumtau \vithetau \leq 2 \E[v_i(\theta)]$ for all $i$ with probability $\geq 1-\eta$. By a bound on $\betagam$ in the QME
    \begin{align*}
        \betagami \geq \frac{b_i}{b_i +  \frac1t \sumtau \vithetau},
    \end{align*}
    (see Section 6 in \citet{gao2022infinite}), we obtain $\betagami \geq \frac{b_i}{b_i + 2\nu_i}$ (recall $\nu_i = \E[v_i(\theta)]$).
 
    To obtain the convergence rate, we simply repeat the proof of \cref{thm:high_prob_containment}. Let $L_Q$ and $\lambda_Q$ be the Lipschitz constant and strong convexity constant of $H$ and $H_t$ on $C_Q$.
    We obtain from \cref{eq:final_t_bound} that with probability $\geq 1-2\alpha$,
    there exists a constant $c'$ such that
    as long as 
    \begin{align}
        \label{eq:t_bound_in_quasilinear_case}
        t & \geq c'\cdot  L_Q \sq  \min \bigg\{ \frac{1}{ \lambda_Q \epsilon} , \frac{1}{\epsilon\sq} \bigg\} \cdot \bigg(n\log\Big(\frac{16L_Q}{\eps - \delta}\Big) + \log \frac1\alpha\bigg)
        \;,
    \end{align}
    it holds $|H(\betagam)- H(\betast)| < \epsilon$ and that $\betagam \in C_Q$ (see \cref{cor:H_concentration}).

    Next we calculate $L_Q$ and $\lambda_Q$. Note on $C_Q$, the minimum eigenvalue of $\nabla\sq\Psi  (\beta )= \Diag\{ \frac{b_i}{(\betai)\sq} \}$ can be lower bounded by $\ubar{b}$. So we conclude $\lambda_Q = \ubar{b}$.
    And the Lipschitzness constant can be seen by the following. For $\beta,\beta' \in C_Q$,
    \begin{align*}
        & |H_t(\beta)  - H_t(\beta')| 
        \\
        & \leq \frac1t \sumtau \big|\max_i \{\vithetau \beta_i \} - \max_i \{\vithetau \beta_i' \}\big| + 
        \sumiton b_i \big| \log \beta_i -  \log \beta_i'\big|
        \\
        & \leq \vbar \| \beta - \beta'\|_\infty + \sumiton b_i \cdot \frac{1}{b_i / (2\nui + b_i)} |\beta_i - \beta_i'|
        \\
        & \leq \big(\vbar + 2\nubar n + 1\big) \| \beta - \beta'\|_\infty
        \;.
    \end{align*}
    Similar argument shows that $H$ is also $(\vbar + 2\nubar n + 1)$-Lipschitz on $C_Q$. 
    We conclude $L_Q =(\vbar + 2\nubar n + 1)$.
    Now \cref{eq:t_bound_in_quasilinear_case} shows that for $ t  = \Omega(\ubar{b}\inv) $ 
    (so that the $1/(\lambda_Q \epsilon)$ term in the min becomes dominant) we have 
    \begin{align*}
        |H(\betagam) - H(\betast)| = 
        \tilde{O}_p
        \bigg(
            \frac{n\big(\vbar + 2\nubar n + 1\big)\sq  }{\ubar{b} t}
        \bigg)
        \;,
    \end{align*}
    where we use $\tilde{O}_p$ to ignore logarithmic factors of $t$. Moreover, 
    \begin{align*}
        \| \betagam - \betast\|_\infty \leq  \sqrt{2|H(\betagam) - H(\betast)|  / \lambda_Q}
        = \tilde O_p\bigg(\frac{\sqrt{n} \big(\vbar + 2\nubar n + 1\big)}{\ubar{b} \sqrt t}\bigg)
        \;.
    \end{align*}
    From here we obtain
    \begin{align*}
        & |\REV^\gam - \REV\st| 
        \\
        & \leq \vbar \| \betagam - \betast\|_\infty +\Big|\tfrac1t \sumtau  \max_i \{ \vithetau \beta^*_i \}  -  \int_\Theta p^*(\theta)\diff S (\theta )\Big|
        \\
        & = \tilde O_p\bigg(\frac{\vbar \sqrt{n} \big(\vbar + 2\nubar n + 1\big)}{\ubar{b} \sqrt t}\bigg) + O_p \bigg(\frac{\vbar}{\sqrt{t}}\bigg)
        \\
        & = \tilde O_p\bigg(\frac{ \vbar \sqrt{n} \big(\vbar + 2\nubar n + 1\big)}{\ubar{b} \sqrt t}\bigg)
        \;.
    \end{align*}
    We conclude
    $ |\REV^\gam - \REV\st| = 
   \tilde{O}_p
   \Big(\frac{\vbar \sqrt{n} (\vbar + 2\nubar n + 1 ) }{ \ubar{b} \sqrt{t}}\Big)
   $. This completes the proof of \cref{thm:rev_convergence}.
\end{proof}

\section{Experiments}
\label{sec:experiments}

We conduct experiments to validate the theoretical findings, namely, the convergence of $\LNSW^\gamma$ to $\LNSW^*$ (\cref{thm:consistency}) and CLT (\cref{it:thm:normality:1}).

\textbf{Verify convergence of NSW to its infinite-dimensional counterpart in a linear Fisher market.}
First, we generate an infinite-dimensional market $\mathcal{M}_1$ of $n=50$ buyers each having a linear valuation $v_i(\theta) = a_i \theta + c_i$ on $\Theta=[0,1]$, with randomly generated $a_i$ and $c_i$ such that $v_i(\theta)\geq 0$ on $[0,1]$. 
Their budgets $b_i$ are also randomly generated. 
We solve for $\LNSW^*$ using the tractable convex conic formulation described in \citet[Section 4]{gao2022infinite}. 
Then, following \cref{sec:data}, for the $j$-th ($j\in[k$]) sampled market of size $t$, we randomly sample $\{ \theta^{t,\tau}_j\}_{\tau \in [t]}$ uniformly and independently from $[0,1]$ and obtain markets with $n$ buyers and $t$ items, with individual valuations $v_{i}(\theta^{t,\tau}_j) = a_i \theta^{t,\tau}_j + c_i $, $j\in [t]$. 
We take $t = 100, 200, \dots, 5000$ and $k=10$.
We compute their equilibrium Nash social welfare, i.e., $\LNSW^\gamma$, and their means and standard errors over $k$ repeats across all $t$. 
As can be seen from \cref{fig:linear-one-dim-plot}, $\LNSW^\gamma$ values quickly approach $\LNSW^*$, which align with the a.s.\ convergence of $\LNSW^\gamma$ in \cref{thm:consistency}. 
Moreover, $\LNSW^\gamma$ values increase as $t$ increase, which align with the monotonicity observation in the beginning of \cref{sec:inference}.

\begin{figure}
    \centering
    \includegraphics[scale=.6]{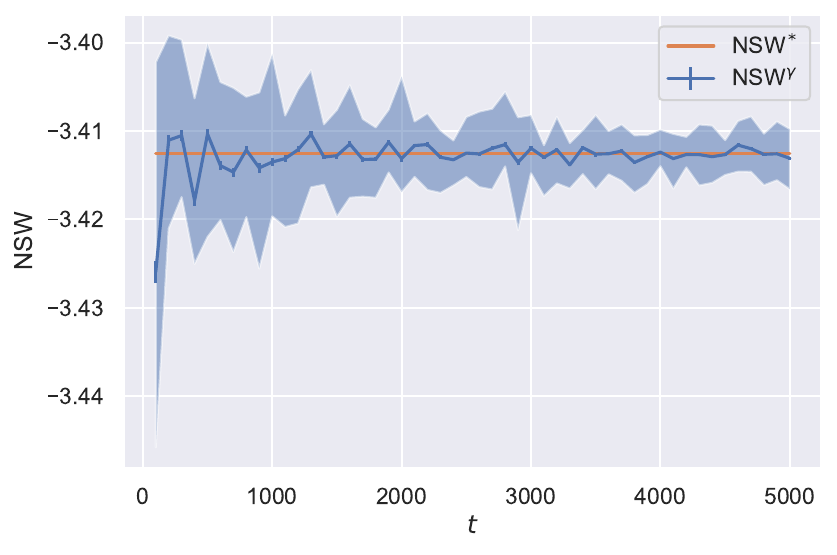}
    \caption{\small
        Mean and standard errors of $\LNSW^\gamma$ of observed markets of sizes $t=100, 200, \dots, 5000$ ($k=10$ repeats) sampled from the infinite-dimensional market $\mathcal{M}_1$ with linear valuations $v_i(\theta) = a_i\theta+c_i$.}
    \label{fig:linear-one-dim-plot}
\end{figure}

\textbf{Verify asymptotic normality of NSW in a linear Fisher market.}
Next, for the same infinite-dimensional market $\mathcal{M}_1$, we set $t=5000$, sample $k=50$ markets of $t$ items analogously, and compute their respective $\LNSW^\gamma$ values. We plot the enpirical distribution of $\sqrt{t}({\LNSW^\gamma - \LNSW^*})$ and the probability density of $N(0, \sigma^2_{\NSW})$, where $\sigma^2_{\NSW}$ is defined in \cref{thm:normality}.\footnote{To compute $\sigma^2_{\NSW}$, we use the fact that  $p^* = \max_i \beta^*_i v_i(\theta) $ is a piecewise linear function, since $v_i$ are linear. Following \citep[Section 4]{gao2022infinite}, we can find the breakpoints of the pure equilibrium allocation $0=a_0 < a_1 < \dots < a_{50} = 1$, and the corresponding interval of each buyer $i$. Then, $\int_0^1 (p^*(\theta))^2 dS(\theta)$ amounts to integrals of quadratic functions on intervals.}
\cref{thm:normality} shows that $\sqrt{t}({\LNSW^\gamma - \LNSW^*}) \tod N(0, \sigma^2_\NSW)$.
As can be seen in \cref{fig:clt-one-dim}, the empirical distribution is close to the limiting normal distribution.
A simple Kolmogorov-Smirnov test shows that the empirical distribution appears normal, that is, the alternative hypothesis of it not being a normal distribution is not statistically significant.
This is further corroborated by the Q-Q plot in \cref{fig:clt-qq}, as the plots of the quantiles of $\sqrt{t}(\LNSW^\gamma - \LNSW^*)$ values against theoretical quantiles of $N(0, \sigma^2_\NSW)$ appear to be a straight line. 

\begin{figure}
    \centering
    \includegraphics[scale=.6]{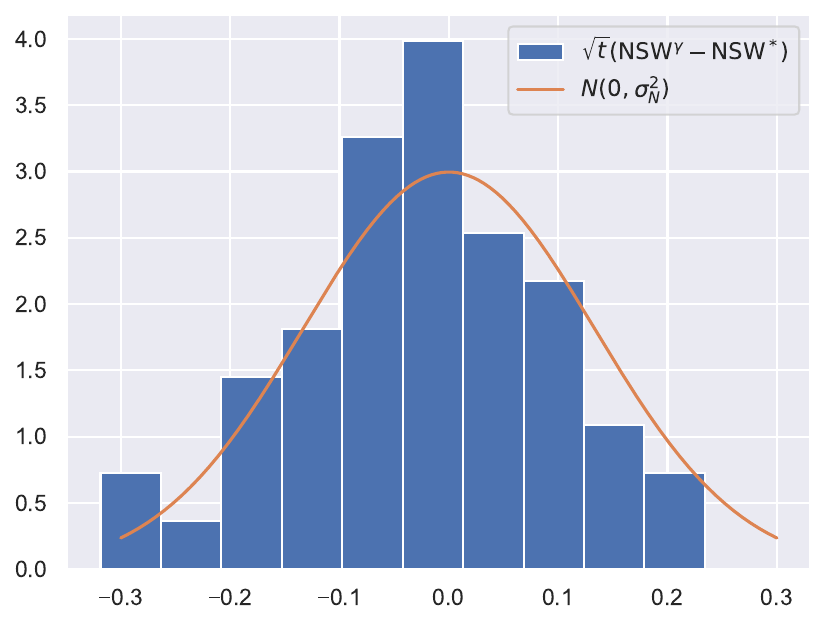}
    \caption{\small
        Empirical distribution of $\sqrt{t}(\LNSW^\gamma - \LNSW^*)$ and $N(0, \sigma^2_\NSW)$. Kolmogorov-Smirnov test null hypothesis: $\sqrt{t}(\LNSW^\gamma - \LNSW^*)$ values are sampled i.i.d. from $N(0, \sigma^2_\NSW)$; alternative hypothesis: they are not sampled i.i.d. from $N(0, \sigma^2_\NSW)$; test statistic: $0.1256$; $p$-value: $0.3779$. }
    \label{fig:clt-one-dim}
\end{figure}

\begin{figure}
    \centering
    \includegraphics[scale=.6]{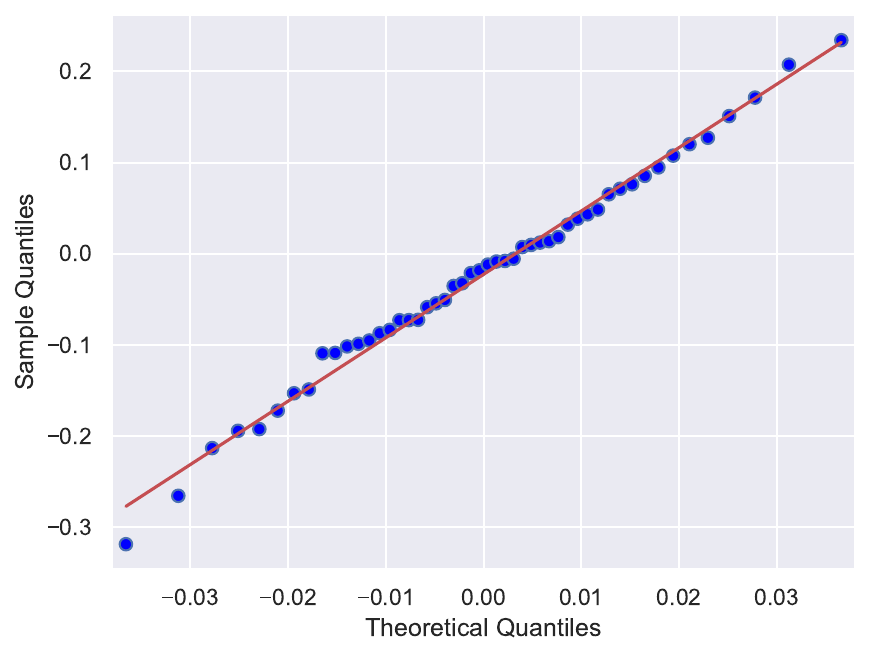}
    \caption{\small
        Q-Q Plot of $\sqrt{t}(\LNSW^\gamma - \LNSW^*)$ values against theoretical quantiles of $N(0, \sigma^2_\NSW)$; a (near) straight line indicates that $\sqrt{t}(\LNSW^\gamma - \LNSW^*)$ values appear to be normal. }
    \label{fig:clt-qq}
\end{figure}

\textbf{Verify NSW convergence in a multidimensional linear Fisher market.}
Finally, we consider an infinite-dimensional market $\mathcal{M}_2$ with multidimensional linear valuations $v_i(\theta) = a_i^\top \theta + c_i $, $a_i\in \R^{10}$. 
We similarly sample markets of sizes $t=100, 200, ..., 5000$ from $\mathcal{M}_2$, where the items $\theta^t_j$, $j\in [k]$ are sampled uniformly and independently from $[0,1]^{10}$. 
As can be seen from \cref{fig:linear-one-dim-plot}, $\LNSW^\gamma$ values increase and converge to a fixed value around $-1.995$. In this case, the underlying true value $\LNSW^*$ (which should be around $-1.995$) cannot be captured by a tractable optimization formulation.

\begin{figure}
    \centering
    \includegraphics[scale=.6]{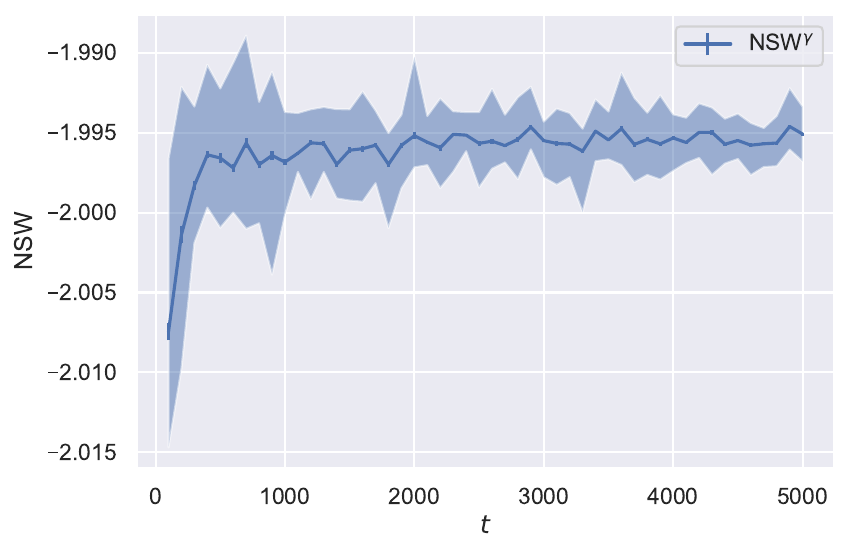}
    \caption{\small
        Mean and standard errors of $\LNSW^\gamma$ of observed markets of sizes $t=100, 200, \dots, 5000$ ($k=10$ repeats) sampled from the infinite-dimensional market $\mathcal{M}_2$ with linear valuations $v_i(\theta) = a_i^\top \theta+c_i$, $a_i\in \R^{10}$. }
    \label{fig:clt-one-dim}
\end{figure}

\end{document}